%% file: arxiv_v1.tex
\PassOptionsToPackage{dvipsnames}{xcolor}
\documentclass[a4paper,pdfa,UKenglish,cleveref,autoref,thm-restate]{lipics-v2021}

\nolinenumbers 
\hideLIPIcs

\usepackage{amsmath,amssymb,amsthm}
\usepackage{graphicx} 
\usepackage{tikz}
\usepackage{ifthen}
\usepackage{xspace}
\usepackage{complexity}
\usepackage{hyperref}
\usepackage{subcaption}
\usepackage{mathtools}
\usepackage{comment}
\usepackage{booktabs}
\usepackage{algorithm}
\usepackage[noEnd=true]{algpseudocodex}
\usepackage{placeins}

\usetikzlibrary{arrows.meta,calc}

\newcommand{\mcG}[1][]{\ensuremath{\mathcal G#1}\xspace}
\newcommand{\mcH}[1][]{\ensuremath{\mathcal H#1}\xspace}

\newcommand{\tw}{\ensuremath{\operatorname{tw}}}
\newcommand{\pw}{\ensuremath{\operatorname{pw}}}

\newcommand*\crossingType[4]{%
  \begin{tikzpicture}[scale=0.08, every path/.style={draw, line cap=round}]
    \coordinate (a) at (-1,-1);
    \coordinate (b) at ( 1,-1);
    \coordinate (c) at ( 1, 1);
    \coordinate (d) at (-1, 1);
    \path (c) -- (a);
    \path (b) -- (d);
    \ifthenelse{\equal{#1}{1}}{\path (a) -- (b);}{}
    \ifthenelse{\equal{#2}{1}}{\path (b) -- (c);}{}
    \ifthenelse{\equal{#3}{1}}{\path (c) -- (d);}{}
    \ifthenelse{\equal{#4}{1}}{\path (d) -- (a);}{}
  \end{tikzpicture}%
}
\newcommand*\fullCt{\crossingType{1}{1}{1}{1}\xspace}
\newcommand*\almostFullCt{\crossingType{0}{1}{1}{1}\xspace}
\newcommand*\bowtieCt{\crossingType{0}{1}{0}{1}\xspace}
\newcommand*\arrowCt{\crossingType{1}{1}{0}{0}\xspace}
\newcommand*\chairCt{\crossingType{0}{0}{1}{0}\xspace}
\newcommand*\XCt{\crossingType{0}{0}{0}{0}\xspace}

\newcommand*\srestricted[1][]{\(\mathcal{S}\)-restricted 1-planar#1\xspace}

\newcommand*\geosrestricted[1][]{geometric \(\mathcal{S}\)-restricted 1-planar#1\xspace}
\newcommand*\geosprimerestricted[1][]{geometric \(\mathcal{S}'\)-restricted 1-planar#1\xspace}

\newcommand*\KuraSub[1][]{Kuratowski subdivision#1\xspace}

\newcommand*\role[1]{\ensuremath{\mathtt{#1}}\xspace}
\newcommand*\mso[2]{\textsc{#1}\ensuremath{#2}}

\newtheorem{obs}[theorem]{Observation}

\newcommand{\threePartition}{\textsc{3-Partition}\xspace}

\newcommand{\fence}[1][u,v]{\ensuremath{F_{#1}}\xspace}

\newcommand*{\tpWidth}{pathwidth\xspace}
\newcommand*{\ChordDescription}{\ensuremath{\Psi}\xspace}
\newcommand*{\sk}[1]{\ensuremath{\mathrm{sk}(#1)}\xspace}
\newcommand*{\skPlus}[1]{\ensuremath{\mathrm{sk}^+(#1)}\xspace}

\newcommand*{\true}{\texttt{true}}
\newcommand*{\false}{\texttt{false}}

\crefname{req}{Condition}{Condition}
\crefname{lemma}{Lemma}{Lemmata}

\input{tikzGenerated.tex}

\title{A Dichotomy for 1-Planarity with Restricted Crossing Types Parameterized by Treewidth}
\titlerunning{1-Planarity with Restricted Crossing Types}

\author{Sergio Cabello}{Faculty of Mathematics and Physics, University of Ljubljana, Slovenia \and Institute~of~Mathematics, Physics and Mechanics, Slovenia}{sergio.cabello@fmf.uni-lj.si}{https://orcid.org/0000-0002-3183-4126}{Funded in part by the Slovenian Research and Innovation Agency (P1-0297, N1-0218, N1-0285). Funded in part by the European Union (ERC, KARST, project number 101071836). Views and opinions expressed are however those of the authors only and do not necessarily reflect those of the European Union or the European Research Council. Neither the European Union nor the granting authority can be held responsible for them.}
\author{Alexander Dobler}{TU Wien, Austria}{adobler@ac.tuwien.ac.at}{https://orcid.org/0000-0002-0712-9726}{Vienna Science and Technology Fund (WWTF)  grant [10.47379/ICT19035]}
\author{Ga\v{s}per Fijav\v{z}}{Faculty of Computer and Information Science, University of Ljubljana, Slovenia}{gasper.fijavz@fri.uni-lj.si}{https://orcid.org/0000-0003-1276-2494}{}
\author{Thekla Hamm}{Department of Mathematics and Computer Science, TU Eindhoven, the Netherlands}{t.l.s.hamm@tue.nl}{https://orcid.org/0000-0002-4595-9982}{}
\author{Mirko H. Wagner}{Theoretical Computer Science, Osnabrück University, Germany}{mirko.wagner@uos.de}{https://orcid.org/0000-0003-4593-8740}{}

\authorrunning{S.\ Cabello \and A.\ Dobler \and G.\ Fijav\v{z} \and T.\ Hamm \and M.H.\ Wagner}

\Copyright{Sergio Cabello, Alexander Dobler, Ga\v{s}per Fijav\v{z}, Thekla Hamm, and Mirko H. Wagner} 



\acknowledgements{This work was started during the Crossing Number Workshop 2024 that took place in Montserrat, Spain. We thank the organizers
and participants for providing a fruitful research environment.}

\begin{document}

\maketitle  
\begin{abstract}
A drawing of a graph is 1-planar if each edge participates 
in at most one crossing and adjacent edges do not cross.
Up to symmetry, each crossing in a 1-planar drawing
belongs to one out of six possible crossing types,
where a type characterizes the subgraph induced
by the four vertices of the crossing edges.
Each of the \( 63 \) possible nonempty subsets \(\mathcal{S}\) 
of crossing types gives a recognition problem: does a given graph 
admit an \(\mathcal{S}\)-restricted drawing, that is,
a 1-planar drawing where the crossing type of each crossing 
is in \(\mathcal{S}\)?

We show that there is a set \(\mathcal{S}_{\rm bad}\) with three 
crossing types and the following properties:
\begin{itemize}
\item If \(\mathcal{S}\) contains no crossing type 
	from \(\mathcal{S}_{\rm bad}\), then the recognition of
	graphs that admit an \(\mathcal{S}\)-restricted drawing
	is fixed-parameter tractable with respect to the treewidth of the input graph.
\item If \(\mathcal{S}\) contains any crossing type
	from \(\mathcal{S}_{\rm bad}\), then 
	it is NP-hard to decide whether 
	a graph has an \(\mathcal{S}\)-restricted drawing,
	even when considering graphs of constant pathwidth.
\end{itemize}
We also extend this characterization of crossing types to 1-planar
straight-line drawings and show the same complexity behaviour parameterized by treewidth.
\end{abstract}
\section{Introduction}

Drawings of graphs are encountered early and used often:
children are challenged to draw \(K_{3,3}\) without crossings,
drawings of graphs are regularly used when teaching graphs, they appear
in any graph-related textbook, and we use them in our research discussions. 
Drawings of graphs naturally led to the concept of planar graphs and,
more generally, an interest to control the crossings in the drawings. 

Planar graphs enjoy a rich structural theory, have spurred research 
in graph theory and graph algorithms, and are relatively well understood. 
However, they are a quite restrictive class of graphs, and non-planar
graphs still have to be drawn. For this reason, a range of beyond-planarity 
concepts have been suggested and investigated~\cite{DidimoLM19,HT2020,Zehavi22}. 
In this work, we will focus on 1-planarity, one of the extensions 
of planarity that has attracted much interest~\cite{KobourovLM17}.
A drawing of a graph is \emph{1-planar} (or has \emph{local crossing number} \(1\))
if each edge participates in at most one crossing and adjacent edges 
do not cross\footnote{The condition of non-crossing adjacent edges
is usually added because whenever the edges~$uv$ and~$uv'$ cross in a drawing, 
we can get a drawing with fewer crossings by redrawing the start of
the edges $uv,uv'$.}; see \Cref{fig:1-planar}, left, for an example.  
A graph is 1-planar iff it admits a 1-planar drawing. 

Generally speaking, the class of 1-planar graphs is not 
well understood.
We know that recognizing $1$-planar graphs is computationally 
hard~\cite{GrigorievB07,KorzhikM13}, even for graphs that are obtained
from a planar graph by adding a single edge~\cite{CabelloM13},
or for graphs with constant treewidth, pathwidth or even 
bandwidth~\cite{BannisterCE18}. A fixed-parameter algorithm for 
deciding the existence of 1-planar drawings has been obtained recently 
by Hamm and Hlinen{\'{y}}~\cite{HammH22} (and actually is an easy 
consequence of \cite{Grohe04}) using the total number 
of crossings in the drawing as a parameter.

Biedl and Murali~\cite{maininspoBiedlMurali23} noticed that
the crossings in a 1-planar drawing can be classified into different
types, and the (non-)existence of some types may have important
consequences.
We will call these different types \emph{crossing types},
and they describe adjacencies 
between the endpoints of the edges involved in that crossing.
Consequently, there are six different crossing types (up to symmetry), 
for which we follow the terminology from~\cite{maininspoBiedlMurali23}:
(1) A \emph{full} crossing is one in which each endpoint of an edge involved 
in that crossing is connected to each other such endpoint (\(\fullCt\)),
(2) an \emph{almost full} crossing is one in which all but one pair of endpoints 
of the edges involved in that crossing are connected to each other 
(\(\almostFullCt)\),
(3) a \emph{bowtie} crossing is one in which the graph induced by 
the endpoints of the edges involved in that crossing is a \(4\)-cycle (\(\bowtieCt\)),
(4) an \emph{arrow} crossing is one in which there is precisely one endpoint of 
an edge involved in that crossing that is connected to all other such endpoints 
(\(\arrowCt\)),
(5) a \emph{chair} crossing is one in which all but one pair of endpoints 
of the edges involved in that crossing are independent if one were to remove 
the edges involved in the crossing (\(\chairCt\))
(6) a \emph{\(\times\)} crossing is one in which all endpoints of the edges 
involved in that crossing are independent if one were to remove the edges 
involved in the crossing (\(\XCt\)).
See \Cref{fig:1-planar}, left, for an example.

\begin{figure}
    \centering
    \includegraphics[page=1]{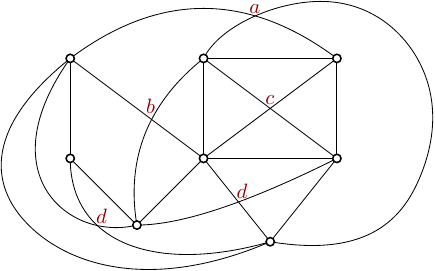}
    \caption{A 1-planar drawing of a graph where 
		the crossing marked with \(a\) is of type \(\bowtieCt\) (bowtie),
		the crossing marked with \(b\) is of type \(\arrowCt\) (arrow),
		the crossing marked with \(c\) is of type \(\fullCt\) (full), and
		the two crossings marked with \(d\) are of type \(\almostFullCt\) (almost full).}
    \label{fig:1-planar}
\end{figure}

For each non-empty \(\mathcal{S}\subseteq \{\fullCt, \almostFullCt, 
\bowtieCt, \arrowCt,\chairCt,\XCt \}\), we say that a drawing is \srestricted 
if it is 1-planar and the type of each crossing belongs to \(\mathcal{S}\), 
and a graph is \srestricted if it admits an \srestricted drawing.  

We already know that limiting the crossing types can have consequences.
For instance, Biedl and Murali~\cite{maininspoBiedlMurali23} described
an algorithm to compute the vertex connectivity of a graph that
is given with a 1-planar drawing without \(\XCt\) crossings.
This was extended by Biedl, Bose and Murali~\cite{BiedlB024}
to allow some \(\XCt\) crossings in a controlled manner.
Bose et al.~\cite{BoseCMM25} have shown that graphs that
have a 1-planar \(\XCt\)-crossing-free drawing have 
bounded cop-number. However, it was not clear how difficult it is
to decide if a given graph has such a 1-planar drawing without 
\(\XCt\) crossings, and this is posed explicitly 
as an open problem in~\cite{maininspoBiedlMurali23}.

On the other end of the spectrum when it comes to crossing types, 
Brandenburg~\cite{Brandenburg15,Brandenburg19} has shown that the class of 
1-planar graphs that admit 1-planar drawings where \emph{all} crossings 
are \(\fullCt\)-crossings is equivalent to the class 
of \(4\)-map graphs with holes, and thus \(\{\fullCt\}\)-restricted 1-planar
graphs are recognizable in polynomial time.
M{\"{u}}nch and Rutter~\cite{MunchR24} show that, for each subset \(\mathcal{S}\)
of crossing types, recognizing \srestricted is \FPT\ with respect to the 
\emph{number} of crossings.

Another important thread in our work is that of straight-line drawings,
which we refer to as \emph{geometric drawings}.
It is well-known that planar graphs admit a geometric planar drawing~\cite{Fary,Wagner}.
However, we know that whether in the drawing we use arbitrary
curves for the edges or we use straight-line segments does affect
the crossings~\cite{BienstockD93}.
Similar phenomena occur for $1$-planar graphs: not all
graphs that have a $1$-planar drawing have a geometric $1$-planar drawing. 
Thomassen~\cite{Thomassen88a} provided a criterion to know when
a $1$-planar drawing can be continuously deformed into
an equivalent $1$-planar drawing with straight-line edges.
The characterization is based on forbidden configurations in the drawing;
we will explain this in detail below.
How to transform efficiently a 1-planar drawing into a geometric 
one has been considered by Hong et al.~\cite{HongELP12}.

\subparagraph{Our contribution.}
We systematically investigate the complexity of recognizing 
\srestricted graphs for every fixed, nonempty subset \(\mathcal{S}\) of crossing types.
Fortunately, the characterization is neat, as it does not need
to go over the \(63\) possible cases, and shows an interesting
dichotomy when considering graphs parameterized by the treewidth:
\begin{itemize}
\item If \(\mathcal{S} \subseteq \{\fullCt,\almostFullCt,\bowtieCt\}\),
  it is \FPT\ parameterized by the treewidth to recognize \(\mathcal{S}\)-restricted 
  1-planar graphs.
\item If \(\mathcal{S} \cap \{\arrowCt,\chairCt,\XCt\} \neq \emptyset\), then it is 
  \NP-hard to recognize 1-planar graphs that are \srestricted, even for graphs
  that have the treewidth bounded by constant.
\end{itemize}
Thereby we do not only resolve the open question from 
\cite{maininspoBiedlMurali23} (negatively), but we also obtain 
a very detailed perspective on how hard is recognizing 
1-planar graphs with constant treewidth that
admit 1-planar drawings with certain crossing types.
The key difference between a group of crossing types and the other group 
is the existence of a \(4\)-cycle on the endpoints of the crossing. 
Of course, the role of such cycles is not clear a priori.

We also show that our dichotomy extends to the 
\emph{geometric} setting, i.e.\ when all edges are required to 
be drawn as straight-line segments.
Surprisingly, the same characterization works. However,
extending the result to the geometric setting requires substantial effort.

\subparagraph{Techniques.}
To obtain our positive result, we first use a sequence of reductions 
to show that it suffices to solve the problem for instances with 
enough connectivity. Here, enough connectivity means that
the graph is internally \(3\)-connected: it can be obtained 
from a \(3\)-connected graph by subdividing some edges once.
This step is quite standard but tedious, especially in the case of 
geometric drawings, where we have to slightly generalize the 
problem we consider so that a prescribed vertex has to lie on the 
outer face of the target 1-planar drawing.
The reductions do use that the crossing types are a subset 
of~\(\{\fullCt,\almostFullCt,\bowtieCt\}\).

To solve the instances with enough connectivity, we use the fact 
that we target treewidth parameterizations 
to be able to invoke Courcelle's theorem~\cite{MR1042649}.
However, expressing 1-planarity is something that is not 
possible efficiently in MSO (the logic required to formulate 
a problem in so that Courcelle's theorem implies a fixed-parameter 
algorithm), unless \(\FPT = \W[1]\).
This is where restricting crossing types becomes essential.
We are able to show that restricting to crossing types that 
force \(4\)-cycles to be present on the endpoints of the crossing 
edges allows for encoding 1-planarity in MSO in a constant-length formula.
The key insight to do so is to realize an elegant way in which 
we can identify pairs of crossing edges without explicitly coding 
them into an MSO-formula.

To adapt this to the geometric setting, we use the fact that small 
topological obstructions to geometric 1-planarity are 
known~\cite{Thomassen88a}.
As MSO cannot speak about topology, it is useful that we can first 
reduce to instances whose prospective planarizations have unique embeddings.
Here it is again handy to have enough connectivity.

Our \NP-hardness results are all based on a single reduction 
from \threePartition which are modularly adaptable for each crossing 
type that does not contain a \(4\)-cycle.

\subparagraph{Organization.}
The rest of the paper is organized as follows.
We start providing some preliminaries in \Cref{sec:preliminaries}.
The positive \FPT\ results follow from \Cref{sec:3con}, where we
show that it suffices to consider instances with enough connectivity,
and \Cref{sec:fpt}, where we provide an \FPT-algorithm 
parameterized by the treewidth, assuming enough connectivity. 
The NP-hardness reductions are presented in \Cref{sec:hardness}.

\section{Preliminaries}
\label{sec:preliminaries}
We use standard terminology and notation from graph theory~\cite{Diestel-GT-2025}.
We also assume general familiarity with parameterized complexity 
theory~\cite{CyganFKLMPPS15,DowneyF13}.
We explain here concepts and notation that we use; some of them may be 
considered standard by several readers, but others are perhaps less common.

\subparagraph{Graph terminology.}
All graphs we consider are simple and undirected.
We denote an edge with vertices \(u\) and \(v\) as \(uv\). For a set \(X\) 
of vertices of a graph \(G\), the graph induced by \(X\) is \(G[X]\). 
For a vertex \(u\) of a graph \(G\) the graph resulting 
from the removal of \(u\) is denoted by \(G-u\). 
For two vertices \(u\) and \(v\) of a graph \(G\), the graph resulting 
from the removal of the edge \(uv\) is denoted by \(G-uv\); note that
we leave the option that \(uv\) is not present in \(G\), and in such
a case \(G-uv=G\). The graph resulting from the addition 
of \(uv \not \in E(G)\) is denoted by \(G + uv\).
For \(i \in \mathbb{N}\), we denote by \(P_i\), \(C_i\) and  \(K_i\) a path, 
cycle and clique on \(i\) vertices respectively and by \(K_{i,i}\) a complete 
bipartite graph with \(i\) vertices on each side.

For a graph \(G\), a vertex set \(S\) whose removal results in at least 
two connected components is called a \emph{separator}.
If \(|S| = k\), then \(S\) is called a \(k\)-separator, and if \(|S|=1\),
then its only element is called a \emph{cutvertex}.
A separator \(S\) of \(G\) induces a \emph{separation} \((X_1,X_2)\) if
\(X_1\cup X_2 = V(G)\), \(X_1\cap X_2 = S\), 
\( X_1\setminus S\neq \emptyset\), \( X_2\setminus S\neq \emptyset\),
and each path in \(G\) from a vertex in \(X_1\setminus S\) to a vertex 
in \(X_2\setminus S\) contains some vertex of \(S\).
A separation induced by a separator is not necessarily unique.
For a separation \((X_1,X_2)\) induced by a separator \(S\), each
edge \(uv\) of \(G\) belongs to \(G[X_1]\), \(G[X_2]\) or both (when 
\(\{u,v\}\subseteq S\)).

The (\emph{vertex}-)\emph{connectivity} of a graph is the smallest 
number of vertices whose removal results in at least two connected 
components or a single vertex.
The graph \(G\) is \emph{\(k\)-connected} if its connectivity is at least \(k\).
We call \(G\) \emph{internally 3-connected} if \(G\) can be obtained from a 
3-connected graph \(G'\) by subdividing (with single vertices) a 
subset of edges of \(G'\). (Here \(G'\) may have some parallel edges and 
at most one edge in each group of parallel edges is not subdivided.) 
An internally 3-connected graph \(G\) is 2-connected and,
if \(S\) is a 2-separator in \(G\), then \(G-S\) has a single
component with more than one vertex.

If \(G\) is a connected graph, then a \emph{block-tree} or \emph{BC-tree} 
of \(G\) is a tree \(T\) whose vertices are cutvertices and blocks of \(G\), 
and adjacency in \(T\) is defined as containment in \(G\). 

A tree \(T\) is the \emph{SPR-tree} of a \(2\)-connected graph \(G\), 
if every node \(t\) of \(T\) is labeled with S-, P-, or R-label, 
and it represents a graph, called the \emph{skeleton of \(t\)} and 
denoted \sk{t}, whose edges are labeled \emph{real} or \emph{virtual}. 
S-nodes represent cycles, P-nodes represent dipoles, and R-nodes 
represent 3-connected graphs that are not a triangle, which are 
considered S-nodes. Every edge \(xy \in E(T)\) pairs a pair 
of directed virtual edges from \sk{x} and \sk{y}, respectively, 
and every virtual edge \(e\) from \sk{x} is paired with another 
virtual edge. We obtain \(G\) from its SPR-tree representation 
by identifying pairs of paired virtual edges and deleting the 
identified pairs. The graphs represented by R-nodes are called 
\emph{3-connected components} of \(G\). 

The \emph{skeleton+} of a node \(x\) of \(T\), denoted \skPlus{x}, 
is obtained as follows: 
we put into \skPlus{x} all vertices and all edges of \sk{x};
for each virtual edge \(uv\) of \sk{x}, we add \(uv\) to
\skPlus{x} and subdivide it once;
for each virtual edge \(uv\) of \sk{x} that is an actual
edge of \(G\), we add \(uv\) to \skPlus{x}. 
See \Cref{fig:sk+} for an example.
Note that the skeleton+ of an R-node is internally 3-connected.

\begin{figure}
    \centering
    \includegraphics[page=3]{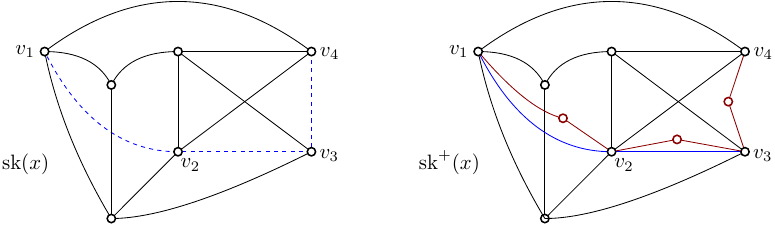}
    \caption{Left: Example of \sk{x} for an R-node \(x\) where the
		virtual edges are dashed blue.
		Right: the corresponding skeleton+ of \(x\), \skPlus{x},
		assuming that \( v_1v_2\) and \(v_2v_3\) are edges of 
		the graph \(G\) but \(v_3v_4\) is not an edge of \(G\).
		The red part comes from subdividing the virtual edges 
		and the blue edges come from the virtual edges of \sk{x} that
		exist in \(G\).
	}
    \label{fig:sk+}
\end{figure}

While BC-trees are a classical concept in graph theory, 
SP(Q)R-trees were first described by 
Di~Battista and Tamassia~\cite{BattistaT89}. 
Both BC-trees and SPR-trees of a graph \(G\) have size \(O(|G|)\) 
and can be computed in linear time~\cite{GutwengerM00}.

\subparagraph{Graph drawing.}
A \emph{drawing \mcG of a graph} \(G\) is an 
injective mapping of \(V(G)\) to \(\mathbb R^2\) and of each edge 
\(e = uv \in E(G)\) to a simple curve \(\mcG(e)\) starting 
in \(\mcG(u)\), ending in \(\mcG(v)\) and not intersecting 
\(\mcG(V(G)) \setminus \{\mcG(u),\mcG(v)\}\).
We call the values of vertices and edges under the drawing of a 
graph the \emph{drawings of these vertices and edges}.
In the following we often identify vertices and edges with 
their respective drawings. We restrict drawings to only allow 
intersections of edges in finite sets and that whenever two edges 
intersect in a set outside of \(\mcG(V(G))\) they do so 
transversally (i.e.\ not tangentially).
We call such intersections \emph{crossings} and say that the 
respective edges \emph{cross} or \emph{are involved} in that crossing.
A drawing is \emph{planar} if it has no crossings and 
a graph is \emph{planar} if it has a planar drawing.

A characterization of planar graphs, which will be useful to us, 
is that a graph is planar, if and only if it does not contain 
a subgraph that is a subdivision of 
\(K_5\) or \(K_{3,3}\)~\cite{Kuratowski30}; such subgraphs 
are also called \emph{\KuraSub[s]}.

\begin{figure}
    \centering
    \includegraphics[page=2]{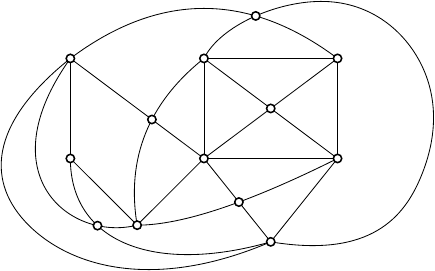}
    \caption{Planarization $\mcG^\times$ of the 
		drawing \mcG of \Cref{fig:1-planar}.}
    \label{fig:planarization}
\end{figure}

The \emph{planarization} \(\mcG^\times\) of a 1-planar 
drawing \mcG of a graph \(G\) is obtained from 
\mcG by replacing each pair of crossing edges 
with a new vertex positioned in a crossing which is linked 
to the endpoints of the crossing edges. Essentially this 
means we replace pairs of crossing edge segments by \(4\)-claws.
See \Cref{fig:planarization}.

A \emph{face} of a drawing is a maximal connected subset 
of~\(\mathbb{R}^2\) that does not intersect the drawing (this 
implies that faces are always open sets, and this definition 
is also for drawings with crossings). Note that the faces
of a 1-planar drawing \mcG are the same
as the faces of its planarization \(\mcG^\times\).
The \emph{outer face} is the unique unbounded face.
For a simple closed curve in a drawing, a vertex is drawn 
\emph{inside} that curve if it is separated from the outer 
face by that curve, and \emph{outside} of that curve otherwise.

We overload the term `\emph{boundary} of a face' to mean 
all the vertices, edges, edge segments, and crossings whose drawings 
lie on the (pointwise set) boundary of a face. We say that 
two vertices are \emph{cofacial} in a drawing if they lie 
together on the boundary of an arbitrary face of the drawing.
It is known and easy to show that for any planar graph \(G\)
and any subset \(U\) of the vertices of \(G\), there is an embedding where
the vertices of \(U\) are in the same face if and only if the
graph with vertex set \(V(G)\cup\{t\}\), where \(t\) is a new
vertex, and edge set \(E(G)\cup \{ut\mid u\in U\}\) is planar.

We will often use the following simple observation about 1-planar
drawings: if the edges \(uv\) and \(u'v'\) cross in a 1-planar 
drawing \mcG, then the vertices \(u\) and \(u'\) are cofacial in \mcG. 
Indeed, we can find a curve in the plane from \(u\) to \(v\)
that closely follows \(uv\) until the crossing with \(u'v'\),
and then closely follows \(u'v'\) until \(u'\). Since 
the drawings of \(uv\) and \(u'v'\) cannot be crossed by
any other edge, because of 1-planarity, the curve
does not cross \mcG at any point.

Two drawings of the same graph are \emph{equivalent} if there 
is a homeomorphism of the plane onto itself that transforms
one drawing into the other.

\subparagraph{Geometric 1-planar drawings.}
A \emph{geometric drawing} of a graph is a drawing where the drawing
of each edge is a straight-line segment. 
Thomassen~\cite{Thomassen88a} showed that a drawing \(\mcG\) 
is equivalent to a geometric 1-planar drawing if and only if 
\(\mcG\) is 1-planar and neither of the 
following is contained in \(\mcG\) (see \Cref{fig:bwConfig} for reference):
\begin{itemize}
    \item Three edges \(ss'\), \(sb\) and \(s'b'\) such that 
		\(sb\) and \(s'b\) cross and \(b\) and \(b'\) are inside 
		the closed curve given by the drawings of \(ss'\), \(s'b'\) 
		from \(s'\) up to its crossing with \(sb\) and \(sb\) from 
		its crossing with \(s'b'\) up to \(s\); 
		such a situation is called a \emph{B-configuration}.
    \item Four edges \(sw_1\), \(sw_2\), \(s'w'_1\) and \(s'w'_2\) 
		such that \(sw_1\) and \(s'w'_1\) cross, \(sw_2\) 
		and \(s'w'_2\) cross, and \(w_1\), \(w_2\), \(w'_1\) and \(w'_2\) 
		lie inside the closed curve given by the drawings of the 
		four edges between their crossings and \(s\) and \(s'\); 
		such a situation is called a \emph{W-configuration}. 
\end{itemize}
\begin{figure}
    \centering
    \includegraphics[page=4]{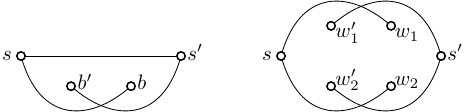}
    \caption{B-configuration (left) and W-configuration (right).}
    \label{fig:bwConfig}
\end{figure}

In both configurations above we call the pair of vertices \(s,s'\) 
the \emph{spine}.
Note that a B- or a W-configuration is a clear obstruction to obtain
an equivalent geometric 1-planar drawing. The other implication is
far from obvious. See \Cref{fig:W-configuration} for an example.

\begin{figure}
    \centering
    \includegraphics[page=5]{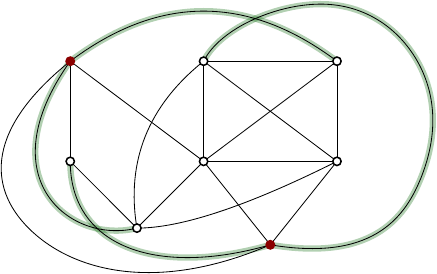}
    \caption{The 1-planar drawing of \Cref{fig:1-planar} has 
		a W-configuration, as highlighted here. The spine is marked 
		with red dots.}
    \label{fig:W-configuration}
\end{figure}

\subparagraph{Parameterized complexity.}
We briefly recall the definition of treewidth and pathwidth.
A \emph{tree decomposition} of a graph \(G\) is a tuple \((T,\chi)\) 
consisting of a tree \(T\) and a mapping 
\(\chi : V(T) \to 2^{V(G)}\) such that, for each \(e \in E(G)\), 
there is some \(t \in V(T)\) such that \(e \subseteq \chi(t)\) 
and, for each \(v \in V(G)\), the induced subgraph 
\(T[\{t \in V(T) \mid v \in \chi(t)\}]\) is connected and non-empty.
A \emph{path decomposition} is defined in the same way, but 
instead of \(T\) being a tree we require that \(T\) is a path.
The \emph{width} of a tree or path decomposition is defined 
as \(\tw(T,\chi) = \max_{t \in V(T)}|\chi(t)| - 1\). 
The \emph{treewidth} (denoted \(\tw(G)\)) and the \emph{pathwidth} 
(denoted \(\pw(G)\)) of \(G\) are defined as the minimum width 
of a tree and path decomposition of \(G\), respectively.

Given a tree decomposition \((T,\chi)\) of a graph \(G\), 
the graph with vertex set \((V(G)\) and edge set 
\(\{uv \mid \exists t \in V(T) \ \{u,v\} \subseteq \chi(t)\}\) 
is called a \emph{chordal completion} of \(G\). By construction, 
\(G\) is a subgraph of any chordal completion and such a chordal 
completion has treewidth at most \(\tw(T,\chi)\). 
Chordal completions are \emph{chordal} graphs, i.e. do not contain 
induced \(4\)-cycles; see for 
example~\cite[Proposition 12.3.11]{Diestel-GT-2025}.

MSO\(_2\) allows quantification over vertex, edge, vertex set 
and edge set variables on expressions built from logical operations, 
containment (\(\in\)), inclusion (\(\subseteq\)) and the 
graph-specific adjacency (\mso{adj}) and incidence (\mso{inc}) 
predicates that describe whether two vertices are adjacent 
to each other and whether a vertex is incident to an edge, respectively.

Courcelle's theorem~\cite{MR1042649,DowneyF13} states that 
any problem encodable by an MSO\(_2\)-sentence \(\varphi\) over 
input graph \(G\) can be decided in \FPT-time parameterized 
by \(\tw(G) + |\varphi|\), where \(|\varphi|\) is the length 
of the sentence \(\varphi\).

\section{Reducing to Internally 3-Connected Instances}
\label{sec:3con}

The first aim of this section is to show that for recognizing
\(\mathcal{S}\)-restricted 1-planar graphs, when 
\(\mathcal{S} \subseteq \{\fullCt,\almostFullCt,\bowtieCt\}\),
it suffices to consider the problem for graphs with enough 
connectivity. More precisely, we will show the following in
\Cref{sec:3con:usual}.

\begin{theorem}
    \label{thm:3con}
    Assume that \(\mathcal{S} \subseteq \{\fullCt,\almostFullCt,\bowtieCt\}\).
	A graph is \srestricted if and only if the skeleton+ of each 
	of its 3-connected components is \srestricted.
\end{theorem}

The second aim is to show a similar result in the geometric setting.
Because the choice of the outer face matters for \geosrestricted drawings, 
we consider a generalization of the problem that takes this 
into account appropriately, as follows.

A graph \(G=(V,E)\) is \emph{\(O\)-\geosrestricted} for 
some \(O\subset V\) with \(|O| \leq 1\), if there is a \geosrestricted 
drawing of \(G\) where the vertex from \(O\), if it exists, 
is on the outer face. The case where \(O=\emptyset\) corresponds
to the original problem, \geosrestricted.
In this geometric setting we will show the following 
in \Cref{sec:3con:geom}.

\begin{theorem}
    \label{thm:3con:geom}
    Assume that \(\mathcal{S} \subseteq \{\fullCt,\almostFullCt,\bowtieCt\}\).
	Given an \FPT-algorithm parameterized by treewidth that decides 
	whether internally \(3\)-connected graphs are \(O\)-\geosrestricted, 
	we can formulate an \FPT-algorithm parameterized by treewidth
	to decide whether general graphs are \geosrestricted.
\end{theorem}

\textbf{Throughout this section, we assume that 
\(\mathcal{S} \subseteq \{\fullCt, \almostFullCt, \bowtieCt\}\) is fixed};
in the statements we do not quantify over \(\mathcal{S}\) anymore. 
As observed earlier, in each \srestricted drawing every pair of edges 
defining a crossing lie on a common \(4\)-cycle. 

\subsection{The usual, topological setting}
\label{sec:3con:usual}

The aim of this section is to prove~\Cref{thm:3con}. 
First, we make the easy observation that we can decide each \(2\)-connected 
component on its own and combine their \srestricted drawings by identifying 
the corresponding cut vertices.

\begin{observation}
\label{obs:2con}
	A graph is \srestricted if and only if all of its \(2\)-connected 
	components are \srestricted.
\end{observation}

For \(3\)-connected components we want to proceed similarly and combine 
\srestricted drawings of \(3\)-connected components by identifying their 
corresponding separators. To this end, we first show that we may restrict 
our attention to certain drawings.

\begin{lemma}
\label{lem:3con}
	A graph \(G\) is \srestricted if 
	and only if there is an \srestricted drawing of \(G\) where no edge
	incident to degree-two vertices is crossed.
\end{lemma}
\begin{proof}
	The backward implication is clear. For the forward implication,
	consider an \srestricted drawing \mcG of \(G\) with the minimum
	number of crossings. Assume, for the sake of reaching a contradiction,
	that in \mcG there is a degree-two vertex \(t\) and an edge 
	incident to \(t\) that is crossed. Let \(s\) and \(s'\) be the 
	two neighbors of \(t\), labeled in such a way that in \mcG the 
	edge \(st\) is crossed. 
	As crossing edges must be on a 4-cycle and \(t\) has degree two, 
	the edge crossing \(st\) must be incident to \(s'\).
	Therefore, in \mcG the vertices \(s\) and \(s'\) are incident 
	to edges that cross each other, and by 1-planarity the vertices
	\(s\) and \(s'\) are cofacial in \mcG. Thus, \(st\) and \(ts'\) 
	can be redrawn without crossings. The resulting drawing is also
	\srestricted and thus \mcG could not be one with the minimum 
	number of crossings.
\end{proof}

We next describe how to combine drawings of parts of 
$2$-connected graphs. See~\cite[Lemma 1]{Brandenburg15} 
for a similar statement in a more restricted setting.

\begin{lemma}
\label{lem:3con:alt}
	Let \(G\) be a 2-connected graph, \(\{s,s'\}\) a 2-separator 
	of \(G\), and \((X_1,X_2)\) a separation induced by \(\{s,s'\}\).
	Then \(G\) is \srestricted if and only if both 
	\(G_1\coloneqq G[X_1]+\{st_1,t_1s'\}\) and 
	\(G_2\coloneqq G[X_2]+\{st_2,t_2s'\}\) are \srestricted, 
	where \(t_1\) and \(t_2\) are new vertices.
	\textup(Note that if \(ss'\) is an edge of \(G\), then it is also
	an edge of \(G_1\) and \(G_2\).\textup)
\end{lemma}
\begin{proof}
    For the forward implication, consider an \srestricted drawing \mcG 
	of \(G\) and its restrictions to \(G[X_1]=G_1-t_1\) and 
	\(G[X_2]=G_2-t_2\), denoted \mcG[_1] and \mcG[_2], respectively.
	Since \(G\) is 2-connected, the graph \(G[X_2]-ss'\)
	is connected. (The edge \(ss'\) is removed only if it exists.)
	This means that there is a curve \(\gamma_2\)
	from \(s\) to \(s'\) that is contained in \(\mcG[_2]-ss'\). 
	
    If there is a crossing in \mcG between an edge from \(G\setminus G_1\) 
	and an edge from \(G\setminus G_2\), that crossing must be between 
	edges of a 4-cycle \(C\) that contains \(s\) and \(s'\) because 
	\((X_1,X_2)\) is a separation induced by \(\{s,s'\}\). Note that
	the edge \(ss'\), if it exists in \(G\), is not an edge of 
	\(G\setminus G_1\) or \(G\setminus G_2\). Therefore the cycle
	\(C\) is of the form \(sv_1s'v_2s\) with vertices \(v_1\in X_1\)
	and \(v_2\in X_2\). In this case, \(s\) and \(s'\) are incident
	to edges that cross in \mcG, which means that they are cofacial in \mcG.
	Thus, \(s\) and \(s'\) also cofacial in \(\mcG[_1]\).
	If there is no crossing in \mcG between an edge from \(G\setminus G_1\) 
	and an edge from \(G\setminus G_2\), then the curve \(\gamma_2\)
	witnesses that \(s\) and \(s'\) are cofacial in \mcG[_1]; it could
	happen that \(\gamma_2\) crosses the drawing of the edge \(ss'\) in \mcG[_1],
	if \(ss'\) exists, but in such a case we can follow \(\gamma_2\)
	from \(s\) until the crossing with \(ss'\) and then follow closely
	the drawing of \(ss'\) until \(s'\). We conclude that, 
	independently of whether there is a crossing between an edge 
	of \(G\setminus G_1\) and an edge from \(G\setminus G_2\),
	the vertices \(s\) and \(s'\) are cofacial in \mcG[_1].
	We can then add \(st_1,t_1s'\) to \mcG[_1] without crossings.
	The same holds for \mcG[_2] because of symmetry.
	This finishes the proof of the forward implication.
	
    For the backward implication, we use that by \Cref{lem:3con} there 
	are \srestricted drawings \mcG[_1] and \mcG[_2] of \(G_1\) and \(G_2\),
	respectively, where \(st_1\) and \(t_1s'\) are not crossed in \mcG[_1]
	and \(st_2\) and \(t_2s'\) are not crossed in \mcG[_2]. For \(i\in\{1,2\}\),
	the removal of the edges \(st_i\) and \(t_is'\) in the drawing \mcG[_i] 
	does not affect the \srestricted property of the drawing because \(t_i\) 
	has degree two and thus \(st_1\) and \(t_1s'\) cannot be part of any 4-cycle defined by a crossing.
	We can then combine these two drawings into a single one, for example 
	placing a drawing equivalent to \(\mcG[_2]-t_2\) inside the 
	(closure of the) face of \(\mcG[_1]-t_1\) that contained 
	the non-crossed 3-path \(\mcG[_1](st_1)\cup \mcG[_1](t_1s')\).
\end{proof}

We are now ready to prove \Cref{thm:3con}.

\begin{proof}[Proof of \Cref{thm:3con}]
	Because of \Cref{obs:2con} we can restrict our attention to each
	\(2\)-connected component separately. Thus, we can just assume
	that \(G\) is \(2\)-connected.
	
	Consider the SPR-tree \(T\) of \(G\). Each virtual edge \(ss'\)
	(in some skeleton) of the tree defines a \(2\)-separator \(\{s,s'\}\)
	of \(G\). Because of \Cref{lem:3con:alt} we can split \(G\) into 
	two subinstances, where the edges \(st_{ss'}\) and \(t_{ss'}s'\) are added 
	to each of the subinstances. (Each time with a new vertex \(t_{s,s'}\).) 
	We repeat this procedure for all (paired) virtual edges, and we obtain
	that \(G\) is \srestricted if and only if the following holds
	for each node \(x\) of \(T\): the subgraph of \(G\) induced by
	the vertices of \sk{x} and with the virtual edges inserted
	and subdivided once, is \srestricted. We note that for each node 
	\(x\) of \(T\), the graph we are considering is actually the skeleton+
	of node \(x\), \skPlus{x}. 
	For each P- and S- nodes, it is obvious that its skeleton+ is \srestricted, 
	as it is planar. Therefore, we obtain the following intermediary statement: 
	A graph is \srestricted if and only if the skeleton+ of each 
	of its 3-connected components (R-nodes) is \srestricted.	
	
    Let us finish the proof with a warning about possible further simplifications.
	Consider an R-node \(x\) of the SPR-tree and a virtual edge \(uv\)
    of its skeleton \sk{x}. One may be tempted to think that, if
    the edge \(uv\) is in \(E(G)\), whether we add it to \skPlus{x} 
	or not does not affect to whether \skPlus{x} is \srestricted.
	After all, in \skPlus{x} we have a 3-path \(ut_{uv}v\), 
	where \(t_{uv}\) is a new degree-two vertex, by \Cref{lem:3con}
	we may assume that in a potential \srestricted drawing  
	of \(\skPlus{x}-uv\) the path \(ut_{uv}v\) is not crossed,
	and thus the edge \(uv\) can be	inserted into the drawing 
	following closely the drawing of the non-crossed path \(ut_{uv}v\). 
	This is true as far as the drawing of \(uv\) goes. 
	However, whether the edge \(uv\) is in \skPlus{x} 
	affects whether some other edges, say \(uu'\) and \(vv'\), have
	a crossing of type in \(\mathcal{S}\) in the drawing, 
	and therefore we cannot remove or insert such an edge \(uv\) 
	from \skPlus{x} without affecting the property of being \srestricted.	
    See e.g.\ the right side  of \cref{fig:sk+}, where the presence or absense of \(v_2 v_3\) and \(v_3 v_4\) together determine whether the shown crossing is of type \arrowCt, \almostFullCt, or \fullCt.
\end{proof}

\subsection{The geometric setting}
\label{sec:3con:geom}

The aim of this section is to prove~\Cref{thm:3con:geom}. 
To ultimately arrive at internally 3-connected instances, 
we again use a sequence of lemmas. The statements are more
cumbersome because we have to keep track of the vertex that
has to be on the outer face, when $O$ is nonempty.

\begin{lemma}
\label{lem:geom-2con}
    Let \(G\) be a connected graph, \(O \subset V(G)\) with \(|O| \leq 1\), 
	\(s\) a cutvertex of \(G\), \((X_1,X_2)\) a separation induced by \(s\),
	\(G_1\coloneqq G[X_1]\) and \(G_2\coloneqq G[X_2]\).
    Then \(G\) is \(O\)-\geosrestricted if and only if at least one of the 
	following holds:
	\begin{enumerate}[(i)]
	\item \label{lem:geom-2con:option1}
		\(O\subset X_1\), the graph \(G_1\) is \(O\)-\geosrestricted, 
		and the graph \(G_2\) is \(\{s\}\)-\geosrestricted; or
	\item \label{lem:geom-2con:option2}
		\(O\subset X_2\), the graph \(G_1\) 
		is \(\{s\}\)-\geosrestricted, and the graph \(G_2\) is 
		\(O\)-\geosrestricted.
	\end{enumerate}
\end{lemma}
\begin{proof}
    For the forward implication, consider an \(O\)-\geosrestricted drawing \mcG 
	of \(G\), let \mcG[_1] be its restriction to \(G_1\) and let 
	\mcG[_2] be its restriction to \(G_2\).	
	Note that there cannot be any crossing between an edge from \(G_1\) and 
	an edge from \(G_2\) because crossings are contained in a 4-cycle but \(s\) is
	a cutvertex of \(G\). Therefore \mcG[_1] and \mcG[_2] are \geosrestricted drawings.
	Moreover, the drawing \mcG[_1] has to be inside	a face of \mcG[_2] 
	that has \(s\) on its boundary. 
	
	Because of the symmetry between \(X_1\) and \(X_2\),
	we may assume without loss of generality that either 
	\(O\neq \emptyset\) and \(O\subset X_1\) or that
	\(O=\emptyset\) and a part of \mcG[_1] contributes to the outer face of \mcG.
	In either case, we have the property that \mcG[_1] contributes some part 
	to the outer face of \mcG.
	Clearly, in both cases \mcG[_1] witnesses that \(G_1\) is 
	\(O\)-\geosrestricted. Moreover, \mcG[_2] witnesses that 
	\(G_2\) is \(\{s\}\)-\geosrestricted; indeed, if in the 
	drawing \mcG[_2] the vertex \(s\) is not on the outer face, 
	then \mcG[_1] would have to be drawn inside a bounded face of \mcG[_2], 
	contradicting that \mcG[_1] contributes to the outer face of \mcG.
	This finishes the proof of the forward implication.

	For the backward implication, let us assume without loss of generality
	that \(O \subset X_1\), that \mcG[_1] is an \(O\)-\geosrestricted drawing 
	of \(G_1\), and that \mcG[_2] is a \(\{s\}\)-\geosrestricted drawing 
	of \(G_2\). Obviously, \mcG[_1] and \mcG[_2] do not have any B- or 
	W-configuration.
	Let \(\Delta_s\) be a geometric triangle with vertex \(s\) such that 
	\(\Delta_s \setminus\{s\}\) is contained inside a face of \mcG[_1]; 
	such a triangle exists because \mcG[_1] is a geometric drawing.
	
	We treat the drawing \mcG[_2] as a non-geometric drawing, thus 
	allowing curves for the edges, deform it continuously to place it 
	inside the triangle \(\Delta_s\), identifying \(s\) in both drawings. 
	Let \mcG be the resulting (possibly non-geometric) drawing 
	of \(G\). If \(O\) has a vertex, then that vertex is on the 
	outer face of \mcG. (Here it is relevant that \mcG[_2] was placed 
	in the triangle \(\Delta_s\) and thus does not enclose \mcG[_1].)
	Since we do not introduce any crossings between \mcG[_1]
	and the deformed \mcG[_2], the drawing \mcG is \srestricted. 
	We note that \mcG has no B- or W-configuration because we did not introduce
	any additional crossings and \mcG[_1] and the deformed \mcG[_2] have
	a single point in common, namely the vertex \(s\).
	Therefore \mcG is an \srestricted drawing of \(G\) without 
	B- or W-configuration, and the vertex
	of \(O\), if there is one, is on the outer face of \mcG.
	Then there is a drawing of \(G\) equivalent to \mcG that is
	\(O\)-\geosrestricted.
\end{proof}

\Cref{lem:geom-2con} is written a bit in a compact way that showcases 
the symmetry. For our forthcoming algorithm, the following non-symmetric
statement will be more convenient.

\begin{corollary}
\label{cor:geom-2con-for-algo}
    With the assumptions and the notation of \Cref{lem:geom-2con}, the
    following holds.
    \begin{enumerate}[(i)]
	\item If \(G_1\) is \(\{s\}\)-\geosrestricted and \(O\subset X_2\), possibly 
		with \(O=\emptyset\), then \(G\) is \(O\)-\geosrestricted if and only
		if \(G_2\) is \(O\)-\geosrestricted.
	\item If \(G_1\) is not \(\{s\}\)-\geosrestricted or \(O\not\subset X_2\),
		then \(G\) is \(O\)-\geosrestricted if and only if
		\(O\subset X_1\),
		the graph \(G_1\) is \(O\)-\geosrestricted, and
		the graph \(G_2\) is \(\{s\}\)-\geosrestricted.
	\end{enumerate}
\end{corollary}
\begin{proof}
	The statement follows logically from \Cref{lem:geom-2con}, using
	a case analysis and without any geometric insights.

	Consider first the case where \(G_1\) is \(\{s\}\)-\geosrestricted 
	and \(O\subseteq \{s\}\), possibly with \(O=\emptyset\).
	In this case, \(G_1\) is also \(O\)-\geosrestricted. This means that
	both options in \Cref{lem:geom-2con} become
	\begin{enumerate}[(i)]
	\item \true, \true, and the graph \(G_2\) is \(\{s\}\)-\geosrestricted; or
	\item \true, \true, and the graph \(G_2\) is \(O\)-\geosrestricted.
	\end{enumerate}
	Whenever option (\ref{lem:geom-2con:option1}) occurs, 
	also option (\ref{lem:geom-2con:option2}) occurs.
	Therefore, in this case, \(G\) is \(O\)-\geosrestricted if and only if
	the graph \(G_2\) is \(O\)-\geosrestricted.
	
	Consider now the case where \(G_1\) is \(\{s\}\)-\geosrestricted, 
	\(O\neq\emptyset\), and \(O \subset X_2\setminus \{s\}\).
	In this case, \(G_1\) is not \(O\)-\geosrestricted. This means that
	both options in \Cref{lem:geom-2con} become
	\begin{enumerate}[(i)]
	\item \false, \false, and the graph \(G_2\) is \(\{s\}\)-\geosrestricted; or
	\item \true, \true, and the graph \(G_2\) is \(O\)-\geosrestricted.
	\end{enumerate}	
	Therefore, in this case, \(G\) is \(O\)-\geosrestricted if and only if
	the graph \(G_2\) is \(O\)-\geosrestricted.
	The two cases considered so far cover the first item to be proved.
	
	Consider now the remaining case, covered in the second item.
	Thus, we assume that \(G_1\) is not \(\{s\}\)-\geosrestricted
	or \(O\not\subset X_2\). 
	This means that option (\ref{lem:geom-2con:option2}) in 
	\Cref{lem:geom-2con} cannot occur,
	and therefore \(G\) is \(O\)-\geosrestricted if and only if
	the conditions of option (\ref{lem:geom-2con:option1}) in 
	\Cref{lem:geom-2con} occur:
	\(O\subset X_1\),
	the graph \(G_1\) is \(O\)-\geosrestricted, and
	the graph \(G_2\) is \(\{s\}\)-\geosrestricted.	
\end{proof}

\begin{lemma}
\label{lem:3con-geom}
	Let \(t\) be a vertex of degree two in a graph \(G\)
	and let \(O \subset V(G)\) with \(|O| \leq 1\).
	Then \(G\) is \(O\)-\geosrestricted 
	if and only if there is a \(O\)-\geosrestricted drawing 
	of \(G\) where the two edges incident to \(t\) are not crossed.
\end{lemma}
\begin{proof}
	The backward implication is clear. For the forward implication,
	consider an \(O\)-\geosrestricted drawing \mcG of \(G\).
	If in \mcG the two edges incident to \(t\) are not crossed,
	we are done. The rest of the proof is about the situation where
	in \mcG at least one of the two edges incident to \(t\) is crossed
	and how to modify the drawing.
	
	Let \(s\) and \(s'\) be the two neighbors of \(t\), labeled in 
	such a way that in \mcG the edge \(st\) is crossed at a point \(p\). 
	As crossing edges must be on a 4-cycle \(C\) and \(t\) has 
	degree two, the edge crossing \(st\) must be incident to \(s'\).
	Let \(v\) be the fourth vertex in \(C\), so that \(C=sts'vs\).
	\Cref{fig:lem:3con-geom:1} might help to follow the discussion.
	If \(O\neq \{t\}\), we can redraw the edges \(st\) and \(s't\)
	placing \(t\) sufficiently near \(p\), and we are done.
	See \Cref{fig:lem:3con-geom:1}, left.
	However, if \(O= \{t\}\), we have to redraw differently to
	keep \(t\) in the outer face. The rest of the proof is about
	this case, where \(O= \{t\}\). 
	See \Cref{fig:lem:3con-geom:1}, center.
	
	\begin{figure}
		\centering
		\includegraphics[page=6, width=\textwidth]{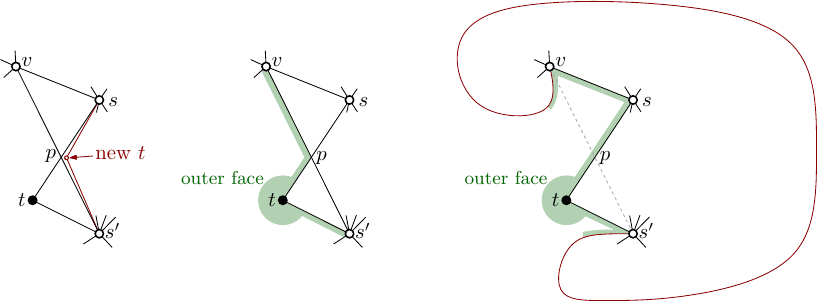}
		\caption{Left: The easy scenario in the proof of 
			\Cref{lem:3con-geom}, when \(O\neq \{t\}\).
			Center: Scenario in the proof of \Cref{lem:3con-geom}
				when \(O=\{t\}\).
			Right: redrawing of the edge \(vs'\) when \(O=\{t\}\) 
				and no edge incident to \(v\) and no edge 
				incident to \(s'\) cross.
			}
		\label{fig:lem:3con-geom:1}
	\end{figure}

	We next note that the edge \(s't\) cannot 
	be crossed by any edge. Indeed, by the same argument
	we used before for \(st\), the edge crossing \(s't\) would
	be contained in a 4-cycle of the form \(sts'v's\) for a vertex
	\(v'\), but then the vertex \(v'\) should be inside the 
	triangle \(\triangle(pts')\) and connected to \(s\) by a 
	straight-line segment that crosses \(ts'\), which is 
	geometrically impossible. (One can also 
	argue that there is a B-configuration with spine \(s,t\)
	and crossing edges \(ts', v's\).) 
	
	Since \(t\) is on the outer face, has degree two, and the edge
	\(vs'\) is crossed only by \(st\), the vertex \(v\) is also
	on the outer face.
	Similarly, 	since \(t\) in on the outer face, has degree two, 
	and the edge \(s't\) is not crossed by any other edge, 
	the vertex \(s'\) is also on the outer face.
	We conclude that \(v,t,s'\) are vertices on the outer face
	of \mcG and they are in that order along the outer face. 
	We are going to redraw the edge \(vs'\).
	
	We consider two cases. In the first case, assume that there
	are no edges \(vu\) and \(s'u'\) that cross. In this case
	we redraw the edge \(vs'\) as a curve through the outer face 
	in such a way that \(t\) remains on the outer face; see 
	\Cref{fig:lem:3con-geom:1}, right. Let \mcG['] be the 
	new drawing of \(G\). With such a redrawing 
	of \(vs'\) we do not introduce any W-configuration because
	such a configuration needs new crossing edges. We do not
	introduce any B-configuration either because the newly
	drawn edge \(vs'\), as it has no crossings, could be only 
	the edge connecting the spine of a B-crossing, but since 
	there are no edges \(vu\) and \(s'u'\) that cross, 
	we cannot have such a B-configuration in \mcG[']. 
	
	In the second case, assume that there are edges	\(vu\) 
	and \(s'u'\) that cross in \mcG. There may be several such pairs
	\( (vu_1, s'u'_1),\dots, (vu_k,s'u'_k)\) of edges such that
	the edges \( vu_i\) and \(s'u'_i\) cross at point \(p_i\)
	in \mcG, for each \(i\in \{1,\dots,k\}\). Each such pair defines
	a triangle \(\Delta_i\) defined by the edge \(vs'\)	and the 
	point \(p_i\) as vertex of the triangle. 
	See \Cref{fig:lem:3con-geom:2}, left.
	Each such a triangle \(\Delta_i\) is contained in the halfplane 
	bounded by the line through \(v\) and \(s'\) and without \(t\) in
	its interior: otherwise \(t\) would not be on the outer face. 
	Since the edges \(vu_i\) and \(s'u'_i\)
	cross each other, they cannot cross any other edge.
	It follows that the triangles \(\Delta_1,\dots,\Delta_k\)
	are nested, meaning that for each two of them, one
	has to be inside the other.

	Continuing with the second case, let us assume, 
	without loss of generality, that \(\Delta_1\)
	is the smallest triangle. We redraw the edge \(vs'\) as a curve 
	inside \(\Delta_1\) that closely follows \(vu_1\) from \(v\)
	until the neighborhood of \(p_1\) and then closely follows \(u'_1s'\)
	until \(s'\). See \Cref{fig:lem:3con-geom:2}, right.
	Let \mcG['] be the resulting drawing of \(G\).
	Because the edges \(vu_1\) and \(s'u'_1\) cannot
	cross any other edge, in the new drawing \mcG['] the edge \(vs'\) 
	does not have any crossings. 
	In the new drawing we do not introduce any W-configuration 
	because such a configuration needs new crossing edges. We do not
	introduce any B-configuration either because the newly
	drawn edge \(vs'\), as it has no crossings, could be only 
	the edge connecting the spine of a B-crossing, but
	none of the crossings between the edges \( vu_i\) and \(s'u'_i\)
	(for any \(i\in \{1,\dots,k\}\) form a B-configuration in \mcG['].
	
	\begin{figure}
		\centering
		\includegraphics[page=7,width=\textwidth]{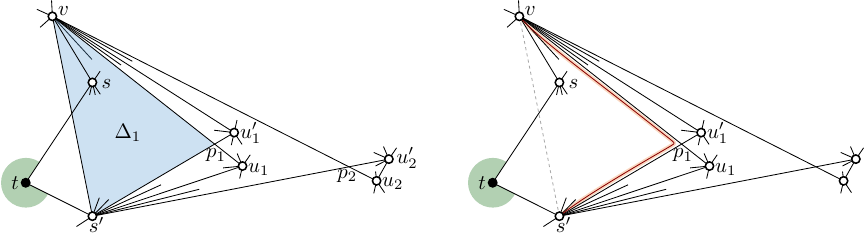}
		\caption{Left: Scenario in the proof of \Cref{lem:3con-geom},
			when \(O=\{t\}\) and some edge incident	\(v\) 
			and some edge incident to \(s'\) cross. 
			The triangle \(\Delta_1\) is shaded in light blue.
			Right: redrawing of the edge \(vs'\) in this case.
			}
		\label{fig:lem:3con-geom:2}
	\end{figure}
	
	In both cases, we have obtained an \srestricted drawing 
	\mcG['] without neither B- nor W-configuration, with \(t\) in the 
	outer face, and the edges two edges \(st,s't\) incident to \(t\) 
	are not crossed; we already argued that \(s't\) is not crossed
	and \(st\) is not crossed because of the redrawing of the edge \(vs'\). 
	The existence of a geometric drawing equivalent to \mcG['] finishes
	the proof.
\end{proof}

\begin{lemma}
\label{lem:geom-3con-R}
	Let \(G\) be a 2-connected graph, \(O \subset V(G)\) with \(|O| \leq 1\),
	\(\{s,s'\}\) a 2-separator of \(G\), and \((X_1,X_2)\) a separation 
	induced by \(\{s,s'\}\).
	Define \(G_1\coloneqq G[X_1]+\{st_1,t_1s'\}\) and 
	\(G_2\coloneqq G[X_2]+\{st_2,t_2s'\}\), where \(t_1\) and \(t_2\)
	are new vertices. \textup(Note that if \(ss'\in E(G)\), 
	then \(ss'\in E(G_1)\) and \(ss'\in E(G_2)\).\textup) 
    Then \(G\) is \(O\)-\geosrestricted if and only if at least one of the 
	following holds:
	\begin{enumerate}[(i)]
	\item \(O\subset X_1\), the graph \(G_1\) is \(O\)-\geosrestricted, 
		and the graph \(G_2\) is \(\{t_2\}\)-\geosrestricted; or
	\item \(O\subset X_2\), the graph \(G_1\) 
		is \(\{t_1\}\)-\geosrestricted, and the graph \(G_2\) is 
		\(O\)-\geosrestricted.
	\end{enumerate}
\end{lemma}
\begin{proof}
	For the forward implication, consider an \(O\)-\geosrestricted 
	drawing \mcG of \(G\), let \mcG[_1] be its restriction to 
	\(G_1-t_1=G[X_1]\) and let \mcG[_2] be its restriction 
	to \(G_2-t_2=G[X_2]\).
	As shown in the proof of \Cref{lem:3con:alt}, the vertices 
	\(s\) and \(s'\) are cofacial in the drawings \mcG[_1] and \mcG[_2].
	
	Because of the symmetry between \(X_1\) and \(X_2\),
	we may assume without loss of generality that either 
	\(O\neq \emptyset\) and \(O\subset X_1\) or that
	\(O=\emptyset\) and a part of \mcG[_1] contributes to the outer face of \mcG.
	In either case, we have the property that \mcG[_1] contributes some part 
	to the outer face of \mcG.
	Clearly, in both cases \mcG[_1] witnesses that \(G[X_1]=G_1-t_1\) 
	is \(O\)-\geosrestricted. Since \(s\) and \(s'\) are cofacial in \mcG[_1], 
	we can add the edges \(st_1\) and \(t_1s'\) to \mcG[_1] without introducing 
	new crossings or B- or W-configurations. If \(s\) and \(s'\) lie in
	the outer face, we have two different ways to draw the path \(st_1s'\),
	and in one of them the vertex of \(O\), if it exists, remains
	on the outer face.  This shows that the graph
	\(G_1\) is \(O\)-\geosrestricted. 

	It remains to argue that \(G_2\) is \(\{t_2\}\)-\geosrestricted.
	If \(s\) and \(s'\) are both on the boundary of the outer face of \mcG[_2],
	we just add to \(G_2\) a point for \(t_2\) in the interior of the outer
	face of \mcG[_2] and draw the edges \(st_2,s't_2\) without crossings.
	We get a drawing of \(G_2\) where \(t_2\) is lying on the boundary
	of the outer face and without neither B- nor W-configurations.
	Thus, there is an equivalent \(\{t_2\}\)-\geosrestricted drawing
	of \(G_2\) and we have finished the forward implication.
	
	If \(s\) or \(s'\) are not on the boundary of the outer face 
	of \mcG[_2], we modify the drawing \mcG[_2] into another geometric
	drawing of \(G_2-t_2\) that has \(s\) and \(s'\) on the boundary 
	of the outer face, and then, as before, we again conclude 
	that \(G_2\) is \(\{t_2\}\)-\geosrestricted.
	We will use the property noted in the proof of \Cref{lem:3con:alt}:
	any crossing between an edge of \(\mcG[_1]-ss'\) and an
	edge of \(\mcG[_2]-ss'\) occurs between an edge incident to \(s\)
	and an edge incident to \(s'\). Therefore, there are no
	crossings between \(\mcG[_1]-\{s,s'\}\) and \(\mcG[_2]-\{s,s'\}\).
	
	\begin{figure}
		\centering
		\includegraphics[page=10]{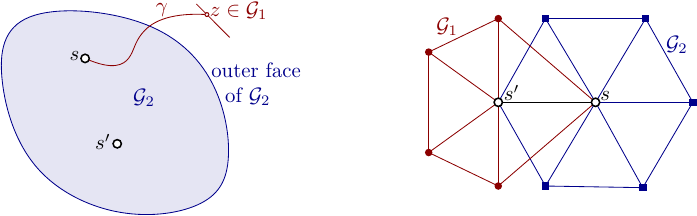}
		\caption{Scenario in the forward implication of the 
			proof of \Cref{lem:geom-3con-R}.
			Left: showing that \(s\) or \(s'\) are on the
			boundary of the outer face of \mcG[_2].
			Right: the claim is tight, as it may be that
			only one vertex from \(\{s,s'\}\) is on the boundary 
			of the outer face of \mcG[_2]. (In this example we assume that
			\(\almostFullCt\in\mathcal{S}\).)
			}
		\label{fig:lem:geom-3con-R:A}
	\end{figure}

	First, we note that at least one of the vertices \(s,s'\)
	must lie on the boundary of the outer face of \mcG[_2]. Indeed, 
	assume for the sake of reaching a contradiction that neither 
	\(s\) nor \(s'\) lie on the boundary of the outer face of \mcG[_2]. 
	Then the boundary of the outer face of \mcG[_2] is formed by 
	full edges of \mcG[_2] that are not incident to \(s\) or \(s'\) 
	and by strict \emph{sub}paths of edges of \mcG[_2] that are 
	crossed by another edge of \mcG[_2]. Thus, the boundary of the 
	outer face of \mcG[_2] cannot be crossed by any edge of \mcG[_1].
	This contradicts the assumption that some point of \mcG[_1],
	say \(z\), is on the boundary of the outer face of \mcG because
	\mcG[_1] contains a curve \(\gamma\) from \(s\) to \(z\) and
	such a curve must cross the boundary of the outer face of \mcG[_2].
	See \Cref{fig:lem:geom-3con-R:A} for a schema, where
	we also see that the claim is tight: it may be that one of 
	vertices from \(\{s,s'\}\) is not on the outer face of \mcG[_2].
	This finishes the proof of the claim.
	(Here we did not use that \mcG is a geometric drawing.)

	Without loss of generality we assume henceforth that \(s'\)
	lies on the boundary of the outer face of \mcG[_2].	
	Therefore \(s\) does not lie on the boundary of the outer face
	of \mcG[_2].
	Since \(G_1\) is 2-connected and \mcG[_1] has some point \(z\)
	on the boundary of \mcG, there exists a path \(\gamma\)
	in \mcG[_1] from \(s\) to \(z\) that is disjoint from \(s'\).
	Because of the restrictive types of crossings between
	\mcG[_1] and \mcG[_2], we have the following property:
	as we follow \(\gamma\) from \(s\) to \(z\), the first edge
	of \(\gamma\) must cross an edge \(s'v\) of \mcG[_2] at a point,
	which we denote by \(p\), and from there on \(\gamma\) cannot
	cross \mcG[_2] anymore. From this property, it follows that 
	\(p\) lies on the boundary of the outer face of \mcG[_2],
	and since \(s'v\) cannot be crossed by any other edge of
	\mcG[_2], the whole edge \(s'v\) lies on the outer face of \mcG[_2].
	We add to \mcG[_2] the vertex \(t_2\) on \(\gamma\) slightly 
	after \(p\) and connect it with straight-line edges to 
	\(s\) and \(s'\). This introduces a new crossing between
	\(st_2\) and \(s'v\) (that may not of the type \(\mathcal{S}\)). 
	We redraw the edge \(s'v\) to get rid of the new crossing, 
	essentially as it was done in the proof of 
	\Cref{lem:3con-geom}: if in \mcG[_2] no edges of the 
	form \(vu\) and \(s'u'\) cross, then we redraw \(s'v\)
	around the opposite side of the outer face of \mcG[_2],
	and if such crossing edges exists, we redraw \(s'v\)
	near the pair of such crossing edges that together
	with the current \(sv\) defines the smallest triangle.
	See \Cref{fig:lem:geom-3con-R:B}.
	
	\begin{figure}
		\centering
		\includegraphics[page=11,width=\textwidth]{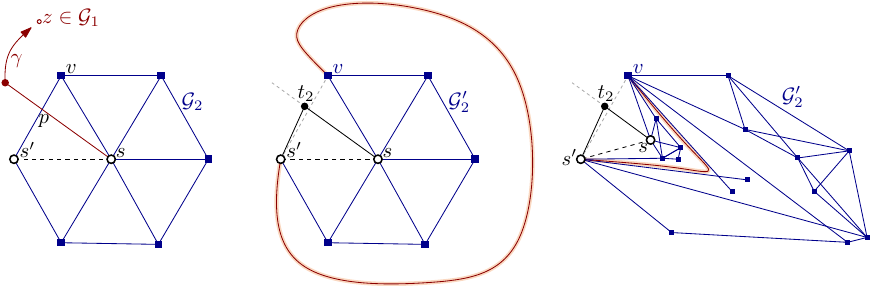}
		\caption{Scenario in the forward implication of the 
			proof of \Cref{lem:geom-3con-R}.
			Left: the edge \(s'v\in \mcG[_2]\) lies in the outer face
			of \mcG[_2].
			Center and right: the possible options for
			redrawing \(s'v\) and inserting the edges \(st_2,s't_2\)
			in the drawing \mcG[_2].
			}
		\label{fig:lem:geom-3con-R:B}
	\end{figure}	

	Let \mcG['_2] be the resulting drawing of \(G_2\). 
	As it was done in the proof of \Cref{lem:3con-geom},
	\mcG['_2] has no B- or W-configuration and the 
	vertices \(s,s'\) lie on the boundary of the outer
	face. Taking a geometric drawing equivalent to \mcG['_2],
	we have reduced the situation to the previous scenario,
	where we have a geometric drawing of \(G_2-t_2\) with
	\(s\) and \(s'\) on the boundary of the outer face.
	This finishes the proof of the forward implication.
		
	For the backward implication, let us assume without loss of generality
	that \(O \subset X_1\), the graph \(G_1\) is \(O\)-\geosrestricted, 
	and the graph \(G_2\) is \(\{t_2\}\)-\geosrestricted.
	Let \mcG[_1] be an \(O\)-\geosrestricted drawing 
	of \(G_1\) and let \mcG[_2] be a \(\{t_2\}\)-\geosrestricted drawing 
	of \(G_2\). Obviously, \mcG[_1] and \mcG[_2] do not have 
	any B- or W-configuration.
	Because of \Cref{lem:3con-geom}, we may assume
	that in \mcG[_1] and in \mcG[_2] the edges incident to \(t_1\) and \(t_2\)
	are not crossed by any other edge. As a consequence, the vertices
	\(s,t_2,s'\) lie in the outer face of the drawings \mcG[_2].

	\begin{figure}
		\centering
		\includegraphics[page=8,width=\textwidth]{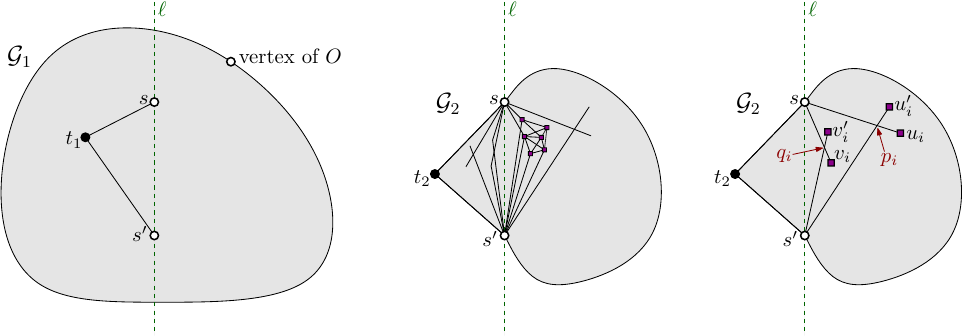}
		\caption{Scenario in the backward implication of the 
			proof of \Cref{lem:geom-3con-R}.
			Left: the drawing \mcG[_1].
			Center: the drawing \mcG[_2]; the vertices of one set
				\(Y_i\) are marked with small squares. 
			Right: each \mcH[_i] has at most one crossing
				between edges incident to \(s,s'\).
			}
		\label{fig:lem:geom-3con-R:1}
	\end{figure}

	Making affine transformations of each drawing, we may assume that in 
	\mcG[_1] and in \mcG[_2] the vertices \(s\) and \(s'\) 
	agree in both drawings, that \(s\) and \(s'\) lie on a vertical 
	line \(\ell\) and that in both drawings \(t_1\) and \(t_2\) lie 
	to the left of the line \(\ell\). See \Cref{fig:lem:geom-3con-R:1}
	to follow the discussion.
	Let \(Y_1,\dots,Y_k\) be the vertex sets of the 
	connected components of \(G_2-\{s,s',t_2\}\). Define, for 
	each \(i\in \{1,\dots,k\}\), the graph 
	\(H_i=G_2[Y_i\cup \{s,s'\}]\) and let \mcH[_i]
	be the restriction of \mcG[_2] to \(H_i\).
	Note that, for each \(i\in \{1,\dots,k\}\), 
	the drawing \mcH[_i] may have at most one
	pair of edges of the form \(su_i\) and \(s'u'_i\)
	that cross. Indeed, if inside \mcH[_i] there would
	be two such pairs of crossing edges, say \(su_i\) crosses
	with \(s'u'_i\) at \(p_i\) and \(sv_i\) crosses 
	with \(s'v'_i\) at \(q_i\), then the three geometric
	paths \(st_2s'\), \(sq_is'\) and \(sp_is'\)
	separate \(\{u_i,u'_i\}\) from \(\{v_i,v'_i\}\), 
	and there cannot be any path in 
	\(H_i-\{s,s'\}=G_2[Y_i]\) from \(u_i\) to \(v_i\); 
	see \Cref{fig:lem:geom-3con-R:1}, right.
	This would contradict the definition of \(Y_i\) because
	it is defined as inducing a connected graph.
	
	\begin{figure}
		\centering
		\includegraphics[page=9,width=\textwidth]{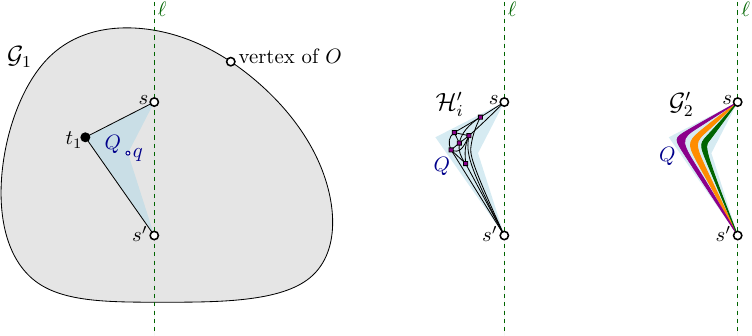}
		\caption{Left: the empty quadrangle	\(Q\) with boundary
			\(st_1s'q\) in \mcG[_1].
			Center: example of how to transform the drawing
			of \Cref{fig:lem:geom-3con-R:1}, center, to
			fit into \(Q\). 
			Right: schematic view of the drawing \mcG['_2],
			where we have three subdrawings \mcH['_1], \mcH['_2] 
			and \mcH['_3], each with a different color.
			}
		\label{fig:lem:geom-3con-R:2}
	\end{figure}

	Select a point \(q\) inside the (geometric) triangle with 
	vertices at \(s,s',t_1\) 
	such that the quadrangle \(Q=t_1s'qs\) is simple and 
	its interior is disjoint from \mcG[_1].
	See \Cref{fig:lem:geom-3con-R:2}, left.
	We are going to redraw \(G_2-t_2-ss'=\cup_{i=1} H_i-ss'\), 
	using curves inside the quadrangle \(Q\).
	For \(i=1,\dots, k\), we do the following:
	if \mcH[_i] has an edge incident to \(s\) and an
	edge incident to \(s'\) that cross to the right 
	of \(\ell\), we replace \mcH[_i] with its mirror
	image through \(\ell\) and bring
	the crossing to the left of \(\ell\);
	then, we continuously deform the drawing \(\mcH[_i]-ss'\)
	into a drawing, denoted \mcH['_i], that fits into \(Q\) 
	without intersecting the previous deformed drawings
	\mcH['_j] for any \(j<i\), but for the common vertices
	\(s,s'\) and the edge \(ss'\), it it exists, that is kept
	as a vertical straight-line. See \Cref{fig:lem:geom-3con-R:2}.
	The final drawing of \(G_2-t_2\) we obtain,
	which we denote \mcG['_2], is simply the union
	of \(\mcH['_1],\dots,\mcH['_k]\), where we remove all
	but one of the copies of \(ss'\), if it exists in \(G_2\).
	The drawing \mcG['_2] is \srestricted because each of
	the drawings \mcH['_i] (for \(i\in \{1,\dots,k\})\)
	is \srestricted and there are no crossings between them.
	(In the deformation to get \mcH['_i] we may have lost 
	crossings with \(ss'\), but losing crossings does not 
	affect the property of being \srestricted.)
	
	Let \mcG['] be the overlap of \(\mcG[_1]-t_1\) and \mcG['_2];
	we remove one of the copies of \(ss'\), if the edge exists. 
	It is clear that \mcG['] is an \srestricted drawing
	of \(G\) that has the vertex of \(O\), it is exists,
	on the boundary of the outer face.
	We next note that \mcG['] has no B-configuration
	because there is no crossing between any of the graphs
	\mcG[_1], \mcH['_1],\dots, \mcH['_k] and none of them
	had a B-configuration.
	Similarly, there is no W-configuration in \mcG['] because
	of the following:
	\begin{itemize}
	\item Each of the drawings \mcG[_1], \mcH['_1],\dots, \mcH['_k]
		has no W-configuration.
	\item A W-configuration combining two of the drawings of
		\mcG[_1], \mcH['_1],\dots, \mcH['_k] should have
		spine \(\{s,s'\}\), but any two crossing
		edges incident to \(s\) and \(s'\) have the endpoints
		outside the closed curve defined by the crossing
		edges and the segment \(ss'\).
	\end{itemize}
	We conclude that \mcG['] is an \srestricted drawing of \(G\)
	without neither B- nor W-configuration and with the vertex
	of \(O\), if it exists, on the outer face of \mcG['].
	Therefore there exists an equivalent \(O\)-\geosrestricted
	drawing of \(G\).	
\end{proof}

Using logical implications, as we did in the proof of 
\Cref{cor:geom-2con-for-algo}, we have the following non-symmetric
version of \Cref{lem:geom-3con-R}, which will be useful for 
our algorithm.

\begin{corollary}
\label{cor:geom-3con-R-for-algo}
    With the assumptions and the notation of \Cref{lem:geom-3con-R}, the
    following holds.
    \begin{enumerate}[(i)]
	\item If \(G_1\) is \(\{t_1\}\)-\geosrestricted and \(O\subset X_2\), possibly 
		with \(O=\emptyset\), then \(G\) is \(O\)-\geosrestricted if and only if
		the graph \(G_2\) is \(O\)-\geosrestricted.
	\item If \(G_1\) is not \(\{t_1\}\)-\geosrestricted or \(O\not\subset X_2\),
		then \(G\) is \(O\)-\geosrestricted if and only if \(O\subset X_1\),
		the graph \(G_1\) is \(O\)-\geosrestricted, and
		the graph \(G_2\) is \(\{t_2\}\)-\geosrestricted.
	\end{enumerate}
\end{corollary}
\begin{proof}
	The proof is very similar to the proof of \Cref{cor:geom-2con-for-algo}
	and we omit the details.
	The only additional observation that is needed is that if \(G_1\) is
	\(\{t_1\}\)-\geosrestricted, then it is also \(\{s\}\)-\geosrestricted
	and \(\{s'\}\)-\geosrestricted because of \Cref{lem:3con-geom}. This 
	is relevant when considering the case in item (i) with 
	\(O\cap \{s,s'\}\neq \emptyset\).
\end{proof}

We are now ready to prove \Cref{thm:3con:geom}.

\begin{proof}[Proof of \Cref{thm:3con:geom}]
	It is obvious that deciding \geosrestricted[ity] is equivalent to 
	deciding \(\emptyset\)-\geosrestricted[ity].

	Assume that we have an \FPT-algorithm \textsc{Solve-i3con-FPT} 
	parameterized by treewidth that decides whether internally 
	\(3\)-connected graphs are \(O\)-\geosrestricted. 

    Consider an arbitrary input graph \(G\). Using 
	\Cref{cor:geom-2con-for-algo}, we make a subroutine, which we call 
	\textsc{Solve-1conn}, that formulates a pruning-leaves dynamic 
	program along the BC-tree of \(G\) and makes calls to a 
	subroutine, called \textsc{Solve-2conn}, that decides whether 
	2-connected graphs are 
	\(O\)-\geosrestricted for some given \(O\) of cardinality at most 1.
    More precisely, the subroutine is the following, 
	where initially \textsc{Solve-1con} is called on \(G, \emptyset\).
	
	\begin{algorithmic}[1]
	\label{alg:geo1}
    \Procedure{Solve-1con}{\(G',O\)} 
		\If{\(G'\) is 2-connected}
			\State \Return \Call{Solve-2con}{\(G',O\)}
		\EndIf
		\State \(\lambda \gets\) leaf in BC-Tree of \(G'\)
		\State \(s \gets \) cutvertex of \(\lambda\)
		\State \(G'_1\gets \lambda\)
        \State \(G'_2\gets G-(V(\lambda)\setminus \{s\})\)
		\If{\(O \subset V(G'_2)\) and \Call{Solve-2con}{\(G'_1,\{s\}\)}}
			\State \Return \Call{Solve-1con}{\(G'_2,O\)}
		\EndIf
		\If{\(O\subset V(G'_1)\) and \Call{Solve-2con}{\(G'_1,O\)}}
			\State \Return \Call{Solve-1con}{\(G'_2,\{s\}\)}
		\EndIf
        \State \Return \false
	\EndProcedure
    \end{algorithmic}
    
    To see the correctness of \textsc{Solve-1conn}, we use 
	\Cref{cor:geom-2con-for-algo}, where \(X_1=V(L)=V(G'_1)\)
	and \(X_2=V(G'_2)= V(G')\setminus (X_1\setminus \{s\}))\).
	The conditions tested in the algorithm are precisely the conditions
	stated in \Cref{cor:geom-2con-for-algo}.
	
    The subroutine \textsc{Solve-2con} for 2-connected graphs is also 
	formulated as a pruning-leaves dynamic program, but this time along the 
	SPR-tree of the graph and making calls to the subroutine
	\textsc{Solve-i3con-FPT}, which handles internally 3-connected graphs.
	If no R-node is present in the considered subtree, then the 
	corresponding graph is planar and any specified vertex can be placed 
	on the outer face. If the graph is internally 3-connected,
	we can directly invoke \textsc{Solve-i3con-FPT}.
	Otherwise, there is an R-node with some non-trivial parts attached to it,
	and we use \Cref{cor:geom-3con-R-for-algo} to dynamically reduce the
	instance. To explain the details, we introduce first some notation.	

	Consider an edge \(\mu\nu\) of an SPR-tree \(T\) of a 2-connected
	graph \(G'\) and let \(\{s,s'\}\) be the 2-separator defined by the 
	virtual edges of \sk{\mu} and \sk{\nu} that are identified.
	The edge \(\mu\nu\) defines two graphs, which we denote 
	denote \(G'(T,\mu,\nu)\) and \(G'(T,\nu,\mu)\), as follows.
	The graph \(G'(T,\mu,\nu)\) is the graph induced by the 
	vertices contained in the connected component of \(T-\mu\nu\) 
	that contains \(\mu\), where we add a new vertex \(t\) and
	edges \(st\) and \(s't\).
	The graph \(G(T,\nu,\mu)\) is defined symmetrically, but using
	the component of \(T-\mu\nu\) that contains \(\nu\).
	Note, that \(G'(T,\mu,\nu)\) and \(G'(T,\nu,\mu)\) are precisely
	the graphs \(G_1\) and \(G_2\) defined in \Cref{lem:geom-3con-R}
	with respect to the 2-separator \(\{s,s'\}\) (and for \(G'\)).

	We note that a graph is internally 3-connected if and only if
	it has a very particular SPR-tree: it has a single R-node \(x\)
	whose neighbors are either S-nodes representing 3-cycles or 
	P-nodes whose neighbors, distinct from \(x\), are again S-nodes 
	representing 3-cycles. See \Cref{fig:i3con} for an example.

	\begin{figure}
		\centering
		\includegraphics[page=12,width=.96\textwidth]{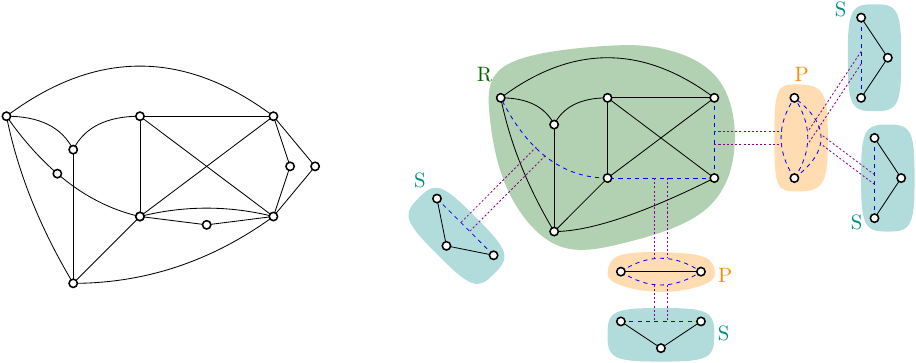}
		\caption{An internally 3-connected graph (left) and its 
			SPR-tree (right). Virtual edges are dashed blue; pairs
			of identified virtual edges are connected with two parallel 
			dotted connections to emphasize the orientation of the 
			identified edges.}
		\label{fig:i3con}
	\end{figure}
	
    If a 2-connected graph \(G'\) is not internally 3-connected,
    then its SPR-tree \(T\) has some edge \(\mu\nu\) such that
    \(G'(T,\mu,\nu)\) is internally 3-connected or a 
	series-parallel graph with at least five vertices.
    Indeed, we can take an edge \(xy\) of \(T\) such
    that \(x\) is the only R-node in the connected component 
	of \(T-xy\) that contains it; then one of the edges
	incident to \(x\) has the desired property, where
	\(x\) may play the role of \(\mu\) or \(\nu\) in the
	sought edge.
	Once we have such an edge \(\mu\nu\) of \(T\),
	both \(G'(T,\mu,\nu)\) and \(G'(T,\mu,\nu)\) are strictly
	smaller than \(G'\) and we can use \Cref{cor:geom-3con-R-for-algo}
	for splitting the problem.
	The notation we use in the algorithm is parallel to the notation
	used in \Cref{cor:geom-3con-R-for-algo} and correctness follows
	from that statement.
	
    \begin{algorithmic}[1]
	\label{alg:geo}
	\medskip
    \Procedure{Solve-2con}{\(G',O\)}
        \If{\(G'\) is series-parallel}
            \State \Return \true
        \EndIf
        \If{\(G'\) is internally 3-connected}
            \State \Return \Call{Solve-i3con-FPT}{\(G',O\)}
        \EndIf
        \State \(\mu\nu \gets\) edge in SPR-tree \(T\) of \(G'\) such
			that \(G'(T,\mu,\nu)\) is internally 3-connected or
			a series-parallel graph with at least five vertices
        \State \(G'_1\gets G'(T,\mu,\nu)\), where \(t_1\) is the new vertex in \(G'_1\)
        \State \(G'_2\gets G'(T,\nu,\mu)\), where \(t_2\) is the new vertex in \(G'_2\)
		\If{\(O \subset V(G'_2)\) and \Call{Solve-2con}{\(G'_1,\{t_1\}\)}}
            \State \Return \Call{Solve-2con}{\(G'_2,O\)}
        \EndIf
        \If{\(O\subset V(G'_1)\) and \Call{Solve-2con}{\(G'_1,O\)}}
            \State \Return \Call{Solve-2con}{\(G'_2, \{t_2\}\)}
        \EndIf
        \State \Return \false
    \EndProcedure 
    \end{algorithmic}

    To bound the running time, we note that each \(2\)-connected component 
	of \(G\) has the role of \(\lambda\) in at most one call of 
	\textsc{Solve-1con}. In that call, \textsc{Solve-2con}\((\lambda,\cdot)\) 
	is called at most twice.
    In each of these calls, each \(3\)-connected component of \(\lambda\) has the 
	role of \(\mu\) in at most one call of \textsc{Solve-2con},
	and \textsc{Solve-i3con-FPT}\((G'_1,\cdot)\) is called at most twice
	for a graph \(G'_1\) that contains \sk{\mu}. Therefore, there are at most 
	four calls to \textsc{Solve-i3con-FPT} per 3-connected component of the
	original, input graph \(G\).
	Each call we make in all three algorithms is on a minor \(G'\) of \(G\), 
	and thus the treewidth of all considered graphs is bounded by the treewidth 
	of \(G\).
	Apart from that, we have only a polynomial-time overhead for computing the
	BC- and SPR-trees, manipulating the subgraphs that have to be constructed,
	testing if a graph is internally 3-connected or series-parallel, etc.   
	This finishes the proof of \Cref{thm:3con:geom}.
\end{proof}

\section{Fixed-Parameter Tractable Cases by Treewidth}
\label{sec:fpt}

The aim of this section is to show the following theorems.
\begin{theorem}
\label{thm:main-algo}
    It is \FPT\ parameterized by treewidth to recognize \(\mathcal{S}\)-restricted 1-planar graphs if \(\mathcal{S} \subseteq \{\fullCt,\almostFullCt,\bowtieCt\}\).
\end{theorem}

\begin{theorem}
\label{thm:geom-algo}
    It is \FPT\ parameterized by treewidth to recognize geometric \(\mathcal{S}\)-restricted 1-planar graphs if \(\mathcal{S} \subseteq \{\fullCt,\almostFullCt,\bowtieCt\}\).
\end{theorem}

\textbf{Throughout the rest of this section, we assume that 
\(\mathcal{S} \subseteq \{\fullCt, \almostFullCt, \bowtieCt\}\) is fixed};
in the statements we do not quantify over \(\mathcal{S}\) anymore. 

Before moving onto the details, we discuss the general approach
in comparison to the approach used for the crossing number.
Parameterizing by treewidth allows us to employ the powerful 
machinery of Courcelle's theorem, which has previously been useful in the 
design of fixed-parameter algorithms for variants of the crossing number 
problem. Specifically, in an MSO-formula one can quantify over a bounded 
number of pairs of subdivision vertices which should correspond to crossing 
vertices in the planarization of a desired drawing, ensuring that said 
planarization contains no \KuraSub (and in the geometric setting also 
no planarized B- or W-configuration).
Explicitly quantifying over each pair of subdivision vertices which are 
``fused'' into a crossing vertex results in a formula whose length depends
on the number of allowed crossings. This quantification is suitable when the 
number of crossings is a parameter, but it is prohibitive for obtaining an 
FPT-algorithm via Courcelle's theorem if the number of crossings 
is not bounded in the parameter. This difficulty applies to recognizing 
(geometric) \srestricted graphs because a linear number of edges may
participate in crossings.
    
On a high level, we overcome this difficulty by quantifying over an 
unpaired set of subdivision vertices (this can be done using a single variable) 
and deriving an unambiguous pairing of them. 
Such an unambiguous pairing is possible for internally 3-connected graphs and
uses that the endpoints of each pair of crossing edges lie on a 4-cycle
when the type of the crossing belongs to 
\(\mathcal{S}\subseteq\{\fullCt, \almostFullCt, \bowtieCt\}\).
How the internally 3-connectivity and the \(4\)-cycles are used will become clear 
in the proof of \Cref{thm:chorddescriptionalwaysexists}.
The lift from internally 3-connected graphs to the general case follows
from the results of \Cref{sec:3con}, namely \Cref{thm:3con,thm:3con:geom}. 

Deriving a pairing and formulating an appropriate MSO-encoding for the usual,
non-geometric setting is carried out in \Cref{sec:fpt:mso}, where we
conclude with the proof of \Cref{thm:main-algo}. 
In \Cref{sec:fpt:mso-geom} we extend the formula to the geometric setting 
and derive \Cref{thm:geom-algo}.

\subsection{An MSO-Encoding in the usual, topological setting}
\label{sec:fpt:mso}

We first focus on the usual, non-geometric setting for 1-planarity.
Because of \Cref{thm:3con}, it suffices to restrict our attention
to internally 3-connected graphs. Consider a fixed input graph \(G\)  
that is internally 3-connected.

We will show that there is a constant-length MSO-encoding 
of a graph \(G\) being \srestricted on an auxiliary graph \(\tilde{G}\)
derived from \(G\) and whose treewidth is not significantly larger 
than that of \(G\).
Our MSO-encoding will express that there exist some pairs of edges
such that the subgraph of \(G\) induced by each pair is
isomorphic to a graph in \(\mathcal{S}\) and such that, 
if we assume these pairs are all pairs of crossing edges, 
the resulting planarization contains no \KuraSub[s], i.e.\ is planar

Expressing the non-existence of a constant-size obstruction such 
as a \KuraSub in a planarization after explicitly identifying pairs 
of crossing edges is not difficult and has previously been done 
e.g.\ in the fixed-parameter algorithm for computing the crossing 
number~\cite{Grohe04}. To be explicit, before 
applying Courcelle's theorem, one can subdivide each edge with 
a vertex corresponding to a possible crossing vertex in the 
planarization in case that edge is crossed.
Then, given explicit access to pairs of crossing edges, 
the MSO-formula precludes the existence of the constant-size 
obstruction in the planarization by explicitly forbidding vertices 
and edges that can form that obstruction when the original edge 
relation of the graph is replaced by two vertices being endpoints 
of an original edge or an edge subdivision vertex for an edge \(e\) 
and a neighbor of the subdivision vertex for the edge \(e'\) for 
which \(e\) and \(e'\) are a pair of crossing edges.
Since this ingredient is quite standard by now, we will focus the rest
of our disccusion on identifying the pairs of crossings edges.

To build the auxiliary graph \(\tilde{G}\), 
we augment \(G\) in two ways, as follows.
Similarly as in \cite{Grohe04}, for every edge of \(G\) we insert 
a \(P_3\) between its endpoints and label the new middle vertex 
with \role{cr}. In MSO, selecting a set of \role{cr}-vertices encodes 
that the corresponding edges are crossing edges.

However, merely selecting a set of crossing edges is insufficient 
to describe and then check planarity of a prospective planarization.
We also need to identify which pairs of crossing edges cross each other.
To pair up the crossing edges, ideally we would want MSO to be able to 
select for a crossing edge the two end vertices of the edge it crosses, 
e.g., by \(\tilde{G}\) having a vertex between the four end vertices 
of every pair of edges that could possibly cross.
But this is not possible to do in general without \(\tw(\tilde{G})\) 
becoming unbounded in \(\tw(G)\).
Instead, we do the next best thing:
We augment \(\tilde{G}\) in such a way that MSO can select for every 
crossing edge one of the end vertices of the edge it crosses.
Because in every crossing type from \(\{\fullCt,\almostFullCt,\bowtieCt\}\) 
there is a \(4\)-cycle that includes both crossing edges, 
this essentially asks MSO to select for each crossing, 
two paths of length two that together form a \(4\)-cycle.

To do so, we compute a tree decomposition \((T,\chi)\) of \(G\) of 
width \(\mathcal{O}(\tw(G))\)~\cite{MR1417901,BodlaenderDDFLP16} and 
use this to construct a chordal completion \(\hat{G}\) of treewidth 
\(\mathcal{O}(\tw(G))\).
For every \(4\)-tuple of edges \((e,f,e',f')\in E(G)^4\) forming a 
\(4\)-cycle in \(G\) in that order, the chordality of \(\hat{G}\) 
implies that there is at least one \emph{chord} for \((e,f,e',f')\): 
an edge \(g \in E(\hat{G})\setminus\{e,f,e',f'\}\) connecting two 
vertices of \(e\cup f\cup e'\cup f'\).
For each \(4\)-cycle \((e,f,e',f')\) whose vertices induce a graph 
underlying a crossing type in \(\mathcal{S}\), if \(e\) and \(e'\) 
would cross, and each chord \(g\) for \((e,f,e',f')\),
we iteratively add the following graph elements to \(\tilde{G}\):
A \(P_3\) between the endpoints of \(g\) with a \emph{new} middle 
vertex labeled \role{ctp} and an edge between the \role{ctp}-vertex 
and the \role{cr}-vertex that was inserted for \(e\). We call 
\((e,g)\) a \emph{chord-tied pair}, where \(e\) is the 
\emph{crossing edge} and \(g\) the \emph{chord edge} of the 
chord-tied pair. See \Cref{fig:mso:supergraph}, left. Each 
chord-tied pair is uniquely identified by a \role{ctp}-vertex. 
    
\begin{figure}
    \centering
	\includegraphics[page=13,scale=1.1]{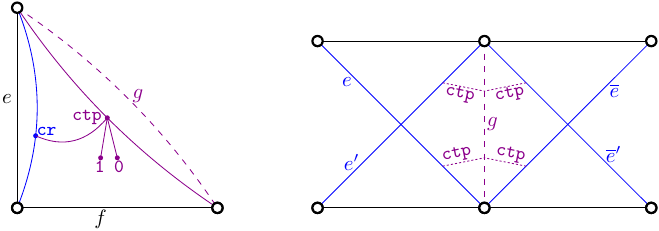}
    \caption{Left: chord-tied pair \((e,g)\) for
		a 4-tuple \((e,f,\cdot,\cdot)\); its flag pendants are
		also displayed. 
		Right: chord edge with two crossings, where the relation
			\role{ctp} is displayed schematically.}
	\label{fig:mso:supergraph}
\end{figure}

At this point, for a given chord-tied pair, there might be multiple 
other chord-tied pairs with the same chord edge. See for example the 
right of \Cref{fig:mso:supergraph}, where there are four chord-tied 
pairs \((e,g)\), \((e',g)\), \((\overline{e},g)\), and \((\overline{e}',g)\) 
with the same chord edge \(g\). This is why we add \emph{flag pendants} 
labeled \role{0}, \role{1}, etc, to each \role{ctp}-vertex in \(\tilde{G}\) 
to encode a flag for every chord-tied pair; see the left side 
of \Cref{fig:mso:supergraph} for an illustration. The flags are used
to make pairings of chord tied pairs with the same chord.
Selecting a pendant vertex is equivalent to selecting a triple 
\((e,g,b)\) where \((e,g)\) is a chord-tied pair and \(b\in \{0,1,\dots\}\) 
is a flag. A pair of edges \(\{e,e'\}\) is a \emph{crossing pair} if \(e\) 
and \(e'\) are the crossing edges of two different, selected chord-tied 
pairs with the same chord edge and the same flag. Equivalently, 
\(\{e,e'\}\) is a crossing pair if two pendant vertices representing
the triples \((e,g,b)\) and \((e',g,b)\) with a common 
chord \(g\) and flag \(b\) are selected.

As we show in \Cref{thm:chorddescriptionalwaysexists}, 
for internally 3-connected \srestricted graphs it suffices to assume 
that at most most two pairs of crossing edges have the same chord edge,
which means that at most four chord-tied pairs share a chord edge.
This means that we only need to model how to make at most two pairs of
crossing edges for each chord, and this justifies that in the construction 
of \(\tilde{G}\) we use precisely \emph{two} pendant flags labeled 
\role{0} and \role{1} per \role{ctp}-vertex. This bound on the number of 
flags is relevant for obtaining a formula of bounded length in MSO.

This addition of two pendant vertices with flags per \role{ctp}-vertex 
completes the construction of \(\tilde{G}\). First we bound its treewidth. 

\begin{lemma}
\label{lem:twbound}
    From a tree decomposition of \(G\) of width \(k\), we get in 
	polynomial time a tree decomposition of \(\tilde{G}\) of width 
	\(\mathcal{O}(k)\) or deduce that \(G\) is not 1-planar.
\end{lemma}
\begin{proof}
    By construction, \(\tilde{G}\) arises from \(\hat{G}\) by inserting 
	\(P_3\)s between endpoints of edges of \(G\), subdividing copies 
	of some chord edges, adding pendant vertices to the corresponding 
	subdivision vertices, and connecting the vertices used in the 
	subdivision for chord-tied pairs.	

	For the chordal completion \(\hat{G}\) of \(G\) we can use the 
	same tree decomposition \((T,\chi)\) as \(G\). This tree 
	decomposition \((T,\chi)\) can in turn be modified into a tree 
	decomposition of \(\tilde{G}\) as follows.
	First insert the \role{cr}-vertices for each edge of \(\hat{G}\) 
	present in \(\chi(t)\) into all bags that contain both vertices
	of that edge. Since a 1-planar graph with \(n\) vertices has at 
	most \(4n-8\) edges~\cite{PachT97,Schumacher86}, in each bag we 
	get \(\mathcal{O}(\tw(G))\) vertices.
    Then, for each chord-tied pair \((e,g)\), there exists a node 
	\(t\) in \(T\) so that the vertices of \(e\cup g\) are contained 
	in the bag \(\chi(t)\). We insert a new node \(t'\) to \(T\), 
	make it adjacent to \(t\), and add to its bag the vertices of 
	\(e\cup g\), the \role{cr}-vertex of \(e\), the new 
	\role{ctp}-vertex of \(g\) and both its flag pendants. 
	This step uses bags of constant size.
	
	Note that if the edge density of the induced subgraph \(G[\chi(t)]\)
	is too large for some node \(t\) of \(T\), we can directly conclude
	that \(G\) is not 1-planar. Otherwise we obtain the desired
	tree decomposition of \(\tilde{G}\).
\end{proof}
  
The next definition allows us to associate pairs of crossing edges
to chord-tied pairs with flags. This is the object we will be searching
for to decide whether an internally 3-connected graph is \srestricted.
It will be crucial that the existence of such an object can be expressed 
using MSO.

\begin{definition}[chord description]
\label{def:cd}
	A set \ChordDescription of chord-tied pairs with flags \textup(i.e., 
	a set of tuples \((e,g,b)\), where \((e,g)\) is a chord-tied 
	pair and \(b\in \{0,1\}\) is a flag\textup) is a 
	\emph{chord description} of \(G\), if it satisfies the following 
	conditions:
    \begin{enumerate}
	\item \label[req]{req:uniqueCTPpartner}
		for every chord-tied pair in \ChordDescription there is 
		exactly one other chord-tied pair in \ChordDescription with 
		the same chord edge and flag; if \((e,g,b)\) and \((e',g,b)\) 
		belong to \ChordDescription, we say that \(\{e,e'\}\) form a 
		\emph{crossing pair of \ChordDescription}; 	 
	\item \label[req]{req:crEdgeInAtMostOneCTP}
		every edge is the crossing edge of at most one chord-tied 
		pair in \ChordDescription; 
	\item \label[req]{req:gCDisPlanar}
		the graph \(G_{\ChordDescription}\), where, starting with \(G\), 
		every crossing pair \(\{e,e'\}\) of \ChordDescription is replaced 
		by a 4-claw on \(V(\{e,e'\})\), is planar; and 
	\item \label[req]{req:onlySrestrCrPairs}
		every crossing pair \(\{e,e'\}\) of \ChordDescription can cross 
		\srestricted[ly], i.e., the vertices in \(e \cup e'\) induce 
		a crossing type from \(\mathcal{S}\) where \(e\) and \(e'\) cross.
	\end{enumerate}
\end{definition}

We can derive a drawing from a chord description as follows.
\begin{lemma}
\label{thm:chorddescription}
	Let \ChordDescription be a chord description of \(G\). Then there is 
	an \srestricted drawing of \(G\) where only crossings pairs of 
	\ChordDescription cross.
\end{lemma}
\begin{proof}
	We start with a planar drawing \(\mcG_{\ChordDescription}\) of 
	\(G_{\ChordDescription}\), which exists by \Cref{req:gCDisPlanar}.
	Then, we un-replace in \(\mcG_{\ChordDescription}\) every 4-claw 
	corresponding to a crossing pair of \ChordDescription by its 
	original edges. This essentially corresponds to undoing a 
	planarization and yields a drawing \mcG of \(G\) where only 
	crossing pairs of \ChordDescription are allowed to cross. 
	(Notice that the un-replacement of 4-claws corresponding to 
	a crossing pair of a \bowtieCt-crossing does not necessarily 
	yield a crossing in \mcG because the endpoints of the edges of the 
	crossing pair might not alternate in the rotation around the center 
	of the 4-claw, but for \(\fullCt\)- and \(\almostFullCt\)-crossings 
	it always does.) By \Cref{req:crEdgeInAtMostOneCTP} and the 
	definition of crossing pairs, each edge is crossed at most 
	once in \mcG. By \Cref{req:onlySrestrCrPairs}, each crossing 
	in \mcG is an \srestricted one. Therefore, \mcG is an 
	\srestricted drawing of \(G\) where only crossing pairs 
	of \ChordDescription may cross.
\end{proof}

Conversely, we can derive a chord description from a hypothetical 
\srestricted drawing of \(G\). Here the internal 3-connectivity 
of \(G\) is important; see \Cref{fig:ambiguous}.
    
\begin{figure}
	\centering
    \includegraphics[page=15]{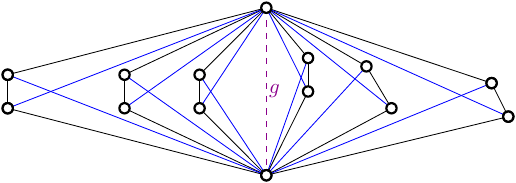}
    \caption{For graphs that are not internally 3-connected, the same
		chord may appear in several chord-tied pairs.}
    \label{fig:ambiguous}
\end{figure}

\begin{lemma}
\label{thm:chorddescriptionalwaysexists}
	Let \(G\) be an internally \(3\)-connected graph that has an \srestricted 
	drawing. Then there exists a chord description \ChordDescription of \(G\).
\end{lemma}
\begin{proof}
	Because of \Cref{lem:3con} and because \(G\) is 2-connected, there 
	is an \srestricted drawing \mcG of \(G\) where no edge incident to 
	a degree-two vertex participates in a crossing.
	For every crossing in \mcG we construct two (different) 
	chord-tied pairs with the same chord edge; such a chord exists
	because a pair of crossing edges is on a 4-cycle which must have 
	a chord in \(\hat{G}\). 
	Let \(\ChordDescription'\) be the family of chord-tied pairs 
	(without flags)	that we obtain this way. Note that we have added
	the chord-tied pairs in groups of two.
	
	We claim that, for each edge \(g\) of \(\hat{G}\), there are at most 
	four chord-tied pairs of the form \((e,g)\) in \(\ChordDescription'\).
	Assume, for the sake of reaching a contradiction, that in
	\(\ChordDescription'\) there are at least six chord-tied pairs with 
	the same chord \(g=uv\). Since we add them in groups of two, this means
	that, for each \(i\in\{1,2,3\}\), we have in \(\ChordDescription'\) 
	the chord-tied pairs \((e_i,g),(e'_i,g)\), where the edges \(e_i\)
	and \(e'_i\) cross in \mcG. Note that the edges 
	\(e_1,e'_1,e_2,e'_2,e_3,e'_3\) are all distinct because each edge
	participates in at most one crossings. See \Cref{fig:mso:3pairs}.
	For each \(i\in\{1,2,3\}\), let \(\gamma_i\) be a simple curve 
	inside \(\mcG(e_i)\cup \mcG(e'_i)\) connecting \(\mcG(u)\) to 
	\(\mcG(v)\). The curves \(\gamma_1,\gamma_2,\gamma_3\) are pairwise
	disjoint but for the common endpoints, and thus split
	the plane into three open regions. At least two of those open regions
	contain vertices of the edges \(e_1,e'_1,e_2,e'_2,e_3,e'_3\).
	Let \(x\) and \(y\) be two such vertices in different regions.
	Because of 1-planarity, each path in \(G\) connecting \(u\) to \(v\) 
	has to pass through \(u\) or \(v\). 
	Because in \mcG no edge incident to a degree-two vertex participates 
	in a crossing, both \(x\) and \(y\) have degree at least three.
	Therefore, \(\{u,v\}\) separates two vertices that have degree
	at least three, and thus \(G\) is not internally 3-connected.
	This finishes the proof of the claim.
	
	\begin{figure}
		\centering
		\includegraphics[page=14]{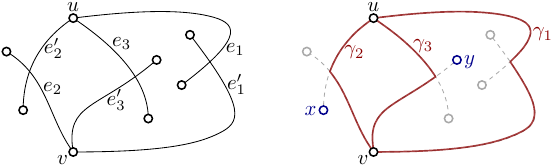}
		\caption{Left: Three pairs of edges crossing in \mcG with 
			the same chord \(g=uv\).
			Right: The curves \(\gamma_1,\gamma_2,\gamma_3\) 
				connecting \(u\) to \(v\) and the vertices \(x\)
				that they separate.}
		\label{fig:mso:3pairs}
	\end{figure}
	
    Since in \(\ChordDescription'\) we have at most four chord-tied 
	pairs with the same chord and we add them in pairs with the same 
	chord, we can choose the flags such that there are 
	at most two chord-tied pairs get the same flag.
	It is clear from the resulting chord description satisfies 
	the four conditions of \Cref{def:cd}.
\end{proof}

\subparagraph*{Building an MSO$_2$-formula.}
The intuition behind the forthcoming formula is to select a subset \(D\) 
of \role{ctp}-vertices, which is equivalent to selecting a subset of 
chord-tied pairs, and select a subset \(M\) of flag pendants\footnote{It is
a reasonable choice to use \(M\) here for the flags because they are 
used to make a \emph{matching} between the chord-tied pairs} such that
for each \role{ctp}-vertex a unique flag \role{0} or \role{1} is selected
and the resulting chord-tied pairs with flags form a chord description. 
Thus, the general structure of the MSO\(_2\)-formula \(\varphi\) 
which we want to apply Courcelle's theorem to will be as follows:
\begin{align} 
	\varphi := \exists D \exists M \ & \forall d \in D \ \role{ctp}(d) \land \forall m \in M \ \role{0}(m) \lor \role{1}(m) \label{line:MSO-overallformula-1}\\
    & \land \forall d \in D \ \exists m \in M \ \big(\mso{adj}{(d,m)} \land \forall m' \in M \ (\mso{adj}{(d,m')} \Rightarrow m = m')\big) \label{line:MSO-overallformula-2}\\
    & \land \forall m \in M \ \exists d \in D \ \mso{adj}{(d,m)} \label{line:MSO-overallformula-3}\\
    & \land \mso{ChordDescription}{(D, M)}, \label{line:MSO-overallformula-4}
\end{align}
where \mso{ChordDescription}{(D,M)} encodes that \(D\) with each chord-tied 
pair corresponding to \(d \in D\) receiving the flag coinciding with the 
label of the unique (by (\ref{line:MSO-overallformula-2})) neighbor 
of \(d\) in \(M\) is a chord description of \(G\).

The encoding of \mso{ChordDescription}{(D,M)} is done formulating 
each of the conditions from \Cref{def:cd}. The description is technical and 
tedious; we use a lemma for each condition from \Cref{def:cd}.

\begin{lemma}
\label{lem:uniqueCTPpartner}
	\Cref{req:uniqueCTPpartner}.\ of \Cref{def:cd} holding for a set of 
	chord-tied pairs arising from a vertex set \(D\) if flagged according 
	to a unique matching to a set \(M\) of flag pendants can be encoded 
	as a constant-length MSO\(_2\)-formula with free variables \(D\) and \(M\).
\end{lemma}
\begin{proof}
	As a building block, we describe a formula \(\mso{ppartner}{(M,d,d')}\) 
	which expresses for distinct \(d,d' \in D\) whether \(d'\) can be a 
	potential partner of \(d\) according to \Cref{req:uniqueCTPpartner}.\
	of \Cref{def:cd}. Formally, 
	\begin{align}
		\mso{ppartner}&{(M,d,d')} := \notag\\
		&\exists u \exists v \ u \neq v \land \neg\role{0}(u) \land \neg\role{1}(u) \land \neg\role{1}(v) \land \neg\role{1}(v)\label{line:ppartner1}\\
        &\qquad\land \mso{adj}{(u,d)} \land \mso{adj}{(v,d)}\label{line:ppartner2}\\
        &\qquad\land \mso{adj}{(u,d')} \land \mso{adj}{(v,d')}\label{line:ppartner3}\\
        &\qquad\land \forall m \in M \ \forall m' \in M \ (\mso{adj}{(d,m)} \land \mso{adj}{(d',m')}) \Rightarrow (\role{0}(m) \Leftrightarrow \role{0}(m'))\label{line:ppartner4}.
	\end{align}
	Obviously, this formula has constant length.
	(\ref{line:ppartner1}-\ref{line:ppartner3}) asserts that \(d\) and \(d'\) 
	have the same chord edge endpoints and hence the same chord edge. 
	(\ref{line:ppartner4}) asserts that \(d\) and \(d'\) have the same 
	matched flag pendant in \(M\).

	We can now obtain the desired formula straightforwardly by using
	\begin{align*}\forall d \in D\ \exists d' \in D & \ d \neq d' \land \mso{ppartner}{(M,d,d')}\\
        & \land \forall d'' \in D \ \mso{ppartner}{(M,d,d'')} \Rightarrow d'' = d \lor d'' = d'. 
	\qedhere
	\end{align*}
\end{proof}

\begin{lemma}
\label{lem:crEdgeInAtMostOneCTP}
	\Cref{req:crEdgeInAtMostOneCTP}.\ of \Cref{def:cd} holding for a set 
	of chord-tied pairs arising from a vertex set \(D\) if flagged according 
	to a unique matching to a set \(M\) of flag pendants can be encoded as 
	a constant-length MSO\(_2\)-formula with free variables \(D\) and \(M\).
\end{lemma}
\begin{proof}
	An appropriate formula is given by the following:
	\begin{align}
		\forall x\  \role{cr}(x)\Rightarrow  
		\big(\forall d \in D\ \forall d' \in D \left( \mso{adj}{(x,d)} \land \mso{adj}{(x,d')}\right) \Rightarrow d = d'\big).
		\label{line:credgeinatmostonectp}
	\end{align}
	This is obviously constant-length.
	(\ref{line:credgeinatmostonectp}) quantifies over \(x\) being the 
	subdivision vertex for an edge and constrains such \(x\) to be in at most 
	one chord-tied pair in \(D\), as desired.
\end{proof}

\begin{lemma}
\label{lem:gCDisPlanar}
	\Cref{req:gCDisPlanar}.\ of \Cref{def:cd} holding for a set of 
	chord-tied pairs arising from a vertex set \(D\) if flagged according 
	to a unique matching to a set \(M\) of flag pendants can be encoded as 
	a constant-length MSO\(_2\)-formula with free variables \(D\) and \(M\).
\end{lemma}
\begin{proof}
	It will be useful to reuse the \mso{ppartner}{(M,d,d')}-formula from the 
	proof of \Cref{lem:uniqueCTPpartner}: it has constant length and 
	expresses that two distinct elements \(d\) and \(d'\) 
	of \(D\) are incident to the same \role{ctp}-labeled vertex, i.e.\ the 
	corresponding chord-tied pairs share a chord edge, and matched to 
	equally-labeled flag pendants in \(M\), i.e.\ the corresponding 
	chord-tied pairs have the same flags. 
	In other words, it expresses that the crossing edges of the 
	chord-tied pairs associated with \(d\) and \(d'\) are a crossing pair.
	We capture this relationship directly for a pair \(x,y\) of 
	\role{cr}-subdivision vertices of crossing edges as follows:
	\begin{align*}
		\mso{crPair}{(D,M,x,y)} := \
			&\role{cr}(x) \land \role{cr}(y)\\
			&\land \exists d \in D\ \exists d' \in D\ 
			   \big( d' \neq d \land \mso{ppartner}{(M,d,d')}\\
			&\phantom{\land \exists d \in D\ \exists d' \in D\ \big(}~ 
				\land \mso{adj}{(x,d)}\land \mso{adj}{(y,d')} \big).
	\end{align*}

	We now need to express the planarity of a graph that is derived from the 
	(subdivided) original graph by replacing (subdivided) edges of a crossing 
	pair by a 4-claw with the crossing pair's endpoints as leaves in MSO.
	Classically, planarity of a graph is MSO-encoded via a constant length 
	formula by precluding the existence of a \KuraSub, i.e.\ the choice of five 
	or six vertex sets and 10 or 9 disjoint edge sets that realize the 
	corresponding internally disjoint paths forming subdivisions of \(K_5\) 
	or \(K_{3,3}\). We modify this formula for our purposes by adapting how 
	we treat edges when speaking about connectivity of vertex sets and edge sets 
	that form paths connecting these vertex sets.
	Firstly, we allow our edge sets for connectivity within vertex sets and paths 
	between vertex sets to only consist of edges, one of whose endpoints is a 
	\role{cr}-vertex and the other of which is not a \role{ctp}-vertex.
	By the construction of \(\tilde{G}\), this implies that we only consider 
	subdivided original edges for \KuraSub[s]. 
	It remains to realize the connections implied by the replacement of edges 
	of a crossing pair by a 4-claw.	We do so by replacing vertex-disjointness 
	conditions on arbitrary edge sets \(X\) and \(X'\), usually written as
	\begin{equation}\label{eq:replacethis}
		\forall e \in X\ \forall e'\in X'\ \nexists v \ 
			\big(\mso{inc}{(v,e)} \land \mso{inc}{(v,e')}\big),
	\end{equation}
	by 
	\begin{equation}\label{eq:replacewiththis}
		\forall e \in X\ \forall e'\in X'\ \nexists v \ 
			\big(\mso{inc}{(v,e)} \land \left(\mso{inc}{(v,e')} \lor \exists v' \ \mso{crPair}{(D,M,v,v')} \land \mso{inc}{(v',e')}\right)\big).
	\end{equation}
	It is obvious that both constraining the edges to be between \role{cr} and 
	non-\role{ctp}-vertices as well as replacing each use of \eqref{eq:replacethis}
	with \eqref{eq:replacewiththis} leaves the planarity-encoding formula constant 
	length. Moreover the graph described on the constrained edge set with the 
	described modifications is precisely the one we want to describe planarity for, 
	as the subdivision vertices of edges in a crossing pair are now identical 
	with each other in terms of considered adjacency and incidence relations, 
	thereby behaving like the center of a four claw whose leaves are the endpoints 
	of the edges of the crossing pair.
\end{proof}

\begin{lemma}
\label{lem:onlySrestrCrPairs}
    \Cref{req:onlySrestrCrPairs}.\ of \Cref{def:cd} holding for a set of 
	chord-tied pairs arising from a vertex set \(D\) if flagged according 
	to a unique matching to a set \(M\) of flag pendants can be encoded as 
	a constant-length MSO\(_2\)-formula with free variables \(D\) and \(M\).        
\end{lemma}
\begin{proof}
	It will be useful to reuse the \mso{crPair}{(D,M,x,y)}-formula from the 
	proof of \Cref{lem:gCDisPlanar}: it has constant length and expresses 
	that two elements \(x\) and \(y\) of \(D\) are the \role{cr} subdivision 
	vertices of the two edges in a crossing pair when \(D\) is labeled 
	according to its matching to \(M\).
	We can now describe the endpoints of these edges as the non-\role{ctp} 
	neighbors of \(x\) and \(y\). By construction of \(\tilde{G}\), 
	this set can only be of size precisely four.

	We can describe the required adjacency relationships for four vertices 
	\(v_1, \dotsc, v_4\) inducing a graph from \(\mathcal{S}\) by an explicit 
	constant-length MSO\(_2\)-formula \mso{induceGraphFromS}{(v_1, \dotsc, v_4)}.

	Putting everything together, a formula as described by the lemma statement 
	hence is:
	\begin{align*}
		\forall x \forall y\ \mso{crPair}{&(D,M,x,y)} \Rightarrow \Big( \forall v_1 \dots \forall v_4 \\ & \big(\bigwedge_{i\in \{1,\dots,4\}} \neg \role{ctp}(v_i)
		 \land \mso{adj}{(v_1,x)}
		  \land \mso{adj}{(v_2,x)}
		 \land \mso{adj}{(v_3,y)} \land \mso{adj}{(v_4,y)}\big)\\
		& \Rightarrow \big( v_1 = v_2 \lor v_3 = v_4 \lor \mso{induceGraphFromS}{(v_1, \dotsc, v_4)}\big)\Big). \qedhere
	\end{align*}
\end{proof}
    
Combining 
\Cref{lem:uniqueCTPpartner,lem:crEdgeInAtMostOneCTP,lem:gCDisPlanar,lem:onlySrestrCrPairs} 
we obtain that \Cref{def:cd} holding for a set \(D\) of chord-tied pairs if 
flagged according to a unique matching to a set \(M\) of flag pendants can be 
encoded as a constant-length MSO\(_2\)-formula with free variables \(D\) and \(M\). 
Together with \Cref{thm:chorddescription,thm:chorddescriptionalwaysexists} 
we obtain the following.
    
\begin{corollary}
\label{cor:varphi}
	There is a constant-length MSO\(_2\)-formula \(\varphi\) such that, if
	\(G\) be an internally \(3\)-connected graph, then 
	\(\tilde{G} \models \varphi\) if and only if \(G\) is \srestricted.
\end{corollary}

Since the formula \(\varphi\) in \Cref{cor:varphi} has constant length and
\(\tw(\tilde{G}) \in \mathcal{O}(\tw(G))\) because of \Cref{lem:twbound}, 
we can apply Courcelle's theorem to conclude that for internally \(3\)-connected graphs
there is an \FPT-algorithm parameterized by treewidth to recognize whether they are 
\(\mathcal{S}\)-restricted 1-planar graphs.
\Cref{thm:main-algo} now follows from \Cref{thm:3con}.

\subsection{An MSO-Encoding in the geometric setting}
\label{sec:fpt:mso-geom}
Now we turn to the geometric setting, that is, recognizing \geosrestricted 
graphs. One difference is that we need to additionally express the 
non-existence of planarized B- and W-configurations in the graph corresponding 
to the planarization of an \srestricted\ drawing.
Further, by the reduction carried out in \Cref{sec:3con:geom} to the case 
of internally \(3\)-connected graphs, 
we need to additionally allow requiring the containment of some distinguished 
vertex on the outer face.

For expressing outer-face requirement, we consider the distinguished vertex 
colored with \role{outer} in the graph when we apply Courcelle's 
theorem, if there is any.

For expressing the non-existence of planarized B- and W-configurations, 
notice that quantifying over all abstract graphs underlying them is trivially 
possible in MSO. To forbid the actual drawings that constitute B- and 
W-configurations, it will be convenient to consider a certain kind of 
well-behaved geometric \srestricted drawings.

\begin{definition}
    A geometric \srestricted drawing \(\mcG\) is \emph{crossing confined} 
	if for each crossing, say between edges \(e\) and \(e'\), no other edge 
	between any pair of vertices among the endpoints of \(e\) and \(e'\) 
	is involved in a crossing. 
\end{definition}

\begin{obs}
\label{obs:oneoptimal-geom}
	If there is a geometric \srestricted drawing of \(G\) then there is 
	also a crossing confined one. Every outer face vertex in the old drawing 
	is also an outer face vertex in the new drawing. 
\end{obs}
\begin{proof}
	Let \mcG be a geometric \srestricted drawing of \(G\) with the minimum
	number of crossings. We claim that \mcG is crossing confined.
	Assume, for the sake of reaching a contradiction, that in \mcG
	the edges \(v_1 v_2\) and \(u_1 u_2\) are mutually crossing 
	and some other edge \(u_1 v_1 \in E(G)\) is crossed in \mcG.
	Because the edges \(v_1 v_2\) and \(u_1 u_2\) cross in \mcG, 
	the vertices  \(v_1\) and \(u_1\) are cofacial in \(\mcG\). 
	We can then redraw the edge \(u_1 v_1\) without a crossing
	within an epsilon distance from \emph{some} pairwise crossing 
	edges incident to \(u_1\) and \(v_1\) (which is not 
	necessarily the pair \(v_1 v_2, u_1 u_2\)), so that it does 
	not introduce new B- or W-configurations. Similar arguments 
	were used in the proofs of \Cref{lem:geom-3con-R,lem:3con-geom}; 
	see the right of \Cref{fig:lem:3con-geom:2} or
	\Cref{fig:lem:geom-3con-R:B}. 
	Since the new drawing has neither B- nor W-configurations,
	there is an equivalent geometric \srestricted drawing \(\mcG'\)
	with fewer crossings than \mcG, and we reach a contradiction.
\end{proof}
    
If we assume crossing confinedness of a targeted drawing, which 
we can do without loss of generality because of \Cref{obs:oneoptimal-geom}, 
then the only way in which the abstract graphs underlying planarized 
B- and W-configurations can occur in a geometric 1-planar drawing is 
if a specific face is not the outer face. We provide such a 
characterization for each configuration using the planarized graph.
\Cref{fig:bwConfig} may be useful to follow the statements.

\begin{lemma}
\label{lem:confidnessB}
	A crossing confined 1-planar drawing \(\mathcal{H}\) of an internally 
	3-connected graph \(H\) contains a B-configuration on edges 
	\(ss'\), \(sb\) and \(s'b'\) in which \(sb\) and \(s'b'\) cross 
	if and only if the crossing between \(sb\) and \(s'b'\) together 
	with \(s\) and \(s'\) are precisely the set of vertices 
	on the boundary of the outer face of 
	\(\mathcal{H}^\times_{s,s',b,b'}:=(\mathcal{H} - 
	\{v \in V(H)\setminus \{b,b'\} \mid N_H(v) = \{s,s'\}\})^\times\).
\end{lemma}
\begin{proof}
	\(\Rightarrow\)
	Because \(\mathcal{H}\) is a 1-planar drawing there are no edges crossing 
	\(sb\) other than \(s'b'\) and vice versa. Because \(\mathcal{H}\) is 
	crossing confined, \(ss'\) is not involved in any crossing. Consider
	the closed curve \(\gamma\) defined \(ss'\), the portion of \(s'b'\) 
	from \(s'\) to its crossing, and the portion from \(sb\) from its 
	crossing to \(s\). The vertices \(b\) and \(b'\) are in the region
	bounded by \(\gamma\). 
	If there is some vertex \(x\) drawn outside \(\gamma\), then it can be 
	connected to \(b\) only by paths through \(s\) or \(s'\). 
	This would contradict the fact that \(H\) is internally 3-connected,
	unless the vertex \(x\) is of degree two and \(N_H(x)=\{s,s'\}\).
	This shows that all the vertices \(x\) outside \(\gamma\) satisfy
	\(N_H(x)=\{s,s'\}\), and therefore the boundary of the outer face of 
	\(\mathcal{H}^\times_{s,s',b,b'}\) is as claimed.
	
	\(\Leftarrow\) 
	By contraposition, if \(s\), \(s'\), and the crossing between \(sb\) 
	and \(s'b'\) are the only vertices on the outer face of 
	\(\mathcal{H}^\times_{s,s',b,b'}\),
	then \(b\) and \(b'\) are drawn inside the curve given by 
	\(ss'\) and \(s'b'\) from \(s'\) to its crossing and \(sb\) from 
	its crossing to \(s\).
	This implies a B-configuration in \(\mathcal{H}\).
\end{proof}

Without even using crossing confinedness (but just the fact that in 
1-planar drawings each edge is involved in at most one crossing), 
one can show a similar statement for W-configurations.

\begin{lemma}
\label{lem:confidnessW}
	A crossing confined 1-planar drawing \(\mathcal{H}\) of an 
	internally 3-connected graph \(H\) contains a W-configuration 
	with spine \(s,s'\) on edges \(sw_1\), \(sw_2\), \(s'w'_1\) 
	and \(s'w'_2\) in which \(sw_1\) and \(sw_2\) cross each other
	and also \(s'w'_1\) and \(s'w'_2\) cross each other if and only 
	if their crossings together with \(s\) and \(s'\) are precisely 
	the set of vertices on the boundary of the outer face of 
	\(\mathcal{H}^\times_{s,s',w_1,w_2,w'_1,w'_2}:=(\mathcal{H} - 
	\{v \in V(H)\setminus \{w_1,w_2,w'_1,w'_2\} \mid N_H(v) = \{s,s'\}\})^\times\).
\end{lemma}
\begin{proof}
	\(\Rightarrow\)
    Because \(\mathcal{H}\) is 1-planar there is no edge crossing 
	\(sw_1\) other than \(s'w'_1\) and vice versa; similarly no edge
	is crossing \(sw_2\) other than \(s'w'_2\) and vice versa.
	Consider the closed curve \(\gamma\) defined by 
	the portion of \(sw_1\) from \(s\) to its crossing with \(sw'_1\), 
	the portion of \(sw'_1\) from the crossing with \(sw_1\) to \(s'\),
	the portion of \(sw'_2\) from \(s'\) to the crossing with \(sw_2\), and
	the portion to \(sw'_1\) from the crossing with \(sw'_2\) to \(s\).
	The vertices \(w_1,w'_1,w_2,w'_2\) are inside the region bounded
	by \(\gamma\).
	If there is some vertex \(x\) drawn outside \(\gamma\), then it can be 
	connected to \(w_1\) only by paths through \(s\) or \(s'\). 
	This would contradict the fact that \(H\) is internally 3-connected,
	unless the vertex \(x\) is of degree two and \(N_H(x)=\{s,s'\}\).
	It follows that the outer face of 
	\(\mathcal{H}^\times_{s,s',w_1,w_2,w'_1,w'_2}\) is as claimed.

	\(\Leftarrow\) 
	By contraposition, if \(s\), \(s'\), the crossing between 
	\(sw_1\) and \(sw_2\), and the crossing between \(s'w'_1\) and \(s'w'_2\) 
	are the only vertices on the outer face of 
	\(\mathcal{H}^\times_{s,s',w_1,w_2,w'_1,w'_2}\), then 
	\(w_1\), \(w_2\), \(w'_1\) and \(w'_2\) are drawn inside the 
	curve along \(s\), \(s'\) and the crossings of \(sw_1\), \(sw_2\), 
	\(s'w'_1\) and \(s'w'_2\).
    This implies a W-configuration in \(\mathcal{H}\).
\end{proof}
    
We explain now how to expand \(\varphi\) from the previous 
subsection for the geometric setting. Here, we use that \(\varphi\) 
was of the form \(\exists D\exists M \tilde\varphi(D,M) \) for an MSO-formula
\(\tilde\varphi\), and add conditions on \(D\) and \(M\) to
obtain \(\varphi_{\text{geom}}\), defined as follows:
\begin{align}
	&\hspace{-.5em}\exists D\exists M \Big(		
		\tilde\varphi(D,M) \land \mso{cr-confined}{(D,M)}\ \land \notag\\
    &~\big(\exists C \ \mso{facecycle}{(C,D,M)}\ \land \notag\\
    &~~~~(\forall v \ \role{outer}(v) \Rightarrow v \in C )\ \land \label{eq:outer}\\
    &~~~~(\forall b_1,\dotsc,b_5,c \ \mso{B-config}{(b_1, \dotsc, b_5,c,D,M)} 
		\Rightarrow C \neq \{b_1,b_2,c\})\ \land \label{eq:B}\\
    &~~~~(\forall b_1,\dotsc,b_6,c_1,c_2 \ \mso{W-config}{(b_1, \dotsc, b_6,c_1,c_2,D,M)} 
		\Rightarrow C \neq \{b_1,b_2,c_1,c_2\})\big) \Big) \label{eq:W}. 
\end{align}
Here, we are using the following elements, where 
\(\ChordDescription=\ChordDescription(D,M)\) is the
flagged chord description of \(G\) arising from the vertex set \(D\) and 
flag pendants \(M\) that satisfy \(\tilde \varphi(D,M)\) and
where \(G_{\ChordDescription}\) is the ``planarization'' described 
by \ChordDescription and defined formally in \Cref{req:gCDisPlanar}.\ of \Cref{def:cd}.)	
\begin{itemize}
\item \mso{cr-confined}{(D,M)} expresses that for each crossing pair \(\{e,e'\}\) 
	of the chord description \(\ChordDescription\) there is no crossing pair 
	of \(\ChordDescription\) containing another edge between any pair of 
	endpoints of \(e\) and \(e'\).
\item \mso{facecycle}{(C,D,M)} expresses that \(C\) is the set of vertices of a face
	in the planarization \(G_{\ChordDescription}\). This is equivalent to
	telling that the graph \( G_{\ChordDescription}+\{cz\mid c\in C\}\) is planar, 
	where \(z\) is a new vertex.
\item \mso{B-config}{(b_1, \dotsc, b_4,c,D,M)} expresses the existence of the 
	underlying graph of a planarized and uncrossed-edge-subdivided B-configuration 
	with crossing vertex \(c\) and uncrossed subdivided edge \(b_1b_2\) 
	in \(G_{\ChordDescription}\).
\item \mso{W-config}{(b_1, \dotsc, b_6,c_1,c_2,D,M)} expresses the existence of 
	the underlying graph of a planarized W-configuration with crossing 
	vertices \(c_1\) and \(c_2\) and shared endpoints of original edges 
	\(b_1\) and \(b_2\) in \(G_{\ChordDescription}\).
\end{itemize}
It is easy to see that the above can be expressed as a constant-length MSO-formula.
The main difficulty lies in encoding the set of crossing pairs of the chord 
description \ChordDescription encoded by \((D,M)\) and the resulting graph
\(G_{\ChordDescription}\), but we argued how to do it in the previous subsection. 

\begin{lemma}
\label{lem:geom-varphi}
	Let \(G\) be an internally \(3\)-connected graph where some vertices may be 
	colored with \role{outer}. Then \(\tilde{G} \models \varphi_{\text{geom}}\) 
	if and only if there is a geometric \srestricted\ drawing of \(G\) 
	in which the \role{outer}-colored vertices, if there are any, lie on the 
	boundary of the outer face.
\end{lemma}
\begin{proof}
	\(\Rightarrow\) By \Cref{cor:varphi}, and since 
	\(\tilde{G} \models \varphi_{\text{geom}}\) implies 
	\(\tilde{G} \models \varphi\), we know that the existentially quantified flagged 
	chord-tied pairs from \(\varphi\) encode an \srestricted drawing of \(G\).
	In addition, let \(C\) make the second part of the conjunction constituting 
	\(\varphi_{\text{geom}}\) true. Because \mso{facecycle}{(C,D,M)} is true,
	\(C\) is the vertex set of a face in the planarization \(G_{\ChordDescription}\),
	and we can choose that face to be the outer face. (For a planar graph,
	we can select any face to be the outer face.) This choice of the outer face 
	results in the \role{outer}-colored vertices being on the boundary of the outer face
	because of \eqref{eq:outer}.
	Because we enforce crossing confinedness of the encoded drawing,
	\Cref{lem:confidnessB,lem:confidnessW} imply that \eqref{eq:B} and \eqref{eq:W}
	exclude the existence of any B- and W-configurations in \(G_{\ChordDescription}\).

	\(\Leftarrow\) Consider a geometric \srestricted (and hence in particular 
	\srestricted) drawing of \(G\) with \role{outer}-colored vertices being on the 
	boundary of the outer face. Because of \Cref{obs:oneoptimal-geom}, 
	we may assume that the drawing is crossing confined.
	Let \(D\) and \(M\) be the flagged chord-tied pairs, which we know that they
	exist because of \Cref{cor:varphi}.
	Let \(C\) be the set of vertices on the outer face of the planarization of that 
	drawing. Add the subdivision vertices of the whole edges that appear the outer 
	face to \(C\).	
	It is straightforward to verify using \Cref{cor:varphi}, crossing confinedness 
	and the characterization of geometric 1-planarity by forbidding 
	B- and W-configurations that setting the existentially quantified variables 
	of \(\varphi_{\text{geom}}\) to \(D\), \(M\) and \(C\) respectively, 
	shows \(\tilde{G} \models \varphi_{\text{geom}}\).
\end{proof}
    
Since the formula \(\varphi\) in \Cref{lem:geom-varphi} has constant length and
\(\tw(\tilde{G}) \in \mathcal{O}(\tw(G))\) because of \Cref{lem:twbound}, 
we can apply Courcelle's theorem to conclude that there is an \FPT-algorithm 
parameterized by treewidth that decides whether internally \(3\)-connected graphs 
are \(O\)-\geosrestricted. 
\Cref{thm:geom-algo} now follows from \Cref{thm:3con:geom}.

\section{Hard Cases for Constant Pathwidth}
\label{sec:hardness}

In this section we show that for the remaining subsets of crossing types, the recognition 
problem becomes \NP-complete in the geometric and non-geometric case, even if the graph 
has constant pathwidth. 

\begin{theorem}
\label{thm:main-npc}
	For \(\mathcal{S} \cap \{\arrowCt, \chairCt, \XCt\} \neq \emptyset\) deciding 
	\geosrestricted[ity] and \srestricted[ity] is \NP-complete and para-\NP-hard when 
	parameterized by the \tpWidth.
\end{theorem}

\NP\ membership follows due to the fact that geometric 1-planar graphs are characterized 
by having drawings without B- and W-configurations. For \NP-hardness, we show for every 
singleton \(\mathcal{S}' \subset \{\arrowCt, \chairCt, \XCt\}\) how to transform any 
\threePartition instance \(I\) into a graph with constant  \tpWidth that is 
\geosprimerestricted if \(I\) is a yes-instance and not even 1-planar if \(I\) is a 
no-instance. This then directly yields a reduction from \threePartition to recognizing 
\srestricted[ity] and \geosrestricted[ity] for any \(\mathcal{S}\) with 
\(\mathcal{S} \cap \{\arrowCt, \chairCt, \XCt\}\ne \emptyset\). 

Our transformation builds upon Grigoriev and Bodlaender's proof that determining 
1-planarity is NP-complete \cite{GrigorievB07}. We provide a standalone presentation.

The problem \threePartition is given a set \(A\) of \(3m\) elements with \(m \geq 3\), 
a bound \(B\in \mathbb{N}\), and a size \(s(a) \in \mathbb{N}\) for each \(a\in A\) 
such that \(\sum_{a\in A} s(a) = mB\), and asks whether \(A\) can be partitioned 
into \(m\) disjoint subsets \(A_1,\dots,A_m\) of size three such that for 
\(1 \leq i \leq m\) we have \(\sum_{a\in A_i} s(a) = B\).  As \threePartition is 
strongly NP-complete, we can assume that \(B\) is polynomial in \(m\). 
For convenience, we furthermore assume that both \(B\) and \(m\) are even.

\begin{figure}
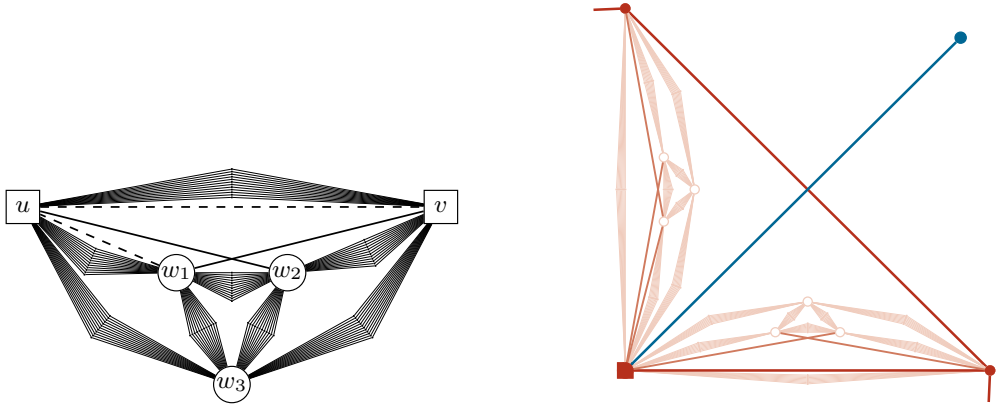

\centering
	\vspace{-.5em}
	\begin{subfigure}[t]{.4\textwidth}
		\centering
		\scalebox{1.1}{\fenceFig}
	\end{subfigure}
	\hfill
	\begin{subfigure}[t]{.5\textwidth}
		\centering
		\scalebox{1.2}{\closeUpFig}
	\end{subfigure}
	\caption{Left: Fence \(\fence\) between nodes \(u\) and \(v\). The single edges 
		are thick; the direct edges whose presence depends on \(\mathcal{S}'\) are dashed. 
		Right: Example of how a splitter edge (blue) crosses a rim edge. All three 
		crossings are \chairCt-crossings.}
	\label{fig:fence}
\end{figure}
        
Before we describe the transformation itself, we construct what we call \emph{fences}.
For any two nodes \(u\) and \(v\), the fence \fence between them is constructed as 
follows (see the left side of \Cref{fig:fence} for reference):
We start with a \(K_5\) on \(u\), \(v\) and three new nodes \(w_1\), \(w_2\), \(w_3\).
Next we replace all of its edges other than \(uw_2\) and \(vw_1\) by a \emph{bundle} 
of \(\ell = 12\) parallel paths of length two. We call \(uw_2\) and \(vw_1\) the 
\emph{single edges} of the fence. For \(\mathcal{S}' = \{\arrowCt\}\) we re-add the 
edges \(uv\) and \(uw_1\), and for \(\mathcal{S}' = \{\chairCt\}\) we re-add either 
the edge \(uv\) or \(uw_1\) (details follow later).
This ensures that the crossing between the single edges is an \(\mathcal{S}'\)-crossing.
The drawings depicted in the left side of \Cref{fig:fence} shows a \geosrestricted 
drawing of a fence where \(u\) and \(v\) share a face. The following lemma shows that 
we can treat fences as uncrossable edges in the context of 1-planar drawings.

\begin{lemma}
\label{lem:fence}
	Let \(G\) be a 1-planar graph and \(u,v \in V(G)\) be two nodes between 
	which there is a fence \(\fence \subseteq G\). Then, in every 1-planar drawing 
	of \(G\) the single edges of \fence cross each other \textup(and whereby \(u\) 
	and \(v\) share a face\textup).
\end{lemma}
\begin{proof}
	Let \(\mathcal{K}\) be the set of \KuraSub[s] with \(\{u,v,w_1,w_2,w_3\}\) as 
	the Kuratowski nodes. Consider some 1-planar drawing of \fence and a crossing 
	involving only one of the single edges. The other edge has to be from one of the 
	bundles which replace a \(K_5\) edge \(e\).
	For every \KuraSub that is resolved by this crossing%
	\footnote{i.e.\ the crossing is between edges from non-adjacent Kuratowski paths 
	of the \KuraSub}, there are \(\ell -1\) corresponding \KuraSub[s] that differ only 
	in the path from the bundle replacing~\(e\). Therefore, such a crossing resolves 
	at most \(\frac{1}{\ell}\) of all the subdivisions in \(\mathcal{K}\). By a similar 
	argument, each crossing involving none of the single edges, including $uw_1$ and $uv$, 
	can resolve at most \(\frac{1}{\ell^2}\) of \(\mathcal{K}\). There are 
	\(8 \cdot 2 \cdot \ell\) non-single edges in the bundles, at most two additional 
	non-single direct edges, and two single edges. Thus, at most 
	\(2 \cdot \frac{1}{\ell} + \frac{2\cdot(8 \ell+1)}{2} \cdot \frac{1}{\ell^2} = 
	\frac{10 \ell+1}{\ell^2} \leq \frac{11}{12} < 1\) of \(\mathcal{K}\) is resolved 
	by crossings involving at most one of the single edges.
	Therefore, the single edges cross in every 1-planar drawing of \fence.
	As one of the single edges of \fence is incident to \(u\) and the other to \(v\) 
	and the single edges cross exactly each other, \(u\) and \(v\) are cofacial.
\end{proof}

\begin{figure}
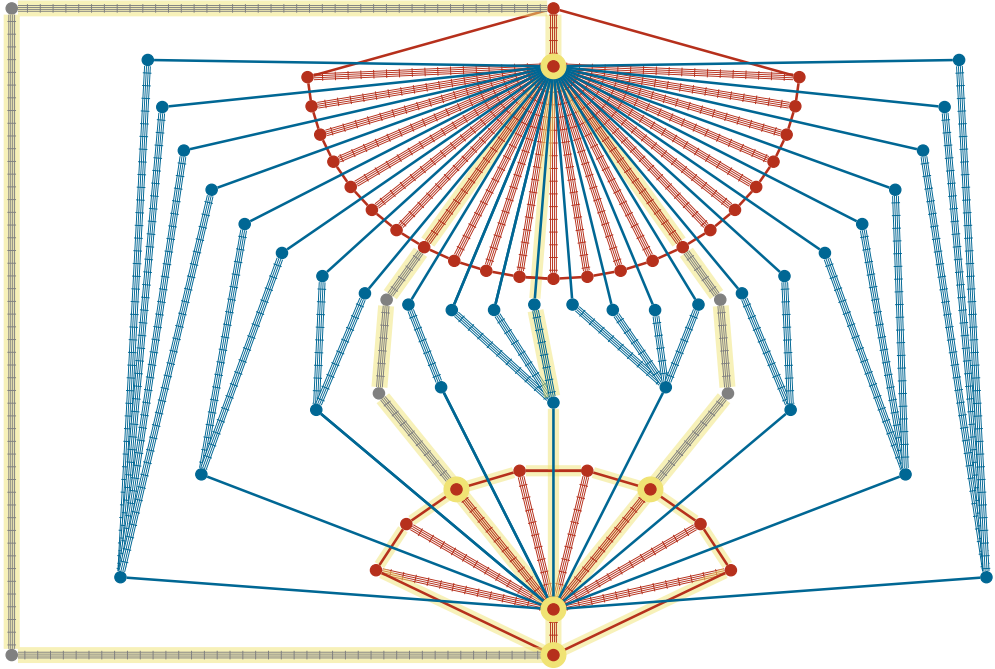

\centering
    \scalebox{1.2}{\instanceFig}
    \caption{Graph \(G_I\) based on \threePartition instance \(I\) with 
		\(A = \{1,2,2,2,3,3,3,4,4\}\), \(B=8\), and \(m=3\). The fences have a fence 
		pattern, the wheels are red, the splitters are blue, and the dividers are gray. 
		A \KuraSub of interest is highlighted in yellow.}
    \label{fig:instance}
\end{figure}

Given a \threePartition instance \(I\) we construct a graph \(G_I\) as follows; 
see \Cref{fig:instance} for an example. We start with two wheels with fences as radians; 
one of size \(3m\), the \emph{transmitter}, and one of size \(Bm\), the \emph{collector}.
For \(\mathcal{S}' = \{\chairCt\}\), every even-numbered radian fence has 
the edge \(uv\) and every odd-numbered has \(uw_1\), where in all cases \(u\)
is the center of the wheel. (Here we use that that \(m\) and \(B\) are even.)
The \emph{rim} edges of each of the wheels are those that are not incident to the
center of the wheel. Note that in a drawing, if an edge incident to the center of 
the wheel crosses a rim edge, then such a crossing is an \(\mathcal{S}'\)-crossing;
see \Cref{fig:instance}, right.
Next, we add a path of three fences 
between the \(3i\)-th node of the transmitter and the \(Bi\)-th node of the collector 
for \(1 \leq i \leq m\). We call the resulting \(m\) paths of fences between 
the wheel centers \emph{dividers}. For each element \(a\in A\) we add a 
\emph{splitter} \(Q_a\), which consists of a claw of size \(s(a)\) of fences and an edge 
between every degree-1 node of the claw and the collector center and 
an edge between the degree-\(s(a)\) node of the claw and the transmitter center. 
The next two lemmas show that \(I\) can be satisfyingly partitioned 
if and only if \(G_I\) is \geosprimerestricted.

\begin{lemma}
\label{lem:npc-forward}
    If \(I\) can be satisfyingly partitioned, then \(G_I\) is \geosprimerestricted.
\end{lemma}
\begin{proof}
	See \Cref{fig:instance} for an example of a \geosprimerestricted drawing that 
	could be the result of the following construction. Let \(A_1, \dots, A_m\) be the 
	satisfying partition. First, we draw \(G_I\) without the splitters 1-planarly 
	such that the wheels are drawn disjointly. In each of the resulting \(m\) distinct 
	faces between the wheels, we draw the claws corresponding to the three elements 
	of one of the triples in the partition.
    As each such face has \(B\) rim edges of the collector and \(3\) rim edges of 
	the transmitter, we can then draw the remaining splitter edges such that each 
	rim edge is crossed by exactly one splitter edge. By the construction of the radians, 
	specifically whether the direct edge \(uv\) was re-added therein, these crossings 
	are \(\mathcal{S}'\)-crossings. See the right side of \Cref{fig:fence} for an 
	example of a \geosprimerestricted subdrawing of such a crossing with 
	\(\mathcal{S}' = \{\chairCt\}\).
    Clearly, the resulting drawing does not contain any B- and W-configuration, 
	and thus there is an equivalent \geosprimerestricted drawing of the graph \(G_I\).
\end{proof}

For the other direction, we only require \(G_I\) to be 1-planar.
\begin{lemma}
\label{lem:npc-backward}
    If \(G_I\) is 1-planar then \(I\) can be satisfyingly partitioned.
\end{lemma}
\begin{proof}
    Let \mcG be a 1-planar drawing of \(G_I\). By \Cref{lem:fence}, we can treat fences 
	as uncrossable edges in this proof.	Consider a $K_5$ \KuraSub where the nodes 
	are the wheel centers and three nodes that are both on the transmitter 
	wheel and an endpoint of a divider.	Three paths of the $K_5$ \KuraSub
	are on the transmitter wheel, and two paths each are on four disjoint paths 
	between the wheel centers: three through the three dividers and one through a 
	splitter \(Q_a\). In \Cref{fig:instance} such a \KuraSub is highlighted.
    The only non-fence, and thus 1-planarly crossable, edges in this \KuraSub 
	are the transmitter rim edges and two edges in the splitter path.
    As the transmitter rim edges are in adjacent Kuratowski paths, the crossing 
	that resolves this \KuraSub in \mcG has to be between a non-fence edge of the 
	splitter and a transmitter rim edge.
    As the transmitter rim consists of \(3m\) edges and there are multiple 
	such \KuraSub[s] for each of the \(3m\) elements of \(A\), we know that 
	in \mcG each transmitter rim edge is crossed by a splitter edge that is 
	adjacent to the transmitter center. By a similar argument, each of the \(Bm\) 
	collector rim edges is crossed by one of the \(Bm\) splitter edges adjacent 
	to the collector center.
	
    There are no other non-fence edges, so the only other crossings in \mcG are 
	the single crossing inside each fence. The embedding of the transmitter and 
	collector wheel and the dividers is unique (up to the internal embedding of 
	the fences) because it is \(3\)-connected and planar, when considering each 
	fence as a single edge. The exterior of the transmitter and collector wheels 
	is divided by the dividers into \(m\) regions. Each such region has three 
	transmitter and \(B\) collector rim edges on its border. 
	Thus, there have to be three splitters inside each such region and the 
	sum of the sizes of their corresponding elements has to be exactly \(B\), 
	yielding a satisfying partition of \(A\).
\end{proof}

Now we bound the \tpWidth of \(G_I\). 

\begin{lemma}
\label{lem:npc-treewidth}
    The constructed graph \(G_I\) has constant \tpWidth.
\end{lemma}
\begin{proof}
    Firstly, we replace each fence \fence\ by the single edge \(uv\); since each 
	fence has constant size, this changes the \tpWidth by a constant factor. Next, 
	we remove the two wheel centers; since we remove two vertices, this changes 
	the \tpWidth by at most two.
    The remaining graph consists of a disconnected set of claws and two cycles 
	connected by some \emph{connecting edges} that form a planar graph: the
	cyclic order of the endpoints of the connecting edges is consistent on
	both cycles. Clearly, the claws have constant \tpWidth. Lastly, the two 
	cycles with their connecting edges also have constant \tpWidth: removing 
	the two vertices of some arbitrary connecting edge, we are left with 
	a subdivision of an \(m \times 2\) grid graph that has constant \tpWidth.
\end{proof}

It is obvious that the graph \(G_I\) can be constructed in polynomial time
because \threePartition is strongly NP-complete: since \(B\) is bounded
by a polynomial in \(m\), the size of \(G_I\) is also bounded by a polynomial
in \(m\).
\Cref{thm:main-npc} immediately follows from 
\Cref{lem:npc-forward,lem:npc-backward,lem:npc-treewidth}.

\section{Conclusions}
We have shown the following dichotomy for recognizing 1-planar graphs with restricted
types of crossings:
\begin{itemize}
\item If \(\mathcal{S} \subseteq \{\fullCt,\almostFullCt,\bowtieCt\}\),
  it is \FPT\ parameterized by the treewidth to recognize \(\mathcal{S}\)-restricted 
  1-planar graphs.
\item If \(\mathcal{S} \cap \{\arrowCt,\chairCt,\XCt\} \neq \emptyset\), then it is 
  \NP-hard to recognize 1-planar graphs that are \srestricted, even for graphs
  that have the treewidth bounded by constant.
\end{itemize}
The same dichotomy carries over to the recognition of 
\geosrestricted graphs.

Our positive algorithms, parameterized by the treewidth, rely on Courcelle's theorem and 
are thus mainly of theoretical interest. It would be interesting to obtain algorithms 
that would be closer to becoming practical. 

In fact, it is unclear whether the parametrization by the treewidth is needed to obtain 
the positive results. Brandenburg~\cite{Brandenburg15,Brandenburg19} has shown that 
\(\{\fullCt\}\)-restricted 1-planar graphs are recognizable in polynomial time.
For the other cases of nonempty \(\mathcal{S} \subseteq \{\fullCt,\almostFullCt,\bowtieCt\}\),
the recognition of \(\mathcal{S}\)-restricted 1-planar graphs with unbounded treewidth
is open. Since our reductions to the case of internally 3-connected graphs 
(see \Cref{sec:3con}) do not assume anything about the treewidth, 
it suffices to consider the case of internally 3-connected graphs. 
We leave this problem, and the variant with straight-line drawings, open for future research.

Another direction for future research would be to consider 2-planar graphs, that is,
graphs that have a drawing where each edge participates in at most two crossings.
A similar characterization for 2-planar graphs with restricted types of crossings
seems very challenging.

\FloatBarrier
\bibliographystyle{plainurl}
\bibliography{references}
\end{document}

%% file: tikzGenerated.tex
\usetikzlibrary {patterns,spy}
\newcommand{\fenceFig}{
	\begin{tikzpicture}[]
		\node[draw,rectangle,,inner sep=.15mm, minimum size=4mm] (u) at (-2.5, 0) {\footnotesize $u$};
		\node[draw,rectangle,,inner sep=.15mm, minimum size=4mm] (v) at (2.5, 0) {\footnotesize $v$};
		\node[draw,circle,,inner sep=.15mm, minimum size=4mm] (w1) at ($(u)!0.5!(v)!0.4137931034482758!230:(v)$) {\footnotesize $w_1$};
		\node[draw,circle,,inner sep=.15mm, minimum size=4mm] (w2) at ($(u)!0.5!(v)!0.4137931034482758!310:(v)$) {\footnotesize $w_2$};
		\node[draw,circle,,inner sep=.15mm, minimum size=4mm] (w3) at ($(u)!0.5!(v)!0.8561236623067774!270:(v)$) {\footnotesize $w_3$};
		\node[draw,circle,fill,inner sep=0mm, minimum size=0.1mm, ] (u-v-0) at ($(u)!.5!(v)!0.11557669441141495cm!90:(v)$) {};
		\path[draw,very thin, ] (u) to (u-v-0) to (v);
		\node[draw,circle,fill,inner sep=0mm, minimum size=0.1mm, ] (u-v-1) at ($(u)!.5!(v)!0.14661117717003563cm!90:(v)$) {};
		\path[draw,very thin, ] (u) to (u-v-1) to (v);
		\node[draw,circle,fill,inner sep=0mm, minimum size=0.1mm, ] (u-v-2) at ($(u)!.5!(v)!0.17764565992865633cm!90:(v)$) {};
		\path[draw,very thin, ] (u) to (u-v-2) to (v);
		\node[draw,circle,fill,inner sep=0mm, minimum size=0.1mm, ] (u-v-3) at ($(u)!.5!(v)!0.208680142687277cm!90:(v)$) {};
		\path[draw,very thin, ] (u) to (u-v-3) to (v);
		\node[draw,circle,fill,inner sep=0mm, minimum size=0.1mm, ] (u-v-4) at ($(u)!.5!(v)!0.2397146254458977cm!90:(v)$) {};
		\path[draw,very thin, ] (u) to (u-v-4) to (v);
		\node[draw,circle,fill,inner sep=0mm, minimum size=0.1mm, ] (u-v-5) at ($(u)!.5!(v)!0.27074910820451836cm!90:(v)$) {};
		\path[draw,very thin, ] (u) to (u-v-5) to (v);
		\node[draw,circle,fill,inner sep=0mm, minimum size=0.1mm, ] (u-v-6) at ($(u)!.5!(v)!0.30178359096313906cm!90:(v)$) {};
		\path[draw,very thin, ] (u) to (u-v-6) to (v);
		\node[draw,circle,fill,inner sep=0mm, minimum size=0.1mm, ] (u-v-7) at ($(u)!.5!(v)!0.33281807372175976cm!90:(v)$) {};
		\path[draw,very thin, ] (u) to (u-v-7) to (v);
		\node[draw,circle,fill,inner sep=0mm, minimum size=0.1mm, ] (u-v-8) at ($(u)!.5!(v)!0.36385255648038045cm!90:(v)$) {};
		\path[draw,very thin, ] (u) to (u-v-8) to (v);
		\node[draw,circle,fill,inner sep=0mm, minimum size=0.1mm, ] (u-v-9) at ($(u)!.5!(v)!0.3948870392390011cm!90:(v)$) {};
		\path[draw,very thin, ] (u) to (u-v-9) to (v);
		\node[draw,circle,fill,inner sep=0mm, minimum size=0.1mm, ] (u-v-10) at ($(u)!.5!(v)!0.4259215219976218cm!90:(v)$) {};
		\path[draw,very thin, ] (u) to (u-v-10) to (v);
		\node[draw,circle,fill,inner sep=0mm, minimum size=0.1mm, ] (u-v-11) at ($(u)!.5!(v)!0.4569560047562425cm!90:(v)$) {};
		\path[draw,very thin, ] (u) to (u-v-11) to (v);
		\path[draw, semithick, , dashed, bend right=0.0] (u)  to (v);
		\node[draw,circle,fill,inner sep=0mm, minimum size=0.1mm, ] (w3-u-0) at ($(w3)!.5!(u)!0.32104637336504155cm!90:(u)$) {};
		\path[draw,very thin, ] (w3) to (w3-u-0) to (u);
		\node[draw,circle,fill,inner sep=0mm, minimum size=0.1mm, ] (w3-u-1) at ($(w3)!.5!(u)!0.35208085612366224cm!90:(u)$) {};
		\path[draw,very thin, ] (w3) to (w3-u-1) to (u);
		\node[draw,circle,fill,inner sep=0mm, minimum size=0.1mm, ] (w3-u-2) at ($(w3)!.5!(u)!0.3831153388822829cm!90:(u)$) {};
		\path[draw,very thin, ] (w3) to (w3-u-2) to (u);
		\node[draw,circle,fill,inner sep=0mm, minimum size=0.1mm, ] (w3-u-3) at ($(w3)!.5!(u)!0.4141498216409035cm!90:(u)$) {};
		\path[draw,very thin, ] (w3) to (w3-u-3) to (u);
		\node[draw,circle,fill,inner sep=0mm, minimum size=0.1mm, ] (w3-u-4) at ($(w3)!.5!(u)!0.4451843043995243cm!90:(u)$) {};
		\path[draw,very thin, ] (w3) to (w3-u-4) to (u);
		\node[draw,circle,fill,inner sep=0mm, minimum size=0.1mm, ] (w3-u-5) at ($(w3)!.5!(u)!0.476218787158145cm!90:(u)$) {};
		\path[draw,very thin, ] (w3) to (w3-u-5) to (u);
		\node[draw,circle,fill,inner sep=0mm, minimum size=0.1mm, ] (w3-u-6) at ($(w3)!.5!(u)!0.5072532699167657cm!90:(u)$) {};
		\path[draw,very thin, ] (w3) to (w3-u-6) to (u);
		\node[draw,circle,fill,inner sep=0mm, minimum size=0.1mm, ] (w3-u-7) at ($(w3)!.5!(u)!0.5382877526753863cm!90:(u)$) {};
		\path[draw,very thin, ] (w3) to (w3-u-7) to (u);
		\node[draw,circle,fill,inner sep=0mm, minimum size=0.1mm, ] (w3-u-8) at ($(w3)!.5!(u)!0.569322235434007cm!90:(u)$) {};
		\path[draw,very thin, ] (w3) to (w3-u-8) to (u);
		\node[draw,circle,fill,inner sep=0mm, minimum size=0.1mm, ] (w3-u-9) at ($(w3)!.5!(u)!0.6003567181926277cm!90:(u)$) {};
		\path[draw,very thin, ] (w3) to (w3-u-9) to (u);
		\node[draw,circle,fill,inner sep=0mm, minimum size=0.1mm, ] (w3-u-10) at ($(w3)!.5!(u)!0.6313912009512485cm!90:(u)$) {};
		\path[draw,very thin, ] (w3) to (w3-u-10) to (u);
		\node[draw,circle,fill,inner sep=0mm, minimum size=0.1mm, ] (w3-u-11) at ($(w3)!.5!(u)!0.662425683709869cm!90:(u)$) {};
		\path[draw,very thin, ] (w3) to (w3-u-11) to (u);
		\node[draw,circle,fill,inner sep=0mm, minimum size=0.1mm, ] (v-w3-0) at ($(v)!.5!(w3)!0.32104637336504155cm!90:(w3)$) {};
		\path[draw,very thin, ] (v) to (v-w3-0) to (w3);
		\node[draw,circle,fill,inner sep=0mm, minimum size=0.1mm, ] (v-w3-1) at ($(v)!.5!(w3)!0.35208085612366224cm!90:(w3)$) {};
		\path[draw,very thin, ] (v) to (v-w3-1) to (w3);
		\node[draw,circle,fill,inner sep=0mm, minimum size=0.1mm, ] (v-w3-2) at ($(v)!.5!(w3)!0.3831153388822829cm!90:(w3)$) {};
		\path[draw,very thin, ] (v) to (v-w3-2) to (w3);
		\node[draw,circle,fill,inner sep=0mm, minimum size=0.1mm, ] (v-w3-3) at ($(v)!.5!(w3)!0.4141498216409035cm!90:(w3)$) {};
		\path[draw,very thin, ] (v) to (v-w3-3) to (w3);
		\node[draw,circle,fill,inner sep=0mm, minimum size=0.1mm, ] (v-w3-4) at ($(v)!.5!(w3)!0.4451843043995243cm!90:(w3)$) {};
		\path[draw,very thin, ] (v) to (v-w3-4) to (w3);
		\node[draw,circle,fill,inner sep=0mm, minimum size=0.1mm, ] (v-w3-5) at ($(v)!.5!(w3)!0.476218787158145cm!90:(w3)$) {};
		\path[draw,very thin, ] (v) to (v-w3-5) to (w3);
		\node[draw,circle,fill,inner sep=0mm, minimum size=0.1mm, ] (v-w3-6) at ($(v)!.5!(w3)!0.5072532699167657cm!90:(w3)$) {};
		\path[draw,very thin, ] (v) to (v-w3-6) to (w3);
		\node[draw,circle,fill,inner sep=0mm, minimum size=0.1mm, ] (v-w3-7) at ($(v)!.5!(w3)!0.5382877526753863cm!90:(w3)$) {};
		\path[draw,very thin, ] (v) to (v-w3-7) to (w3);
		\node[draw,circle,fill,inner sep=0mm, minimum size=0.1mm, ] (v-w3-8) at ($(v)!.5!(w3)!0.569322235434007cm!90:(w3)$) {};
		\path[draw,very thin, ] (v) to (v-w3-8) to (w3);
		\node[draw,circle,fill,inner sep=0mm, minimum size=0.1mm, ] (v-w3-9) at ($(v)!.5!(w3)!0.6003567181926277cm!90:(w3)$) {};
		\path[draw,very thin, ] (v) to (v-w3-9) to (w3);
		\node[draw,circle,fill,inner sep=0mm, minimum size=0.1mm, ] (v-w3-10) at ($(v)!.5!(w3)!0.6313912009512485cm!90:(w3)$) {};
		\path[draw,very thin, ] (v) to (v-w3-10) to (w3);
		\node[draw,circle,fill,inner sep=0mm, minimum size=0.1mm, ] (v-w3-11) at ($(v)!.5!(w3)!0.662425683709869cm!90:(w3)$) {};
		\path[draw,very thin, ] (v) to (v-w3-11) to (w3);
		\node[draw,circle,fill,inner sep=0mm, minimum size=0.1mm, ] (w1-u-0) at ($(w1)!.5!(u)!0.11557669441141495cm!90:(u)$) {};
		\path[draw,very thin, ] (w1) to (w1-u-0) to (u);
		\node[draw,circle,fill,inner sep=0mm, minimum size=0.1mm, ] (w1-u-1) at ($(w1)!.5!(u)!0.14661117717003563cm!90:(u)$) {};
		\path[draw,very thin, ] (w1) to (w1-u-1) to (u);
		\node[draw,circle,fill,inner sep=0mm, minimum size=0.1mm, ] (w1-u-2) at ($(w1)!.5!(u)!0.17764565992865633cm!90:(u)$) {};
		\path[draw,very thin, ] (w1) to (w1-u-2) to (u);
		\node[draw,circle,fill,inner sep=0mm, minimum size=0.1mm, ] (w1-u-3) at ($(w1)!.5!(u)!0.208680142687277cm!90:(u)$) {};
		\path[draw,very thin, ] (w1) to (w1-u-3) to (u);
		\node[draw,circle,fill,inner sep=0mm, minimum size=0.1mm, ] (w1-u-4) at ($(w1)!.5!(u)!0.2397146254458977cm!90:(u)$) {};
		\path[draw,very thin, ] (w1) to (w1-u-4) to (u);
		\node[draw,circle,fill,inner sep=0mm, minimum size=0.1mm, ] (w1-u-5) at ($(w1)!.5!(u)!0.27074910820451836cm!90:(u)$) {};
		\path[draw,very thin, ] (w1) to (w1-u-5) to (u);
		\node[draw,circle,fill,inner sep=0mm, minimum size=0.1mm, ] (w1-u-6) at ($(w1)!.5!(u)!0.30178359096313906cm!90:(u)$) {};
		\path[draw,very thin, ] (w1) to (w1-u-6) to (u);
		\node[draw,circle,fill,inner sep=0mm, minimum size=0.1mm, ] (w1-u-7) at ($(w1)!.5!(u)!0.33281807372175976cm!90:(u)$) {};
		\path[draw,very thin, ] (w1) to (w1-u-7) to (u);
		\node[draw,circle,fill,inner sep=0mm, minimum size=0.1mm, ] (w1-u-8) at ($(w1)!.5!(u)!0.36385255648038045cm!90:(u)$) {};
		\path[draw,very thin, ] (w1) to (w1-u-8) to (u);
		\node[draw,circle,fill,inner sep=0mm, minimum size=0.1mm, ] (w1-u-9) at ($(w1)!.5!(u)!0.3948870392390011cm!90:(u)$) {};
		\path[draw,very thin, ] (w1) to (w1-u-9) to (u);
		\node[draw,circle,fill,inner sep=0mm, minimum size=0.1mm, ] (w1-u-10) at ($(w1)!.5!(u)!0.4259215219976218cm!90:(u)$) {};
		\path[draw,very thin, ] (w1) to (w1-u-10) to (u);
		\node[draw,circle,fill,inner sep=0mm, minimum size=0.1mm, ] (w1-u-11) at ($(w1)!.5!(u)!0.4569560047562425cm!90:(u)$) {};
		\path[draw,very thin, ] (w1) to (w1-u-11) to (u);
		\path[draw, semithick, , dashed, bend right=0.0] (w1)  to (u);
		\node[draw,circle,fill,inner sep=0mm, minimum size=0.1mm, ] (w1-w3-0) at ($(w1)!.5!(w3)!-0.1862068965517241cm!90:(w3)$) {};
		\path[draw,very thin, ] (w1) to (w1-w3-0) to (w3);
		\node[draw,circle,fill,inner sep=0mm, minimum size=0.1mm, ] (w1-w3-1) at ($(w1)!.5!(w3)!-0.15517241379310343cm!90:(w3)$) {};
		\path[draw,very thin, ] (w1) to (w1-w3-1) to (w3);
		\node[draw,circle,fill,inner sep=0mm, minimum size=0.1mm, ] (w1-w3-2) at ($(w1)!.5!(w3)!-0.12413793103448273cm!90:(w3)$) {};
		\path[draw,very thin, ] (w1) to (w1-w3-2) to (w3);
		\node[draw,circle,fill,inner sep=0mm, minimum size=0.1mm, ] (w1-w3-3) at ($(w1)!.5!(w3)!-0.09310344827586205cm!90:(w3)$) {};
		\path[draw,very thin, ] (w1) to (w1-w3-3) to (w3);
		\node[draw,circle,fill,inner sep=0mm, minimum size=0.1mm, ] (w1-w3-4) at ($(w1)!.5!(w3)!-0.062068965517241365cm!90:(w3)$) {};
		\path[draw,very thin, ] (w1) to (w1-w3-4) to (w3);
		\node[draw,circle,fill,inner sep=0mm, minimum size=0.1mm, ] (w1-w3-5) at ($(w1)!.5!(w3)!-0.031034482758620682cm!90:(w3)$) {};
		\path[draw,very thin, ] (w1) to (w1-w3-5) to (w3);
		\node[draw,circle,fill,inner sep=0mm, minimum size=0.1mm, ] (w1-w3-6) at ($(w1)!.5!(w3)!0.0cm!90:(w3)$) {};
		\path[draw,very thin, ] (w1) to (w1-w3-6) to (w3);
		\node[draw,circle,fill,inner sep=0mm, minimum size=0.1mm, ] (w1-w3-7) at ($(w1)!.5!(w3)!0.031034482758620682cm!90:(w3)$) {};
		\path[draw,very thin, ] (w1) to (w1-w3-7) to (w3);
		\node[draw,circle,fill,inner sep=0mm, minimum size=0.1mm, ] (w1-w3-8) at ($(w1)!.5!(w3)!0.062068965517241365cm!90:(w3)$) {};
		\path[draw,very thin, ] (w1) to (w1-w3-8) to (w3);
		\node[draw,circle,fill,inner sep=0mm, minimum size=0.1mm, ] (w1-w3-9) at ($(w1)!.5!(w3)!0.09310344827586207cm!90:(w3)$) {};
		\path[draw,very thin, ] (w1) to (w1-w3-9) to (w3);
		\node[draw,circle,fill,inner sep=0mm, minimum size=0.1mm, ] (w1-w3-10) at ($(w1)!.5!(w3)!0.12413793103448273cm!90:(w3)$) {};
		\path[draw,very thin, ] (w1) to (w1-w3-10) to (w3);
		\node[draw,circle,fill,inner sep=0mm, minimum size=0.1mm, ] (w1-w3-11) at ($(w1)!.5!(w3)!0.1551724137931034cm!90:(w3)$) {};
		\path[draw,very thin, ] (w1) to (w1-w3-11) to (w3);
		\node[draw,circle,fill,inner sep=0mm, minimum size=0.1mm, ] (w3-w2-0) at ($(w3)!.5!(w2)!-0.1862068965517241cm!90:(w2)$) {};
		\path[draw,very thin, ] (w3) to (w3-w2-0) to (w2);
		\node[draw,circle,fill,inner sep=0mm, minimum size=0.1mm, ] (w3-w2-1) at ($(w3)!.5!(w2)!-0.15517241379310343cm!90:(w2)$) {};
		\path[draw,very thin, ] (w3) to (w3-w2-1) to (w2);
		\node[draw,circle,fill,inner sep=0mm, minimum size=0.1mm, ] (w3-w2-2) at ($(w3)!.5!(w2)!-0.12413793103448273cm!90:(w2)$) {};
		\path[draw,very thin, ] (w3) to (w3-w2-2) to (w2);
		\node[draw,circle,fill,inner sep=0mm, minimum size=0.1mm, ] (w3-w2-3) at ($(w3)!.5!(w2)!-0.09310344827586205cm!90:(w2)$) {};
		\path[draw,very thin, ] (w3) to (w3-w2-3) to (w2);
		\node[draw,circle,fill,inner sep=0mm, minimum size=0.1mm, ] (w3-w2-4) at ($(w3)!.5!(w2)!-0.062068965517241365cm!90:(w2)$) {};
		\path[draw,very thin, ] (w3) to (w3-w2-4) to (w2);
		\node[draw,circle,fill,inner sep=0mm, minimum size=0.1mm, ] (w3-w2-5) at ($(w3)!.5!(w2)!-0.031034482758620682cm!90:(w2)$) {};
		\path[draw,very thin, ] (w3) to (w3-w2-5) to (w2);
		\node[draw,circle,fill,inner sep=0mm, minimum size=0.1mm, ] (w3-w2-6) at ($(w3)!.5!(w2)!0.0cm!90:(w2)$) {};
		\path[draw,very thin, ] (w3) to (w3-w2-6) to (w2);
		\node[draw,circle,fill,inner sep=0mm, minimum size=0.1mm, ] (w3-w2-7) at ($(w3)!.5!(w2)!0.031034482758620682cm!90:(w2)$) {};
		\path[draw,very thin, ] (w3) to (w3-w2-7) to (w2);
		\node[draw,circle,fill,inner sep=0mm, minimum size=0.1mm, ] (w3-w2-8) at ($(w3)!.5!(w2)!0.062068965517241365cm!90:(w2)$) {};
		\path[draw,very thin, ] (w3) to (w3-w2-8) to (w2);
		\node[draw,circle,fill,inner sep=0mm, minimum size=0.1mm, ] (w3-w2-9) at ($(w3)!.5!(w2)!0.09310344827586207cm!90:(w2)$) {};
		\path[draw,very thin, ] (w3) to (w3-w2-9) to (w2);
		\node[draw,circle,fill,inner sep=0mm, minimum size=0.1mm, ] (w3-w2-10) at ($(w3)!.5!(w2)!0.12413793103448273cm!90:(w2)$) {};
		\path[draw,very thin, ] (w3) to (w3-w2-10) to (w2);
		\node[draw,circle,fill,inner sep=0mm, minimum size=0.1mm, ] (w3-w2-11) at ($(w3)!.5!(w2)!0.1551724137931034cm!90:(w2)$) {};
		\path[draw,very thin, ] (w3) to (w3-w2-11) to (w2);
		\node[draw,circle,fill,inner sep=0mm, minimum size=0.1mm, ] (v-w2-0) at ($(v)!.5!(w2)!0.0cm!90:(w2)$) {};
		\path[draw,very thin, ] (v) to (v-w2-0) to (w2);
		\node[draw,circle,fill,inner sep=0mm, minimum size=0.1mm, ] (v-w2-1) at ($(v)!.5!(w2)!0.031034482758620682cm!90:(w2)$) {};
		\path[draw,very thin, ] (v) to (v-w2-1) to (w2);
		\node[draw,circle,fill,inner sep=0mm, minimum size=0.1mm, ] (v-w2-2) at ($(v)!.5!(w2)!0.062068965517241365cm!90:(w2)$) {};
		\path[draw,very thin, ] (v) to (v-w2-2) to (w2);
		\node[draw,circle,fill,inner sep=0mm, minimum size=0.1mm, ] (v-w2-3) at ($(v)!.5!(w2)!0.09310344827586205cm!90:(w2)$) {};
		\path[draw,very thin, ] (v) to (v-w2-3) to (w2);
		\node[draw,circle,fill,inner sep=0mm, minimum size=0.1mm, ] (v-w2-4) at ($(v)!.5!(w2)!0.12413793103448273cm!90:(w2)$) {};
		\path[draw,very thin, ] (v) to (v-w2-4) to (w2);
		\node[draw,circle,fill,inner sep=0mm, minimum size=0.1mm, ] (v-w2-5) at ($(v)!.5!(w2)!0.15517241379310343cm!90:(w2)$) {};
		\path[draw,very thin, ] (v) to (v-w2-5) to (w2);
		\node[draw,circle,fill,inner sep=0mm, minimum size=0.1mm, ] (v-w2-6) at ($(v)!.5!(w2)!0.1862068965517241cm!90:(w2)$) {};
		\path[draw,very thin, ] (v) to (v-w2-6) to (w2);
		\node[draw,circle,fill,inner sep=0mm, minimum size=0.1mm, ] (v-w2-7) at ($(v)!.5!(w2)!0.2172413793103448cm!90:(w2)$) {};
		\path[draw,very thin, ] (v) to (v-w2-7) to (w2);
		\node[draw,circle,fill,inner sep=0mm, minimum size=0.1mm, ] (v-w2-8) at ($(v)!.5!(w2)!0.24827586206896546cm!90:(w2)$) {};
		\path[draw,very thin, ] (v) to (v-w2-8) to (w2);
		\node[draw,circle,fill,inner sep=0mm, minimum size=0.1mm, ] (v-w2-9) at ($(v)!.5!(w2)!0.2793103448275862cm!90:(w2)$) {};
		\path[draw,very thin, ] (v) to (v-w2-9) to (w2);
		\node[draw,circle,fill,inner sep=0mm, minimum size=0.1mm, ] (v-w2-10) at ($(v)!.5!(w2)!0.31034482758620685cm!90:(w2)$) {};
		\path[draw,very thin, ] (v) to (v-w2-10) to (w2);
		\node[draw,circle,fill,inner sep=0mm, minimum size=0.1mm, ] (v-w2-11) at ($(v)!.5!(w2)!0.3413793103448275cm!90:(w2)$) {};
		\path[draw,very thin, ] (v) to (v-w2-11) to (w2);
		\node[draw,circle,fill,inner sep=0mm, minimum size=0.1mm, ] (w2-w1-0) at ($(w2)!.5!(w1)!0.0cm!90:(w1)$) {};
		\path[draw,very thin, ] (w2) to (w2-w1-0) to (w1);
		\node[draw,circle,fill,inner sep=0mm, minimum size=0.1mm, ] (w2-w1-1) at ($(w2)!.5!(w1)!0.031034482758620682cm!90:(w1)$) {};
		\path[draw,very thin, ] (w2) to (w2-w1-1) to (w1);
		\node[draw,circle,fill,inner sep=0mm, minimum size=0.1mm, ] (w2-w1-2) at ($(w2)!.5!(w1)!0.062068965517241365cm!90:(w1)$) {};
		\path[draw,very thin, ] (w2) to (w2-w1-2) to (w1);
		\node[draw,circle,fill,inner sep=0mm, minimum size=0.1mm, ] (w2-w1-3) at ($(w2)!.5!(w1)!0.09310344827586205cm!90:(w1)$) {};
		\path[draw,very thin, ] (w2) to (w2-w1-3) to (w1);
		\node[draw,circle,fill,inner sep=0mm, minimum size=0.1mm, ] (w2-w1-4) at ($(w2)!.5!(w1)!0.12413793103448273cm!90:(w1)$) {};
		\path[draw,very thin, ] (w2) to (w2-w1-4) to (w1);
		\node[draw,circle,fill,inner sep=0mm, minimum size=0.1mm, ] (w2-w1-5) at ($(w2)!.5!(w1)!0.15517241379310343cm!90:(w1)$) {};
		\path[draw,very thin, ] (w2) to (w2-w1-5) to (w1);
		\node[draw,circle,fill,inner sep=0mm, minimum size=0.1mm, ] (w2-w1-6) at ($(w2)!.5!(w1)!0.1862068965517241cm!90:(w1)$) {};
		\path[draw,very thin, ] (w2) to (w2-w1-6) to (w1);
		\node[draw,circle,fill,inner sep=0mm, minimum size=0.1mm, ] (w2-w1-7) at ($(w2)!.5!(w1)!0.2172413793103448cm!90:(w1)$) {};
		\path[draw,very thin, ] (w2) to (w2-w1-7) to (w1);
		\node[draw,circle,fill,inner sep=0mm, minimum size=0.1mm, ] (w2-w1-8) at ($(w2)!.5!(w1)!0.24827586206896546cm!90:(w1)$) {};
		\path[draw,very thin, ] (w2) to (w2-w1-8) to (w1);
		\node[draw,circle,fill,inner sep=0mm, minimum size=0.1mm, ] (w2-w1-9) at ($(w2)!.5!(w1)!0.2793103448275862cm!90:(w1)$) {};
		\path[draw,very thin, ] (w2) to (w2-w1-9) to (w1);
		\node[draw,circle,fill,inner sep=0mm, minimum size=0.1mm, ] (w2-w1-10) at ($(w2)!.5!(w1)!0.31034482758620685cm!90:(w1)$) {};
		\path[draw,very thin, ] (w2) to (w2-w1-10) to (w1);
		\node[draw,circle,fill,inner sep=0mm, minimum size=0.1mm, ] (w2-w1-11) at ($(w2)!.5!(w1)!0.3413793103448275cm!90:(w1)$) {};
		\path[draw,very thin, ] (w2) to (w2-w1-11) to (w1);
	\path[draw,semithick, , bend right=0.0] (w2) to (u) {};
	\path[draw,semithick, , bend right=0.0] (v) to (w1) {};
	\end{tikzpicture}
}
\tikzstyle{fencewheelHighlightedE}=[draw, double=yellow!80!gray!80!white!50!white, double distance=1.5pt, wheelS, line width=0.3pt,postaction={draw, double=yellow!80!gray!80!white!50!white, double distance=0.3pt, wheelS, line width=0.3pt},postaction={draw, wheelS, line width=2.7pt,dash pattern=on 0.3pt off 3.7pt, dash phase=2.0}]
\tikzstyle{fencewheelE}=[draw, double=white, double distance=1.5pt, wheelS, line width=0.3pt,postaction={draw, double=white, double distance=0.3pt, wheelS, line width=0.3pt},postaction={draw, wheelS, line width=2.7pt,dash pattern=on 0.3pt off 3.7pt, dash phase=2.0}]
\tikzstyle{fencesplitterHighlightedE}=[draw, double=yellow!80!gray!80!white!50!white, double distance=1.5pt, splitterS, line width=0.3pt,postaction={draw, double=yellow!80!gray!80!white!50!white, double distance=0.3pt, splitterS, line width=0.3pt},postaction={draw, splitterS, line width=2.7pt,dash pattern=on 0.3pt off 3.7pt, dash phase=2.0}]
\tikzstyle{fencesplitterE}=[draw, double=white, double distance=1.5pt, splitterS, line width=0.3pt,postaction={draw, double=white, double distance=0.3pt, splitterS, line width=0.3pt},postaction={draw, splitterS, line width=2.7pt,dash pattern=on 0.3pt off 3.7pt, dash phase=2.0}]
\tikzstyle{fencedividerHighlightedE}=[draw, double=yellow!80!gray!80!white!50!white, double distance=1.5pt, dividerS, line width=0.3pt,postaction={draw, double=yellow!80!gray!80!white!50!white, double distance=0.3pt, dividerS, line width=0.3pt},postaction={draw, dividerS, line width=2.7pt,dash pattern=on 0.3pt off 3.7pt, dash phase=2.0}]
\tikzstyle{fencedividerE}=[draw, double=white, double distance=1.5pt, dividerS, line width=0.3pt,postaction={draw, double=white, double distance=0.3pt, dividerS, line width=0.3pt},postaction={draw, dividerS, line width=2.7pt,dash pattern=on 0.3pt off 3.7pt, dash phase=2.0}]
\tikzstyle{highlightFenceE}=[opacity=0.5, draw, yellow!80!gray!80!white, line width=5pt,rounded corners]
\tikzstyle{highlightE}=[opacity=0.5, draw, yellow!80!gray!80!white, line width=3.5pt,rounded corners]
\tikzstyle{smallV}=[draw,fill,inner sep=0mm, minimum size=3.510pt]

\tikzstyle{singleE}=[draw,thick]

\tikzstyle{wheelS}=[BrickRed]
\tikzstyle{splitterS}=[MidnightBlue]
\tikzstyle{dividerS}=[gray]
\newcommand{\instanceFig}{
	\begin{tikzpicture}[rotate=90]
		\node[smallV,wheelS,circle] (c-0) at (6, 0) {};
		\node[wheelS,smallV,circle] (r-0-0)        at (6.64, 0) {};
		\path[highlightFenceE] (c-0) to (r-0-0);
		\path[fencewheelHighlightedE] (c-0) to (r-0-0);
		\node[wheelS,smallV,circle] (r-0-1) at (5.880534500203794, 2.6965024147104044) {};
		\path[fencewheelE] (c-0) to (r-0-1);
		\path[wheelS,singleE] (r-0-0) to (r-0-1);
		\node[wheelS,smallV,circle] (r-0-2) at (5.558719620104085, 2.6518807902707033) {};
		\path[fencewheelE] (c-0) to (r-0-2);
		\path[wheelS,singleE] (r-0-1) to (r-0-2);
		\node[wheelS,smallV,circle] (r-0-3) at (5.245316037983882, 2.5567113496367857) {};
		\path[fencewheelE] (c-0) to (r-0-3);
		\path[wheelS,singleE] (r-0-2) to (r-0-3);
		\node[wheelS,smallV,circle] (r-0-4) at (4.946297576920654, 2.4128081288732135) {};
		\path[fencewheelE] (c-0) to (r-0-4);
		\path[wheelS,singleE] (r-0-3) to (r-0-4);
		\node[wheelS,smallV,circle] (r-0-5) at (4.667363863482819, 2.222914084443037) {};
		\path[fencewheelE] (c-0) to (r-0-5);
		\path[wheelS,singleE] (r-0-4) to (r-0-5);
		\node[wheelS,smallV,circle] (r-0-6) at (4.41383168646758, 1.9906488093873351) {};
		\path[fencewheelE] (c-0) to (r-0-6);
		\path[wheelS,singleE] (r-0-5) to (r-0-6);
		\node[wheelS,smallV,circle] (r-0-7) at (4.190533652955203, 1.7204395398303287) {};
		\path[fencewheelE] (c-0) to (r-0-7);
		\path[wheelS,singleE] (r-0-6) to (r-0-7);
		\node[wheelS,smallV,circle] (r-0-8) at (4.001726073408783, 1.4174367669033117) {};
		\path[highlightFenceE] (c-0) to (r-0-8);
		\path[fencewheelHighlightedE] (c-0) to (r-0-8);
		\path[wheelS,singleE] (r-0-7) to (r-0-8);
		\node[wheelS,smallV,circle] (r-0-9) at (3.8510078316326024, 1.0874160626185907) {};
		\path[fencewheelE] (c-0) to (r-0-9);
		\path[wheelS,singleE] (r-0-8) to (r-0-9);
		\node[wheelS,smallV,circle] (r-0-10) at (3.741251786018473, 0.7366679910021451) {};
		\path[fencewheelE] (c-0) to (r-0-10);
		\path[wheelS,singleE] (r-0-9) to (r-0-10);
		\node[wheelS,smallV,circle] (r-0-11) at (3.6745500096491157, 0.37187820290207824) {};
		\path[fencewheelE] (c-0) to (r-0-11);
		\path[wheelS,singleE] (r-0-10) to (r-0-11);
		\node[wheelS,smallV,circle] (r-0-12) at (3.652173913043478, 3.3065463576978537e-16) {};
		\path[fencewheelE] (c-0) to (r-0-12);
		\path[wheelS,singleE] (r-0-11) to (r-0-12);
		\node[wheelS,smallV,circle] (r-0-13) at (3.674550009649116, -0.37187820290207996) {};
		\path[fencewheelE] (c-0) to (r-0-13);
		\path[wheelS,singleE] (r-0-12) to (r-0-13);
		\node[wheelS,smallV,circle] (r-0-14) at (3.741251786018473, -0.7366679910021444) {};
		\path[fencewheelE] (c-0) to (r-0-14);
		\path[wheelS,singleE] (r-0-13) to (r-0-14);
		\node[wheelS,smallV,circle] (r-0-15) at (3.851007831632602, -1.0874160626185903) {};
		\path[fencewheelE] (c-0) to (r-0-15);
		\path[wheelS,singleE] (r-0-14) to (r-0-15);
		\node[wheelS,smallV,circle] (r-0-16) at (4.001726073408784, -1.4174367669033132) {};
		\path[highlightFenceE] (c-0) to (r-0-16);
		\path[fencewheelHighlightedE] (c-0) to (r-0-16);
		\path[wheelS,singleE] (r-0-15) to (r-0-16);
		\node[wheelS,smallV,circle] (r-0-17) at (4.190533652955202, -1.720439539830328) {};
		\path[fencewheelE] (c-0) to (r-0-17);
		\path[wheelS,singleE] (r-0-16) to (r-0-17);
		\node[wheelS,smallV,circle] (r-0-18) at (4.41383168646758, -1.9906488093873345) {};
		\path[fencewheelE] (c-0) to (r-0-18);
		\path[wheelS,singleE] (r-0-17) to (r-0-18);
		\node[wheelS,smallV,circle] (r-0-19) at (4.667363863482819, -2.222914084443037) {};
		\path[fencewheelE] (c-0) to (r-0-19);
		\path[wheelS,singleE] (r-0-18) to (r-0-19);
		\node[wheelS,smallV,circle] (r-0-20) at (4.946297576920653, -2.412808128873213) {};
		\path[fencewheelE] (c-0) to (r-0-20);
		\path[wheelS,singleE] (r-0-19) to (r-0-20);
		\node[wheelS,smallV,circle] (r-0-21) at (5.245316037983881, -2.5567113496367853) {};
		\path[fencewheelE] (c-0) to (r-0-21);
		\path[wheelS,singleE] (r-0-20) to (r-0-21);
		\node[wheelS,smallV,circle] (r-0-22) at (5.5587196201040845, -2.651880790270703) {};
		\path[fencewheelE] (c-0) to (r-0-22);
		\path[wheelS,singleE] (r-0-21) to (r-0-22);
		\node[wheelS,smallV,circle] (r-0-23) at (5.880534500203796, -2.6965024147104044) {};
		\path[fencewheelE] (c-0) to (r-0-23);
		\path[wheelS,singleE] (r-0-22) to (r-0-23);
		\path[wheelS,singleE] (r-0-23) to (r-0-0);
		\node[splitterS,smallV,circle] (p-0-0) at (6.070324222205935, 4.447836253319064) {};
		\path[splitterS,singleE] (c-0) to (p-0-0);
		\node[splitterS,smallV,circle] (p-0-1) at (5.550330616593532, 4.288937761479118) {};
		\path[splitterS,singleE] (c-0) to (p-0-1);
		\node[splitterS,smallV,circle] (p-0-2) at (5.069484136135764, 4.052526856138468) {};
		\path[splitterS,singleE] (c-0) to (p-0-2);
		\node[splitterS,smallV,circle] (p-0-3) at (4.636151087933275, 3.7479123875965628) {};
		\path[splitterS,singleE] (c-0) to (p-0-3);
		\node[splitterS,smallV,circle] (p-0-4) at (4.257224541540772, 3.3855308891994884) {};
		\path[splitterS,singleE] (c-0) to (p-0-4);
		\node[splitterS,smallV,circle] (p-0-5) at (3.9380189910044834, 2.9766593877636334) {};
		\path[splitterS,singleE] (c-0) to (p-0-5);
		\node[splitterS,smallV,circle] (p-0-6) at (3.6822054284205237, 2.5331171607861904) {};
		\path[splitterS,singleE] (c-0) to (p-0-6);
		\node[splitterS,smallV,circle] (p-0-7) at (3.4917873732601588, 2.0669637128851726) {};
		\path[splitterS,singleE] (c-0) to (p-0-7);
		\node[splitterS,smallV,circle] (p-0-8) at (3.3671174375099095, 1.5902001910292167) {};
		\path[splitterS,singleE] (c-0) to (p-0-8);
		\node[splitterS,smallV,circle] (p-0-9) at (3.306953067194997, 1.1144812397421646) {};
		\path[splitterS,singleE] (c-0) to (p-0-9);
		\node[splitterS,smallV,circle] (p-0-10) at (3.3085492071812697, 0.6508439206751547) {};
		\path[splitterS,singleE] (c-0) to (p-0-10);
		\node[splitterS,smallV,circle] (p-0-11) at (3.367784808118428, 0.20945979630042652) {};
		\path[highlightE] (c-0) to (p-0-11);
		\path[splitterS,singleE] (c-0) to (p-0-10);
		\path[splitterS,singleE] (c-0) to (p-0-9);
		\path[wheelS,singleE] (r-0-11) to (r-0-12);
		\path[splitterS,singleE] (c-0) to (p-0-11);
		\node[splitterS,smallV,circle] (p-0-12) at (3.367784808118428, -0.2094597963004271) {};
		\path[splitterS,singleE] (c-0) to (p-0-12);
		\node[splitterS,smallV,circle] (p-0-13) at (3.30854920718127, -0.6508439206751553) {};
		\path[splitterS,singleE] (c-0) to (p-0-13);
		\node[splitterS,smallV,circle] (p-0-14) at (3.306953067194997, -1.114481239742165) {};
		\path[splitterS,singleE] (c-0) to (p-0-14);
		\node[splitterS,smallV,circle] (p-0-15) at (3.3671174375099095, -1.5902001910292174) {};
		\path[splitterS,singleE] (c-0) to (p-0-15);
		\node[splitterS,smallV,circle] (p-0-16) at (3.4917873732601588, -2.066963712885173) {};
		\path[splitterS,singleE] (c-0) to (p-0-16);
		\node[splitterS,smallV,circle] (p-0-17) at (3.6822054284205255, -2.533117160786192) {};
		\path[splitterS,singleE] (c-0) to (p-0-17);
		\node[splitterS,smallV,circle] (p-0-18) at (3.938018991004486, -2.9766593877636365) {};
		\path[splitterS,singleE] (c-0) to (p-0-18);
		\node[splitterS,smallV,circle] (p-0-19) at (4.257224541540771, -3.385530889199488) {};
		\path[splitterS,singleE] (c-0) to (p-0-19);
		\node[splitterS,smallV,circle] (p-0-20) at (4.6361510879332775, -3.7479123875965636) {};
		\path[splitterS,singleE] (c-0) to (p-0-20);
		\node[splitterS,smallV,circle] (p-0-21) at (5.069484136135763, -4.052526856138468) {};
		\path[splitterS,singleE] (c-0) to (p-0-21);
		\node[splitterS,smallV,circle] (p-0-22) at (5.550330616593531, -4.288937761479118) {};
		\path[splitterS,singleE] (c-0) to (p-0-22);
		\node[splitterS,smallV,circle] (p-0-23) at (6.070324222205936, -4.447836253319064) {};
		\path[splitterS,singleE] (c-0) to (p-0-23);
		\node[dividerS,smallV,circle] (div-0-8) at (3.4212750756846666, 1.8291683991942738) {};
		\node[dividerS,smallV,circle] (div-0-16) at (3.421275075684668, -1.8291683991942758) {};
		\node[smallV,wheelS,circle] (c-1) at (0, 0) {};
		\node[wheelS,smallV,circle] (r-1-0)        at (-0.5050000000000001, 0) {};
		\path[highlightFenceE] (c-1) to (r-1-0);
		\path[fencewheelHighlightedE] (c-1) to (r-1-0);
		\node[wheelS,smallV,circle] (r-1-1) at (0.4322611226670338, 1.9454688493763594) {};
		\path[fencewheelE] (c-1) to (r-1-1);
		\path[highlightE] (r-1-0) to (r-1-1);
		\path[wheelS,singleE] (r-1-0) to (r-1-1);
		\node[wheelS,smallV,circle] (r-1-2) at (0.9422605450997091, 1.6124981116580557) {};
		\path[fencewheelE] (c-1) to (r-1-2);
		\path[highlightE] (r-1-1) to (r-1-2);
		\path[wheelS,singleE] (r-1-1) to (r-1-2);
		\node[wheelS,smallV,circle] (r-1-3) at (1.3257778936037883, 1.0630775751774846) {};
		\path[highlightFenceE] (c-1) to (r-1-3);
		\path[fencewheelHighlightedE] (c-1) to (r-1-3);
		\path[highlightE] (r-1-2) to (r-1-3);
		\path[wheelS,singleE] (r-1-2) to (r-1-3);
		\node[wheelS,smallV,circle] (r-1-4) at (1.5313326409740706, 0.37095738184011157) {};
		\path[fencewheelE] (c-1) to (r-1-4);
		\path[highlightE] (r-1-3) to (r-1-4);
		\path[wheelS,singleE] (r-1-3) to (r-1-4);
		\node[wheelS,smallV,circle] (r-1-5) at (1.5313326409740706, -0.37095738184011157) {};
		\path[fencewheelE] (c-1) to (r-1-5);
		\path[highlightE] (r-1-4) to (r-1-5);
		\path[wheelS,singleE] (r-1-4) to (r-1-5);
		\node[wheelS,smallV,circle] (r-1-6) at (1.3257778936037885, -1.0630775751774844) {};
		\path[highlightFenceE] (c-1) to (r-1-6);
		\path[fencewheelHighlightedE] (c-1) to (r-1-6);
		\path[highlightE] (r-1-5) to (r-1-6);
		\path[wheelS,singleE] (r-1-5) to (r-1-6);
		\node[wheelS,smallV,circle] (r-1-7) at (0.9422605450997095, -1.6124981116580555) {};
		\path[fencewheelE] (c-1) to (r-1-7);
		\path[highlightE] (r-1-6) to (r-1-7);
		\path[wheelS,singleE] (r-1-6) to (r-1-7);
		\node[wheelS,smallV,circle] (r-1-8) at (0.4322611226670338, -1.9454688493763594) {};
		\path[fencewheelE] (c-1) to (r-1-8);
		\path[highlightE] (r-1-7) to (r-1-8);
		\path[wheelS,singleE] (r-1-7) to (r-1-8);
		\path[highlightE] (r-1-8) to (r-1-0);
		\path[wheelS,singleE] (r-1-8) to (r-1-0);
		\node[splitterS,smallV,circle] (p-1-0) at (0.35522142558012665, 4.747594335933401) {};
		\path[splitterS,singleE] (c-1) to (p-1-0);
		\node[splitterS,smallV,circle] (p-1-1) at (1.491394198819998, 3.860493006197141) {};
		\path[splitterS,singleE] (c-1) to (p-1-1);
		\node[splitterS,smallV,circle] (p-1-2) at (2.2045425621014214, 2.600670420349872) {};
		\path[splitterS,singleE] (c-1) to (p-1-2);
		\node[splitterS,smallV,circle] (p-1-3) at (2.4542017462949555, 1.23180781883211) {};
		\path[splitterS,singleE] (c-1) to (p-1-3);
		\node[splitterS,smallV,circle] (p-1-4) at (2.284615384615385, 0.0) {};
		\path[highlightE] (c-1) to (p-1-4);
		\path[splitterS,singleE] (c-1) to (p-1-3);
		\path[splitterS,singleE] (c-1) to (p-1-2);
		\path[wheelS,singleE] (r-1-4) to (r-1-5);
		\path[splitterS,singleE] (c-1) to (p-1-4);
		\node[splitterS,smallV,circle] (p-1-5) at (2.4542017462949555, -1.23180781883211) {};
		\path[splitterS,singleE] (c-1) to (p-1-5);
		\node[splitterS,smallV,circle] (p-1-6) at (2.2045425621014214, -2.600670420349872) {};
		\path[splitterS,singleE] (c-1) to (p-1-6);
		\node[splitterS,smallV,circle] (p-1-7) at (1.491394198819997, -3.8604930061971414) {};
		\path[splitterS,singleE] (c-1) to (p-1-7);
		\node[splitterS,smallV,circle] (p-1-8) at (0.3552214255801275, -4.747594335933401) {};
		\path[splitterS,singleE] (c-1) to (p-1-8);
		\node[dividerS,smallV,circle] (div-1-3) at (2.3864002084868186, 1.913539635319472) {};
		\node[dividerS,smallV,circle] (div-1-6) at (2.386400208486819, -1.9135396353194718) {};
		\path[highlightFenceE] (r-0-8) to (div-0-8) to (div-1-3) to (r-1-3);
		\path[fencedividerHighlightedE] (r-0-8) to (div-0-8) to (div-1-3) to (r-1-3);
		\path[highlightFenceE] (r-0-16) to (div-0-16) to (div-1-6) to (r-1-6);
		\path[fencedividerHighlightedE] (r-0-16) to (div-0-16) to (div-1-6) to (r-1-6);
		\node[dividerS,smallV,circle] (div-0-0) at (6.64cm, 5.940000000000001cm) {};
		\node[dividerS,smallV,circle] (div-1-0) at (-0.505cm, 5.940000000000001cm) {};
		\path[highlightFenceE] (r-0-0) to (div-0-0) to (div-1-0) to (r-1-0);
		\path[fencedividerHighlightedE] (r-0-0) to (div-0-0) to (div-1-0) to (r-1-0);
		\path[fencesplitterE] (p-0-0) to (p-1-0);
		\path[fencesplitterE] (p-0-1) to (p-1-0);
		\path[fencesplitterE] (p-0-2) to (p-1-0);
		\path[fencesplitterE] (p-0-3) to (p-1-0);
		\path[fencesplitterE] (p-0-4) to (p-1-1);
		\path[fencesplitterE] (p-0-5) to (p-1-1);
		\path[fencesplitterE] (p-0-6) to (p-1-2);
		\path[fencesplitterE] (p-0-7) to (p-1-2);
		\path[fencesplitterE] (p-0-8) to (p-1-3);
		\path[fencesplitterE] (p-0-9) to (p-1-4);
		\path[fencesplitterE] (p-0-10) to (p-1-4);
		\path[highlightFenceE] (p-0-11) to (p-1-4);
		\path[fencesplitterHighlightedE] (p-0-11) to (p-1-4);
		\path[fencesplitterE] (p-0-12) to (p-1-5);
		\path[fencesplitterE] (p-0-13) to (p-1-5);
		\path[fencesplitterE] (p-0-14) to (p-1-5);
		\path[fencesplitterE] (p-0-15) to (p-1-5);
		\path[fencesplitterE] (p-0-16) to (p-1-6);
		\path[fencesplitterE] (p-0-17) to (p-1-6);
		\path[fencesplitterE] (p-0-18) to (p-1-7);
		\path[fencesplitterE] (p-0-19) to (p-1-7);
		\path[fencesplitterE] (p-0-20) to (p-1-7);
		\path[fencesplitterE] (p-0-21) to (p-1-8);
		\path[fencesplitterE] (p-0-22) to (p-1-8);
		\path[fencesplitterE] (p-0-23) to (p-1-8);
		\node[smallV,yellow!80!gray!80!white,circle,scale=2.2] () at (c-0) {};
		\node[smallV,wheelS,circle] () at (c-0) {};
		\node[smallV,yellow!80!gray!80!white,circle,scale=2.2] () at (c-1) {};
		\node[smallV,wheelS,circle] () at (c-1) {};
		\node[smallV,yellow!80!gray!80!white,circle,scale=2.2] () at (r-1-0) {};
		\node[smallV,wheelS,circle] () at (r-1-0) {};
		\node[smallV,yellow!80!gray!80!white,circle,scale=2.2] () at (r-1-3) {};
		\node[smallV,wheelS,circle] () at (r-1-3) {};
		\node[smallV,yellow!80!gray!80!white,circle,scale=2.2] () at (r-1-6) {};
		\node[smallV,wheelS,circle] () at (r-1-6) {};
	\end{tikzpicture}
}
\newcommand{\closeUpFig}{
	\begin{tikzpicture}[]
\tikzstyle{importantV}=[draw,circle,fill,inner sep=0mm, minimum size=2.8899999999999997pt]
		\node[draw,wheelS,inner sep=0mm, minimum size=1.7mm, fill] (c-1) at (-2.5, 0) {};
		\node[wheelS,importantV] (r-1-0) at (1.5, 0.0) {};
		\node[wheelS,importantV] (r-1-1) at (-2.4999999999999996, 4.0) {};
		\draw[wheelS,singleE] (r-1-0) to (1.4847787923669822, -0.34862297099063266);
		\draw[wheelS,singleE] (r-1-0) to (r-1-1);
		\draw[wheelS,singleE] (r-1-1) to (-2.848622970990633, 3.984778792366982);
		\node[draw,circle,BrickRed!20,inner sep=0mm, minimum size=1.0mm] (u_0) at (r-1-0) {};
		\node[draw,circle,BrickRed!20,inner sep=0mm, minimum size=1.0mm] (v_0) at (c-1) {};
		\node[draw,circle,BrickRed!20,inner sep=0mm, minimum size=1.0mm] (w1_0) at ($(u_0)!0.5!(v_0)!0.27586206896551724!230:(v_0)$) {};
		\node[draw,circle,BrickRed!20,inner sep=0mm, minimum size=1.0mm] (w2_0) at ($(u_0)!0.5!(v_0)!0.27586206896551724!310:(v_0)$) {};
		\node[draw,circle,BrickRed!20,inner sep=0mm, minimum size=1.0mm] (w3_0) at ($(u_0)!0.5!(v_0)!0.3804994054696789!270:(v_0)$) {};
		\node[draw,circle,fill,inner sep=0mm, minimum size=0.1mm, BrickRed!20] (u-v-0_0) at ($(u_0)!.5!(v_0)!0.05136741973840665cm!90:(v_0)$) {};
		\path[draw,very thin, BrickRed!20] (u_0) to (u-v-0_0) to (v_0);
		\node[draw,circle,fill,inner sep=0mm, minimum size=0.1mm, BrickRed!20] (u-v-1_0) at ($(u_0)!.5!(v_0)!0.06516052318668251cm!90:(v_0)$) {};
		\path[draw,very thin, BrickRed!20] (u_0) to (u-v-1_0) to (v_0);
		\node[draw,circle,fill,inner sep=0mm, minimum size=0.1mm, BrickRed!20] (u-v-2_0) at ($(u_0)!.5!(v_0)!0.07895362663495838cm!90:(v_0)$) {};
		\path[draw,very thin, BrickRed!20] (u_0) to (u-v-2_0) to (v_0);
		\node[draw,circle,fill,inner sep=0mm, minimum size=0.1mm, BrickRed!20] (u-v-3_0) at ($(u_0)!.5!(v_0)!0.09274673008323422cm!90:(v_0)$) {};
		\path[draw,very thin, BrickRed!20] (u_0) to (u-v-3_0) to (v_0);
		\node[draw,circle,fill,inner sep=0mm, minimum size=0.1mm, BrickRed!20] (u-v-4_0) at ($(u_0)!.5!(v_0)!0.1065398335315101cm!90:(v_0)$) {};
		\path[draw,very thin, BrickRed!20] (u_0) to (u-v-4_0) to (v_0);
		\node[draw,circle,fill,inner sep=0mm, minimum size=0.1mm, BrickRed!20] (u-v-5_0) at ($(u_0)!.5!(v_0)!0.12033293697978598cm!90:(v_0)$) {};
		\path[draw,very thin, BrickRed!20] (u_0) to (u-v-5_0) to (v_0);
		\node[draw,circle,fill,inner sep=0mm, minimum size=0.1mm, BrickRed!20] (u-v-6_0) at ($(u_0)!.5!(v_0)!0.1341260404280618cm!90:(v_0)$) {};
		\path[draw,very thin, BrickRed!20] (u_0) to (u-v-6_0) to (v_0);
		\node[draw,circle,fill,inner sep=0mm, minimum size=0.1mm, BrickRed!20] (u-v-7_0) at ($(u_0)!.5!(v_0)!0.14791914387633767cm!90:(v_0)$) {};
		\path[draw,very thin, BrickRed!20] (u_0) to (u-v-7_0) to (v_0);
		\path[draw, semithick, BrickRed!20, thick,BrickRed, bend right=0.0] (u_0)  to (v_0);
		\node[draw,circle,fill,inner sep=0mm, minimum size=0.1mm, BrickRed!20] (w3-u-0_0) at ($(w3_0)!.5!(u_0)!0.14268727705112957cm!90:(u_0)$) {};
		\path[draw,very thin, BrickRed!20] (w3_0) to (w3-u-0_0) to (u_0);
		\node[draw,circle,fill,inner sep=0mm, minimum size=0.1mm, BrickRed!20] (w3-u-1_0) at ($(w3_0)!.5!(u_0)!0.15648038049940544cm!90:(u_0)$) {};
		\path[draw,very thin, BrickRed!20] (w3_0) to (w3-u-1_0) to (u_0);
		\node[draw,circle,fill,inner sep=0mm, minimum size=0.1mm, BrickRed!20] (w3-u-2_0) at ($(w3_0)!.5!(u_0)!0.1702734839476813cm!90:(u_0)$) {};
		\path[draw,very thin, BrickRed!20] (w3_0) to (w3-u-2_0) to (u_0);
		\node[draw,circle,fill,inner sep=0mm, minimum size=0.1mm, BrickRed!20] (w3-u-3_0) at ($(w3_0)!.5!(u_0)!0.18406658739595716cm!90:(u_0)$) {};
		\path[draw,very thin, BrickRed!20] (w3_0) to (w3-u-3_0) to (u_0);
		\node[draw,circle,fill,inner sep=0mm, minimum size=0.1mm, BrickRed!20] (w3-u-4_0) at ($(w3_0)!.5!(u_0)!0.19785969084423302cm!90:(u_0)$) {};
		\path[draw,very thin, BrickRed!20] (w3_0) to (w3-u-4_0) to (u_0);
		\node[draw,circle,fill,inner sep=0mm, minimum size=0.1mm, BrickRed!20] (w3-u-5_0) at ($(w3_0)!.5!(u_0)!0.21165279429250886cm!90:(u_0)$) {};
		\path[draw,very thin, BrickRed!20] (w3_0) to (w3-u-5_0) to (u_0);
		\node[draw,circle,fill,inner sep=0mm, minimum size=0.1mm, BrickRed!20] (w3-u-6_0) at ($(w3_0)!.5!(u_0)!0.22544589774078475cm!90:(u_0)$) {};
		\path[draw,very thin, BrickRed!20] (w3_0) to (w3-u-6_0) to (u_0);
		\node[draw,circle,fill,inner sep=0mm, minimum size=0.1mm, BrickRed!20] (w3-u-7_0) at ($(w3_0)!.5!(u_0)!0.2392390011890606cm!90:(u_0)$) {};
		\path[draw,very thin, BrickRed!20] (w3_0) to (w3-u-7_0) to (u_0);
		\node[draw,circle,fill,inner sep=0mm, minimum size=0.1mm, BrickRed!20] (v-w3-0_0) at ($(v_0)!.5!(w3_0)!0.14268727705112957cm!90:(w3_0)$) {};
		\path[draw,very thin, BrickRed!20] (v_0) to (v-w3-0_0) to (w3_0);
		\node[draw,circle,fill,inner sep=0mm, minimum size=0.1mm, BrickRed!20] (v-w3-1_0) at ($(v_0)!.5!(w3_0)!0.15648038049940544cm!90:(w3_0)$) {};
		\path[draw,very thin, BrickRed!20] (v_0) to (v-w3-1_0) to (w3_0);
		\node[draw,circle,fill,inner sep=0mm, minimum size=0.1mm, BrickRed!20] (v-w3-2_0) at ($(v_0)!.5!(w3_0)!0.1702734839476813cm!90:(w3_0)$) {};
		\path[draw,very thin, BrickRed!20] (v_0) to (v-w3-2_0) to (w3_0);
		\node[draw,circle,fill,inner sep=0mm, minimum size=0.1mm, BrickRed!20] (v-w3-3_0) at ($(v_0)!.5!(w3_0)!0.18406658739595716cm!90:(w3_0)$) {};
		\path[draw,very thin, BrickRed!20] (v_0) to (v-w3-3_0) to (w3_0);
		\node[draw,circle,fill,inner sep=0mm, minimum size=0.1mm, BrickRed!20] (v-w3-4_0) at ($(v_0)!.5!(w3_0)!0.19785969084423302cm!90:(w3_0)$) {};
		\path[draw,very thin, BrickRed!20] (v_0) to (v-w3-4_0) to (w3_0);
		\node[draw,circle,fill,inner sep=0mm, minimum size=0.1mm, BrickRed!20] (v-w3-5_0) at ($(v_0)!.5!(w3_0)!0.21165279429250886cm!90:(w3_0)$) {};
		\path[draw,very thin, BrickRed!20] (v_0) to (v-w3-5_0) to (w3_0);
		\node[draw,circle,fill,inner sep=0mm, minimum size=0.1mm, BrickRed!20] (v-w3-6_0) at ($(v_0)!.5!(w3_0)!0.22544589774078475cm!90:(w3_0)$) {};
		\path[draw,very thin, BrickRed!20] (v_0) to (v-w3-6_0) to (w3_0);
		\node[draw,circle,fill,inner sep=0mm, minimum size=0.1mm, BrickRed!20] (v-w3-7_0) at ($(v_0)!.5!(w3_0)!0.2392390011890606cm!90:(w3_0)$) {};
		\path[draw,very thin, BrickRed!20] (v_0) to (v-w3-7_0) to (w3_0);
		\node[draw,circle,fill,inner sep=0mm, minimum size=0.1mm, BrickRed!20] (w1-u-0_0) at ($(w1_0)!.5!(u_0)!0.0cm!90:(u_0)$) {};
		\path[draw,very thin, BrickRed!20] (w1_0) to (w1-u-0_0) to (u_0);
		\node[draw,circle,fill,inner sep=0mm, minimum size=0.1mm, BrickRed!20] (w1-u-1_0) at ($(w1_0)!.5!(u_0)!0.013793103448275862cm!90:(u_0)$) {};
		\path[draw,very thin, BrickRed!20] (w1_0) to (w1-u-1_0) to (u_0);
		\node[draw,circle,fill,inner sep=0mm, minimum size=0.1mm, BrickRed!20] (w1-u-2_0) at ($(w1_0)!.5!(u_0)!0.027586206896551724cm!90:(u_0)$) {};
		\path[draw,very thin, BrickRed!20] (w1_0) to (w1-u-2_0) to (u_0);
		\node[draw,circle,fill,inner sep=0mm, minimum size=0.1mm, BrickRed!20] (w1-u-3_0) at ($(w1_0)!.5!(u_0)!0.04137931034482758cm!90:(u_0)$) {};
		\path[draw,very thin, BrickRed!20] (w1_0) to (w1-u-3_0) to (u_0);
		\node[draw,circle,fill,inner sep=0mm, minimum size=0.1mm, BrickRed!20] (w1-u-4_0) at ($(w1_0)!.5!(u_0)!0.05517241379310345cm!90:(u_0)$) {};
		\path[draw,very thin, BrickRed!20] (w1_0) to (w1-u-4_0) to (u_0);
		\node[draw,circle,fill,inner sep=0mm, minimum size=0.1mm, BrickRed!20] (w1-u-5_0) at ($(w1_0)!.5!(u_0)!0.06896551724137931cm!90:(u_0)$) {};
		\path[draw,very thin, BrickRed!20] (w1_0) to (w1-u-5_0) to (u_0);
		\node[draw,circle,fill,inner sep=0mm, minimum size=0.1mm, BrickRed!20] (w1-u-6_0) at ($(w1_0)!.5!(u_0)!0.08275862068965516cm!90:(u_0)$) {};
		\path[draw,very thin, BrickRed!20] (w1_0) to (w1-u-6_0) to (u_0);
		\node[draw,circle,fill,inner sep=0mm, minimum size=0.1mm, BrickRed!20] (w1-u-7_0) at ($(w1_0)!.5!(u_0)!0.09655172413793103cm!90:(u_0)$) {};
		\path[draw,very thin, BrickRed!20] (w1_0) to (w1-u-7_0) to (u_0);
		\node[draw,circle,fill,inner sep=0mm, minimum size=0.1mm, BrickRed!20] (w1-w3-0_0) at ($(w1_0)!.5!(w3_0)!-0.05517241379310345cm!90:(w3_0)$) {};
		\path[draw,very thin, BrickRed!20] (w1_0) to (w1-w3-0_0) to (w3_0);
		\node[draw,circle,fill,inner sep=0mm, minimum size=0.1mm, BrickRed!20] (w1-w3-1_0) at ($(w1_0)!.5!(w3_0)!-0.04137931034482758cm!90:(w3_0)$) {};
		\path[draw,very thin, BrickRed!20] (w1_0) to (w1-w3-1_0) to (w3_0);
		\node[draw,circle,fill,inner sep=0mm, minimum size=0.1mm, BrickRed!20] (w1-w3-2_0) at ($(w1_0)!.5!(w3_0)!-0.027586206896551724cm!90:(w3_0)$) {};
		\path[draw,very thin, BrickRed!20] (w1_0) to (w1-w3-2_0) to (w3_0);
		\node[draw,circle,fill,inner sep=0mm, minimum size=0.1mm, BrickRed!20] (w1-w3-3_0) at ($(w1_0)!.5!(w3_0)!-0.013793103448275864cm!90:(w3_0)$) {};
		\path[draw,very thin, BrickRed!20] (w1_0) to (w1-w3-3_0) to (w3_0);
		\node[draw,circle,fill,inner sep=0mm, minimum size=0.1mm, BrickRed!20] (w1-w3-4_0) at ($(w1_0)!.5!(w3_0)!0.0cm!90:(w3_0)$) {};
		\path[draw,very thin, BrickRed!20] (w1_0) to (w1-w3-4_0) to (w3_0);
		\node[draw,circle,fill,inner sep=0mm, minimum size=0.1mm, BrickRed!20] (w1-w3-5_0) at ($(w1_0)!.5!(w3_0)!0.013793103448275864cm!90:(w3_0)$) {};
		\path[draw,very thin, BrickRed!20] (w1_0) to (w1-w3-5_0) to (w3_0);
		\node[draw,circle,fill,inner sep=0mm, minimum size=0.1mm, BrickRed!20] (w1-w3-6_0) at ($(w1_0)!.5!(w3_0)!0.027586206896551717cm!90:(w3_0)$) {};
		\path[draw,very thin, BrickRed!20] (w1_0) to (w1-w3-6_0) to (w3_0);
		\node[draw,circle,fill,inner sep=0mm, minimum size=0.1mm, BrickRed!20] (w1-w3-7_0) at ($(w1_0)!.5!(w3_0)!0.04137931034482759cm!90:(w3_0)$) {};
		\path[draw,very thin, BrickRed!20] (w1_0) to (w1-w3-7_0) to (w3_0);
		\node[draw,circle,fill,inner sep=0mm, minimum size=0.1mm, BrickRed!20] (w3-w2-0_0) at ($(w3_0)!.5!(w2_0)!-0.05517241379310345cm!90:(w2_0)$) {};
		\path[draw,very thin, BrickRed!20] (w3_0) to (w3-w2-0_0) to (w2_0);
		\node[draw,circle,fill,inner sep=0mm, minimum size=0.1mm, BrickRed!20] (w3-w2-1_0) at ($(w3_0)!.5!(w2_0)!-0.04137931034482758cm!90:(w2_0)$) {};
		\path[draw,very thin, BrickRed!20] (w3_0) to (w3-w2-1_0) to (w2_0);
		\node[draw,circle,fill,inner sep=0mm, minimum size=0.1mm, BrickRed!20] (w3-w2-2_0) at ($(w3_0)!.5!(w2_0)!-0.027586206896551724cm!90:(w2_0)$) {};
		\path[draw,very thin, BrickRed!20] (w3_0) to (w3-w2-2_0) to (w2_0);
		\node[draw,circle,fill,inner sep=0mm, minimum size=0.1mm, BrickRed!20] (w3-w2-3_0) at ($(w3_0)!.5!(w2_0)!-0.013793103448275864cm!90:(w2_0)$) {};
		\path[draw,very thin, BrickRed!20] (w3_0) to (w3-w2-3_0) to (w2_0);
		\node[draw,circle,fill,inner sep=0mm, minimum size=0.1mm, BrickRed!20] (w3-w2-4_0) at ($(w3_0)!.5!(w2_0)!0.0cm!90:(w2_0)$) {};
		\path[draw,very thin, BrickRed!20] (w3_0) to (w3-w2-4_0) to (w2_0);
		\node[draw,circle,fill,inner sep=0mm, minimum size=0.1mm, BrickRed!20] (w3-w2-5_0) at ($(w3_0)!.5!(w2_0)!0.013793103448275864cm!90:(w2_0)$) {};
		\path[draw,very thin, BrickRed!20] (w3_0) to (w3-w2-5_0) to (w2_0);
		\node[draw,circle,fill,inner sep=0mm, minimum size=0.1mm, BrickRed!20] (w3-w2-6_0) at ($(w3_0)!.5!(w2_0)!0.027586206896551717cm!90:(w2_0)$) {};
		\path[draw,very thin, BrickRed!20] (w3_0) to (w3-w2-6_0) to (w2_0);
		\node[draw,circle,fill,inner sep=0mm, minimum size=0.1mm, BrickRed!20] (w3-w2-7_0) at ($(w3_0)!.5!(w2_0)!0.04137931034482759cm!90:(w2_0)$) {};
		\path[draw,very thin, BrickRed!20] (w3_0) to (w3-w2-7_0) to (w2_0);
		\node[draw,circle,fill,inner sep=0mm, minimum size=0.1mm, BrickRed!20] (v-w2-0_0) at ($(v_0)!.5!(w2_0)!0.0cm!90:(w2_0)$) {};
		\path[draw,very thin, BrickRed!20] (v_0) to (v-w2-0_0) to (w2_0);
		\node[draw,circle,fill,inner sep=0mm, minimum size=0.1mm, BrickRed!20] (v-w2-1_0) at ($(v_0)!.5!(w2_0)!0.013793103448275862cm!90:(w2_0)$) {};
		\path[draw,very thin, BrickRed!20] (v_0) to (v-w2-1_0) to (w2_0);
		\node[draw,circle,fill,inner sep=0mm, minimum size=0.1mm, BrickRed!20] (v-w2-2_0) at ($(v_0)!.5!(w2_0)!0.027586206896551724cm!90:(w2_0)$) {};
		\path[draw,very thin, BrickRed!20] (v_0) to (v-w2-2_0) to (w2_0);
		\node[draw,circle,fill,inner sep=0mm, minimum size=0.1mm, BrickRed!20] (v-w2-3_0) at ($(v_0)!.5!(w2_0)!0.04137931034482758cm!90:(w2_0)$) {};
		\path[draw,very thin, BrickRed!20] (v_0) to (v-w2-3_0) to (w2_0);
		\node[draw,circle,fill,inner sep=0mm, minimum size=0.1mm, BrickRed!20] (v-w2-4_0) at ($(v_0)!.5!(w2_0)!0.05517241379310345cm!90:(w2_0)$) {};
		\path[draw,very thin, BrickRed!20] (v_0) to (v-w2-4_0) to (w2_0);
		\node[draw,circle,fill,inner sep=0mm, minimum size=0.1mm, BrickRed!20] (v-w2-5_0) at ($(v_0)!.5!(w2_0)!0.06896551724137931cm!90:(w2_0)$) {};
		\path[draw,very thin, BrickRed!20] (v_0) to (v-w2-5_0) to (w2_0);
		\node[draw,circle,fill,inner sep=0mm, minimum size=0.1mm, BrickRed!20] (v-w2-6_0) at ($(v_0)!.5!(w2_0)!0.08275862068965516cm!90:(w2_0)$) {};
		\path[draw,very thin, BrickRed!20] (v_0) to (v-w2-6_0) to (w2_0);
		\node[draw,circle,fill,inner sep=0mm, minimum size=0.1mm, BrickRed!20] (v-w2-7_0) at ($(v_0)!.5!(w2_0)!0.09655172413793103cm!90:(w2_0)$) {};
		\path[draw,very thin, BrickRed!20] (v_0) to (v-w2-7_0) to (w2_0);
		\node[draw,circle,fill,inner sep=0mm, minimum size=0.1mm, BrickRed!20] (w2-w1-0_0) at ($(w2_0)!.5!(w1_0)!0.0cm!90:(w1_0)$) {};
		\path[draw,very thin, BrickRed!20] (w2_0) to (w2-w1-0_0) to (w1_0);
		\node[draw,circle,fill,inner sep=0mm, minimum size=0.1mm, BrickRed!20] (w2-w1-1_0) at ($(w2_0)!.5!(w1_0)!0.013793103448275862cm!90:(w1_0)$) {};
		\path[draw,very thin, BrickRed!20] (w2_0) to (w2-w1-1_0) to (w1_0);
		\node[draw,circle,fill,inner sep=0mm, minimum size=0.1mm, BrickRed!20] (w2-w1-2_0) at ($(w2_0)!.5!(w1_0)!0.027586206896551724cm!90:(w1_0)$) {};
		\path[draw,very thin, BrickRed!20] (w2_0) to (w2-w1-2_0) to (w1_0);
		\node[draw,circle,fill,inner sep=0mm, minimum size=0.1mm, BrickRed!20] (w2-w1-3_0) at ($(w2_0)!.5!(w1_0)!0.04137931034482758cm!90:(w1_0)$) {};
		\path[draw,very thin, BrickRed!20] (w2_0) to (w2-w1-3_0) to (w1_0);
		\node[draw,circle,fill,inner sep=0mm, minimum size=0.1mm, BrickRed!20] (w2-w1-4_0) at ($(w2_0)!.5!(w1_0)!0.05517241379310345cm!90:(w1_0)$) {};
		\path[draw,very thin, BrickRed!20] (w2_0) to (w2-w1-4_0) to (w1_0);
		\node[draw,circle,fill,inner sep=0mm, minimum size=0.1mm, BrickRed!20] (w2-w1-5_0) at ($(w2_0)!.5!(w1_0)!0.06896551724137931cm!90:(w1_0)$) {};
		\path[draw,very thin, BrickRed!20] (w2_0) to (w2-w1-5_0) to (w1_0);
		\node[draw,circle,fill,inner sep=0mm, minimum size=0.1mm, BrickRed!20] (w2-w1-6_0) at ($(w2_0)!.5!(w1_0)!0.08275862068965516cm!90:(w1_0)$) {};
		\path[draw,very thin, BrickRed!20] (w2_0) to (w2-w1-6_0) to (w1_0);
		\node[draw,circle,fill,inner sep=0mm, minimum size=0.1mm, BrickRed!20] (w2-w1-7_0) at ($(w2_0)!.5!(w1_0)!0.09655172413793103cm!90:(w1_0)$) {};
		\path[draw,very thin, BrickRed!20] (w2_0) to (w2-w1-7_0) to (w1_0);
	\path[draw,semithick, BrickRed!60, bend right=0.0] (w2_0) to (u_0) {};
	\path[draw,semithick, BrickRed!60, bend right=0.0] (v_0) to (w1_0) {};
		\node[draw,circle,BrickRed!20,inner sep=0mm, minimum size=1.0mm] (u_1) at (c-1) {};
		\node[draw,circle,BrickRed!20,inner sep=0mm, minimum size=1.0mm] (v_1) at (r-1-1) {};
		\node[draw,circle,BrickRed!20,inner sep=0mm, minimum size=1.0mm] (w1_1) at ($(u_1)!0.5!(v_1)!0.27586206896551724!230:(v_1)$) {};
		\node[draw,circle,BrickRed!20,inner sep=0mm, minimum size=1.0mm] (w2_1) at ($(u_1)!0.5!(v_1)!0.27586206896551724!310:(v_1)$) {};
		\node[draw,circle,BrickRed!20,inner sep=0mm, minimum size=1.0mm] (w3_1) at ($(u_1)!0.5!(v_1)!0.3804994054696789!270:(v_1)$) {};
		\node[draw,circle,fill,inner sep=0mm, minimum size=0.1mm, BrickRed!20] (u-v-0_1) at ($(u_1)!.5!(v_1)!0.0cm!90:(v_1)$) {};
		\path[draw,very thin, BrickRed!20] (u_1) to (u-v-0_1) to (v_1);
		\node[draw,circle,fill,inner sep=0mm, minimum size=0.1mm, BrickRed!20] (u-v-1_1) at ($(u_1)!.5!(v_1)!0.013793103448275862cm!90:(v_1)$) {};
		\path[draw,very thin, BrickRed!20] (u_1) to (u-v-1_1) to (v_1);
		\node[draw,circle,fill,inner sep=0mm, minimum size=0.1mm, BrickRed!20] (u-v-2_1) at ($(u_1)!.5!(v_1)!0.027586206896551724cm!90:(v_1)$) {};
		\path[draw,very thin, BrickRed!20] (u_1) to (u-v-2_1) to (v_1);
		\node[draw,circle,fill,inner sep=0mm, minimum size=0.1mm, BrickRed!20] (u-v-3_1) at ($(u_1)!.5!(v_1)!0.04137931034482758cm!90:(v_1)$) {};
		\path[draw,very thin, BrickRed!20] (u_1) to (u-v-3_1) to (v_1);
		\node[draw,circle,fill,inner sep=0mm, minimum size=0.1mm, BrickRed!20] (u-v-4_1) at ($(u_1)!.5!(v_1)!0.05517241379310345cm!90:(v_1)$) {};
		\path[draw,very thin, BrickRed!20] (u_1) to (u-v-4_1) to (v_1);
		\node[draw,circle,fill,inner sep=0mm, minimum size=0.1mm, BrickRed!20] (u-v-5_1) at ($(u_1)!.5!(v_1)!0.06896551724137931cm!90:(v_1)$) {};
		\path[draw,very thin, BrickRed!20] (u_1) to (u-v-5_1) to (v_1);
		\node[draw,circle,fill,inner sep=0mm, minimum size=0.1mm, BrickRed!20] (u-v-6_1) at ($(u_1)!.5!(v_1)!0.08275862068965516cm!90:(v_1)$) {};
		\path[draw,very thin, BrickRed!20] (u_1) to (u-v-6_1) to (v_1);
		\node[draw,circle,fill,inner sep=0mm, minimum size=0.1mm, BrickRed!20] (u-v-7_1) at ($(u_1)!.5!(v_1)!0.09655172413793103cm!90:(v_1)$) {};
		\path[draw,very thin, BrickRed!20] (u_1) to (u-v-7_1) to (v_1);
		\node[draw,circle,fill,inner sep=0mm, minimum size=0.1mm, BrickRed!20] (w3-u-0_1) at ($(w3_1)!.5!(u_1)!0.14268727705112957cm!90:(u_1)$) {};
		\path[draw,very thin, BrickRed!20] (w3_1) to (w3-u-0_1) to (u_1);
		\node[draw,circle,fill,inner sep=0mm, minimum size=0.1mm, BrickRed!20] (w3-u-1_1) at ($(w3_1)!.5!(u_1)!0.15648038049940544cm!90:(u_1)$) {};
		\path[draw,very thin, BrickRed!20] (w3_1) to (w3-u-1_1) to (u_1);
		\node[draw,circle,fill,inner sep=0mm, minimum size=0.1mm, BrickRed!20] (w3-u-2_1) at ($(w3_1)!.5!(u_1)!0.1702734839476813cm!90:(u_1)$) {};
		\path[draw,very thin, BrickRed!20] (w3_1) to (w3-u-2_1) to (u_1);
		\node[draw,circle,fill,inner sep=0mm, minimum size=0.1mm, BrickRed!20] (w3-u-3_1) at ($(w3_1)!.5!(u_1)!0.18406658739595716cm!90:(u_1)$) {};
		\path[draw,very thin, BrickRed!20] (w3_1) to (w3-u-3_1) to (u_1);
		\node[draw,circle,fill,inner sep=0mm, minimum size=0.1mm, BrickRed!20] (w3-u-4_1) at ($(w3_1)!.5!(u_1)!0.19785969084423302cm!90:(u_1)$) {};
		\path[draw,very thin, BrickRed!20] (w3_1) to (w3-u-4_1) to (u_1);
		\node[draw,circle,fill,inner sep=0mm, minimum size=0.1mm, BrickRed!20] (w3-u-5_1) at ($(w3_1)!.5!(u_1)!0.21165279429250886cm!90:(u_1)$) {};
		\path[draw,very thin, BrickRed!20] (w3_1) to (w3-u-5_1) to (u_1);
		\node[draw,circle,fill,inner sep=0mm, minimum size=0.1mm, BrickRed!20] (w3-u-6_1) at ($(w3_1)!.5!(u_1)!0.22544589774078475cm!90:(u_1)$) {};
		\path[draw,very thin, BrickRed!20] (w3_1) to (w3-u-6_1) to (u_1);
		\node[draw,circle,fill,inner sep=0mm, minimum size=0.1mm, BrickRed!20] (w3-u-7_1) at ($(w3_1)!.5!(u_1)!0.2392390011890606cm!90:(u_1)$) {};
		\path[draw,very thin, BrickRed!20] (w3_1) to (w3-u-7_1) to (u_1);
		\node[draw,circle,fill,inner sep=0mm, minimum size=0.1mm, BrickRed!20] (v-w3-0_1) at ($(v_1)!.5!(w3_1)!0.14268727705112957cm!90:(w3_1)$) {};
		\path[draw,very thin, BrickRed!20] (v_1) to (v-w3-0_1) to (w3_1);
		\node[draw,circle,fill,inner sep=0mm, minimum size=0.1mm, BrickRed!20] (v-w3-1_1) at ($(v_1)!.5!(w3_1)!0.15648038049940544cm!90:(w3_1)$) {};
		\path[draw,very thin, BrickRed!20] (v_1) to (v-w3-1_1) to (w3_1);
		\node[draw,circle,fill,inner sep=0mm, minimum size=0.1mm, BrickRed!20] (v-w3-2_1) at ($(v_1)!.5!(w3_1)!0.1702734839476813cm!90:(w3_1)$) {};
		\path[draw,very thin, BrickRed!20] (v_1) to (v-w3-2_1) to (w3_1);
		\node[draw,circle,fill,inner sep=0mm, minimum size=0.1mm, BrickRed!20] (v-w3-3_1) at ($(v_1)!.5!(w3_1)!0.18406658739595716cm!90:(w3_1)$) {};
		\path[draw,very thin, BrickRed!20] (v_1) to (v-w3-3_1) to (w3_1);
		\node[draw,circle,fill,inner sep=0mm, minimum size=0.1mm, BrickRed!20] (v-w3-4_1) at ($(v_1)!.5!(w3_1)!0.19785969084423302cm!90:(w3_1)$) {};
		\path[draw,very thin, BrickRed!20] (v_1) to (v-w3-4_1) to (w3_1);
		\node[draw,circle,fill,inner sep=0mm, minimum size=0.1mm, BrickRed!20] (v-w3-5_1) at ($(v_1)!.5!(w3_1)!0.21165279429250886cm!90:(w3_1)$) {};
		\path[draw,very thin, BrickRed!20] (v_1) to (v-w3-5_1) to (w3_1);
		\node[draw,circle,fill,inner sep=0mm, minimum size=0.1mm, BrickRed!20] (v-w3-6_1) at ($(v_1)!.5!(w3_1)!0.22544589774078475cm!90:(w3_1)$) {};
		\path[draw,very thin, BrickRed!20] (v_1) to (v-w3-6_1) to (w3_1);
		\node[draw,circle,fill,inner sep=0mm, minimum size=0.1mm, BrickRed!20] (v-w3-7_1) at ($(v_1)!.5!(w3_1)!0.2392390011890606cm!90:(w3_1)$) {};
		\path[draw,very thin, BrickRed!20] (v_1) to (v-w3-7_1) to (w3_1);
		\node[draw,circle,fill,inner sep=0mm, minimum size=0.1mm, BrickRed!20] (w1-u-0_1) at ($(w1_1)!.5!(u_1)!0.05136741973840665cm!90:(u_1)$) {};
		\path[draw,very thin, BrickRed!20] (w1_1) to (w1-u-0_1) to (u_1);
		\node[draw,circle,fill,inner sep=0mm, minimum size=0.1mm, BrickRed!20] (w1-u-1_1) at ($(w1_1)!.5!(u_1)!0.06516052318668251cm!90:(u_1)$) {};
		\path[draw,very thin, BrickRed!20] (w1_1) to (w1-u-1_1) to (u_1);
		\node[draw,circle,fill,inner sep=0mm, minimum size=0.1mm, BrickRed!20] (w1-u-2_1) at ($(w1_1)!.5!(u_1)!0.07895362663495838cm!90:(u_1)$) {};
		\path[draw,very thin, BrickRed!20] (w1_1) to (w1-u-2_1) to (u_1);
		\node[draw,circle,fill,inner sep=0mm, minimum size=0.1mm, BrickRed!20] (w1-u-3_1) at ($(w1_1)!.5!(u_1)!0.09274673008323422cm!90:(u_1)$) {};
		\path[draw,very thin, BrickRed!20] (w1_1) to (w1-u-3_1) to (u_1);
		\node[draw,circle,fill,inner sep=0mm, minimum size=0.1mm, BrickRed!20] (w1-u-4_1) at ($(w1_1)!.5!(u_1)!0.1065398335315101cm!90:(u_1)$) {};
		\path[draw,very thin, BrickRed!20] (w1_1) to (w1-u-4_1) to (u_1);
		\node[draw,circle,fill,inner sep=0mm, minimum size=0.1mm, BrickRed!20] (w1-u-5_1) at ($(w1_1)!.5!(u_1)!0.12033293697978598cm!90:(u_1)$) {};
		\path[draw,very thin, BrickRed!20] (w1_1) to (w1-u-5_1) to (u_1);
		\node[draw,circle,fill,inner sep=0mm, minimum size=0.1mm, BrickRed!20] (w1-u-6_1) at ($(w1_1)!.5!(u_1)!0.1341260404280618cm!90:(u_1)$) {};
		\path[draw,very thin, BrickRed!20] (w1_1) to (w1-u-6_1) to (u_1);
		\node[draw,circle,fill,inner sep=0mm, minimum size=0.1mm, BrickRed!20] (w1-u-7_1) at ($(w1_1)!.5!(u_1)!0.14791914387633767cm!90:(u_1)$) {};
		\path[draw,very thin, BrickRed!20] (w1_1) to (w1-u-7_1) to (u_1);
		\path[draw, semithick, BrickRed!20, semithick, BrickRed!60, bend right=0.0] (w1_1)  to (u_1);
		\node[draw,circle,fill,inner sep=0mm, minimum size=0.1mm, BrickRed!20] (w1-w3-0_1) at ($(w1_1)!.5!(w3_1)!-0.05517241379310345cm!90:(w3_1)$) {};
		\path[draw,very thin, BrickRed!20] (w1_1) to (w1-w3-0_1) to (w3_1);
		\node[draw,circle,fill,inner sep=0mm, minimum size=0.1mm, BrickRed!20] (w1-w3-1_1) at ($(w1_1)!.5!(w3_1)!-0.04137931034482758cm!90:(w3_1)$) {};
		\path[draw,very thin, BrickRed!20] (w1_1) to (w1-w3-1_1) to (w3_1);
		\node[draw,circle,fill,inner sep=0mm, minimum size=0.1mm, BrickRed!20] (w1-w3-2_1) at ($(w1_1)!.5!(w3_1)!-0.027586206896551724cm!90:(w3_1)$) {};
		\path[draw,very thin, BrickRed!20] (w1_1) to (w1-w3-2_1) to (w3_1);
		\node[draw,circle,fill,inner sep=0mm, minimum size=0.1mm, BrickRed!20] (w1-w3-3_1) at ($(w1_1)!.5!(w3_1)!-0.013793103448275864cm!90:(w3_1)$) {};
		\path[draw,very thin, BrickRed!20] (w1_1) to (w1-w3-3_1) to (w3_1);
		\node[draw,circle,fill,inner sep=0mm, minimum size=0.1mm, BrickRed!20] (w1-w3-4_1) at ($(w1_1)!.5!(w3_1)!0.0cm!90:(w3_1)$) {};
		\path[draw,very thin, BrickRed!20] (w1_1) to (w1-w3-4_1) to (w3_1);
		\node[draw,circle,fill,inner sep=0mm, minimum size=0.1mm, BrickRed!20] (w1-w3-5_1) at ($(w1_1)!.5!(w3_1)!0.013793103448275864cm!90:(w3_1)$) {};
		\path[draw,very thin, BrickRed!20] (w1_1) to (w1-w3-5_1) to (w3_1);
		\node[draw,circle,fill,inner sep=0mm, minimum size=0.1mm, BrickRed!20] (w1-w3-6_1) at ($(w1_1)!.5!(w3_1)!0.027586206896551717cm!90:(w3_1)$) {};
		\path[draw,very thin, BrickRed!20] (w1_1) to (w1-w3-6_1) to (w3_1);
		\node[draw,circle,fill,inner sep=0mm, minimum size=0.1mm, BrickRed!20] (w1-w3-7_1) at ($(w1_1)!.5!(w3_1)!0.04137931034482759cm!90:(w3_1)$) {};
		\path[draw,very thin, BrickRed!20] (w1_1) to (w1-w3-7_1) to (w3_1);
		\node[draw,circle,fill,inner sep=0mm, minimum size=0.1mm, BrickRed!20] (w3-w2-0_1) at ($(w3_1)!.5!(w2_1)!-0.05517241379310345cm!90:(w2_1)$) {};
		\path[draw,very thin, BrickRed!20] (w3_1) to (w3-w2-0_1) to (w2_1);
		\node[draw,circle,fill,inner sep=0mm, minimum size=0.1mm, BrickRed!20] (w3-w2-1_1) at ($(w3_1)!.5!(w2_1)!-0.04137931034482758cm!90:(w2_1)$) {};
		\path[draw,very thin, BrickRed!20] (w3_1) to (w3-w2-1_1) to (w2_1);
		\node[draw,circle,fill,inner sep=0mm, minimum size=0.1mm, BrickRed!20] (w3-w2-2_1) at ($(w3_1)!.5!(w2_1)!-0.027586206896551724cm!90:(w2_1)$) {};
		\path[draw,very thin, BrickRed!20] (w3_1) to (w3-w2-2_1) to (w2_1);
		\node[draw,circle,fill,inner sep=0mm, minimum size=0.1mm, BrickRed!20] (w3-w2-3_1) at ($(w3_1)!.5!(w2_1)!-0.013793103448275864cm!90:(w2_1)$) {};
		\path[draw,very thin, BrickRed!20] (w3_1) to (w3-w2-3_1) to (w2_1);
		\node[draw,circle,fill,inner sep=0mm, minimum size=0.1mm, BrickRed!20] (w3-w2-4_1) at ($(w3_1)!.5!(w2_1)!0.0cm!90:(w2_1)$) {};
		\path[draw,very thin, BrickRed!20] (w3_1) to (w3-w2-4_1) to (w2_1);
		\node[draw,circle,fill,inner sep=0mm, minimum size=0.1mm, BrickRed!20] (w3-w2-5_1) at ($(w3_1)!.5!(w2_1)!0.013793103448275864cm!90:(w2_1)$) {};
		\path[draw,very thin, BrickRed!20] (w3_1) to (w3-w2-5_1) to (w2_1);
		\node[draw,circle,fill,inner sep=0mm, minimum size=0.1mm, BrickRed!20] (w3-w2-6_1) at ($(w3_1)!.5!(w2_1)!0.027586206896551717cm!90:(w2_1)$) {};
		\path[draw,very thin, BrickRed!20] (w3_1) to (w3-w2-6_1) to (w2_1);
		\node[draw,circle,fill,inner sep=0mm, minimum size=0.1mm, BrickRed!20] (w3-w2-7_1) at ($(w3_1)!.5!(w2_1)!0.04137931034482759cm!90:(w2_1)$) {};
		\path[draw,very thin, BrickRed!20] (w3_1) to (w3-w2-7_1) to (w2_1);
		\node[draw,circle,fill,inner sep=0mm, minimum size=0.1mm, BrickRed!20] (v-w2-0_1) at ($(v_1)!.5!(w2_1)!0.0cm!90:(w2_1)$) {};
		\path[draw,very thin, BrickRed!20] (v_1) to (v-w2-0_1) to (w2_1);
		\node[draw,circle,fill,inner sep=0mm, minimum size=0.1mm, BrickRed!20] (v-w2-1_1) at ($(v_1)!.5!(w2_1)!0.013793103448275862cm!90:(w2_1)$) {};
		\path[draw,very thin, BrickRed!20] (v_1) to (v-w2-1_1) to (w2_1);
		\node[draw,circle,fill,inner sep=0mm, minimum size=0.1mm, BrickRed!20] (v-w2-2_1) at ($(v_1)!.5!(w2_1)!0.027586206896551724cm!90:(w2_1)$) {};
		\path[draw,very thin, BrickRed!20] (v_1) to (v-w2-2_1) to (w2_1);
		\node[draw,circle,fill,inner sep=0mm, minimum size=0.1mm, BrickRed!20] (v-w2-3_1) at ($(v_1)!.5!(w2_1)!0.04137931034482758cm!90:(w2_1)$) {};
		\path[draw,very thin, BrickRed!20] (v_1) to (v-w2-3_1) to (w2_1);
		\node[draw,circle,fill,inner sep=0mm, minimum size=0.1mm, BrickRed!20] (v-w2-4_1) at ($(v_1)!.5!(w2_1)!0.05517241379310345cm!90:(w2_1)$) {};
		\path[draw,very thin, BrickRed!20] (v_1) to (v-w2-4_1) to (w2_1);
		\node[draw,circle,fill,inner sep=0mm, minimum size=0.1mm, BrickRed!20] (v-w2-5_1) at ($(v_1)!.5!(w2_1)!0.06896551724137931cm!90:(w2_1)$) {};
		\path[draw,very thin, BrickRed!20] (v_1) to (v-w2-5_1) to (w2_1);
		\node[draw,circle,fill,inner sep=0mm, minimum size=0.1mm, BrickRed!20] (v-w2-6_1) at ($(v_1)!.5!(w2_1)!0.08275862068965516cm!90:(w2_1)$) {};
		\path[draw,very thin, BrickRed!20] (v_1) to (v-w2-6_1) to (w2_1);
		\node[draw,circle,fill,inner sep=0mm, minimum size=0.1mm, BrickRed!20] (v-w2-7_1) at ($(v_1)!.5!(w2_1)!0.09655172413793103cm!90:(w2_1)$) {};
		\path[draw,very thin, BrickRed!20] (v_1) to (v-w2-7_1) to (w2_1);
		\node[draw,circle,fill,inner sep=0mm, minimum size=0.1mm, BrickRed!20] (w2-w1-0_1) at ($(w2_1)!.5!(w1_1)!0.0cm!90:(w1_1)$) {};
		\path[draw,very thin, BrickRed!20] (w2_1) to (w2-w1-0_1) to (w1_1);
		\node[draw,circle,fill,inner sep=0mm, minimum size=0.1mm, BrickRed!20] (w2-w1-1_1) at ($(w2_1)!.5!(w1_1)!0.013793103448275862cm!90:(w1_1)$) {};
		\path[draw,very thin, BrickRed!20] (w2_1) to (w2-w1-1_1) to (w1_1);
		\node[draw,circle,fill,inner sep=0mm, minimum size=0.1mm, BrickRed!20] (w2-w1-2_1) at ($(w2_1)!.5!(w1_1)!0.027586206896551724cm!90:(w1_1)$) {};
		\path[draw,very thin, BrickRed!20] (w2_1) to (w2-w1-2_1) to (w1_1);
		\node[draw,circle,fill,inner sep=0mm, minimum size=0.1mm, BrickRed!20] (w2-w1-3_1) at ($(w2_1)!.5!(w1_1)!0.04137931034482758cm!90:(w1_1)$) {};
		\path[draw,very thin, BrickRed!20] (w2_1) to (w2-w1-3_1) to (w1_1);
		\node[draw,circle,fill,inner sep=0mm, minimum size=0.1mm, BrickRed!20] (w2-w1-4_1) at ($(w2_1)!.5!(w1_1)!0.05517241379310345cm!90:(w1_1)$) {};
		\path[draw,very thin, BrickRed!20] (w2_1) to (w2-w1-4_1) to (w1_1);
		\node[draw,circle,fill,inner sep=0mm, minimum size=0.1mm, BrickRed!20] (w2-w1-5_1) at ($(w2_1)!.5!(w1_1)!0.06896551724137931cm!90:(w1_1)$) {};
		\path[draw,very thin, BrickRed!20] (w2_1) to (w2-w1-5_1) to (w1_1);
		\node[draw,circle,fill,inner sep=0mm, minimum size=0.1mm, BrickRed!20] (w2-w1-6_1) at ($(w2_1)!.5!(w1_1)!0.08275862068965516cm!90:(w1_1)$) {};
		\path[draw,very thin, BrickRed!20] (w2_1) to (w2-w1-6_1) to (w1_1);
		\node[draw,circle,fill,inner sep=0mm, minimum size=0.1mm, BrickRed!20] (w2-w1-7_1) at ($(w2_1)!.5!(w1_1)!0.09655172413793103cm!90:(w1_1)$) {};
		\path[draw,very thin, BrickRed!20] (w2_1) to (w2-w1-7_1) to (w1_1);
	\path[draw,semithick, BrickRed!60, bend right=0.0] (w2_1) to (u_1) {};
	\path[draw,semithick, BrickRed!60, bend right=0.0] (v_1) to (w1_1) {};
		\node[wheelS,importantV] at (1.5, 0.0) {};
		\node[wheelS,importantV] at (-2.4999999999999996, 4.0) {};
		\node[draw,wheelS,circle,inner sep=0mm, minimum size=1.7mm, fill] (c-1) at (-2.5, 0) {};
		\node[splitterS,smallV,circle] (p-1-0) at ($(c-1)!5.2cm!45.0:(0,0)$) {};
		\path[splitterS,draw,thick] (c-1) to (p-1-0);
	\end{tikzpicture}
}

%% file: references.bib
@article{BienstockD93,
  author       = {Daniel Bienstock and Nathaniel Dean},
  title        = {Bounds for rectilinear crossing numbers},
  journal      = {J. Graph Theory},
  volume       = {17},
  number       = {3},
  pages        = {333--348},
  year         = {1993},
  doi          = {10.1002/JGT.3190170308},
}

@article{Fary,
  author       = {F{á}ry, Istv{á}n},
  title        = {On straight-line representation of planar graphs},
  journal      = {Acta Sci. Math. (Szeged)},
  volume       = {11},
  pages        = {229–-233},
  year         = {1948},
}

@article{Wagner,
  author       = {Wagner, Klaus},
  title        = {{Bemerkungen zum Vierfarbenproblem}},
  journal      = {Jahresbericht der Deutschen Mathematiker-Vereinigung},
  volume       = {46},
  pages        = {26--32},
  year         = {1936},
}

@book{HT2020,
  editor       = {Seok{-}Hee Hong and Takeshi Tokuyama},
  title        = {Beyond Planar Graphs, Communications of {NII} Shonan Meetings},
  publisher    = {Springer},
  year         = {2020},
  doi          = {10.1007/978-981-15-6533-5},
}

@inproceedings{HammH22,
  author       = {Thekla Hamm and Petr Hlinen{\'{y}}},
  title        = {Parameterised Partially-Predrawn Crossing Number},
  booktitle    = {38th International Symposium on Computational Geometry, {SoCG} 2022},
  series       = {LIPIcs},
  volume       = {224},
  pages        = {46:1--46:15},
  year         = {2022},
  doi          = {10.4230/LIPICS.SOCG.2022.46},
}

@inproceedings{MunchR24,
  author       = {Miriam M{\"{u}}nch and Ignaz Rutter},
  title        = {Parameterized Algorithms for Beyond-Planar Crossing Numbers},
  booktitle    = {32nd International Symposium on Graph Drawing and Network Visualization, {GD} 2024},
  series       = {LIPIcs},
  volume       = {320},
  pages        = {25:1--25:16},
  year         = {2024},
  doi          = {10.4230/LIPICS.GD.2024.25},
}

@article{Brandenburg19,
  author       = {Franz J. Brandenburg},
  title        = {Characterizing and Recognizing 4-Map Graphs},
  journal      = {Algorithmica},
  volume       = {81},
  number       = {5},
  pages        = {1818--1843},
  year         = {2019},
  doi          = {10.1007/S00453-018-0510-X},
}

@inproceedings{HongELP12,
  author       = {Seok{-}Hee Hong and Peter Eades and Giuseppe Liotta and Sheung{-}Hung Poon},
  title        = {F{\'{a}}ry's Theorem for 1-Planar Graphs},
  booktitle    = {Computing and Combinatorics - 18th Annual International Conference, {COCOON} 2012},
  series       = {Lecture Notes in Computer Science},
  volume       = {7434},
  pages        = {335--346},
  year         = {2012},
  doi          = {10.1007/978-3-642-32241-9\_29},
}

@article{Zehavi22,
  author       = {Meirav Zehavi},
  title        = {Parameterized analysis and crossing minimization problems},
  journal      = {Comput. Sci. Rev.},
  volume       = {45},
  pages        = {100490},
  year         = {2022},
  doi          = {10.1016/J.COSREV.2022.100490},
}

@article{KobourovLM17,
  author       = {Stephen G. Kobourov and Giuseppe Liotta and Fabrizio Montecchiani},
  title        = {An annotated bibliography on 1-planarity},
  journal      = {Comput. Sci. Rev.},
  volume       = {25},
  pages        = {49--67},
  year         = {2017},
  doi          = {10.1016/J.COSREV.2017.06.002},
}

@article{DidimoLM19,
  author       = {Walter Didimo and Giuseppe Liotta and Fabrizio Montecchiani},
  title        = {A Survey on Graph Drawing Beyond Planarity},
  journal      = {{ACM} Comput. Surv.},
  volume       = {52},
  number       = {1},
  pages        = {4:1--4:37},
  year         = {2019},
  doi          = {10.1145/3301281},
}

@article{KorzhikM13,
  author       = {Vladimir P. Korzhik and Bojan Mohar},
  title        = {Minimal Obstructions for 1-Immersions and Hardness of 1-Planarity Testing},
  journal      = {J. Graph Theory},
  volume       = {72},
  number       = {1},
  pages        = {30--71},
  year         = {2013},
  doi          = {10.1002/JGT.21630},
}

@article{CabelloM13,
  author       = {Sergio Cabello and Bojan Mohar},
  title        = {Adding One Edge to Planar Graphs Makes Crossing Number and 1-Planarity Hard},
  journal      = {{SIAM} J. Comput.},
  volume       = {42},
  number       = {5},
  pages        = {1803--1829},
  year         = {2013},
  doi          = {10.1137/120872310},
}

@article{BannisterCE18,
  author       = {Michael J. Bannister and Sergio Cabello and David Eppstein},
  title        = {Parameterized Complexity of 1-Planarity},
  journal      = {J. Graph Algorithms Appl.},
  volume       = {22},
  number       = {1},
  pages        = {23--49},
  year         = {2018},
  doi          = {10.7155/JGAA.00457},
}

@inproceedings{BiedlB024,
  author       = {Therese Biedl and Prosenjit Bose and Karthik Murali},
  title        = {A Parameterized Algorithm for Vertex and Edge Connectivity of Embedded Graphs},
  booktitle    = {32nd Annual European Symposium on Algorithms, {ESA} 2024},
  series       = {LIPIcs},
  volume       = {308},
  pages        = {24:1--24:15},
  year         = {2024},
  doi          = {10.4230/LIPICS.ESA.2024.24},
}

@article{BoseCMM25,
  author       = {Prosenjit Bose and Jean{-}Lou De Carufel and
                  Anil Maheshwari and Karthik Murali},
  title        = {On 1-planar graphs with bounded cop-number},
  journal      = {Theor. Comput. Sci.},
  volume       = {1037},
  pages        = {115160},
  year         = {2025},
  doi          = {10.1016/J.TCS.2025.115160},
}

@InProceedings{maininspoBiedlMurali23,
  author =	{Biedl, Therese and Murali, Karthik},
  title =	{On Computing the Vertex Connectivity of 1-Plane Graphs},
  booktitle =	{50th International Colloquium on Automata, Languages, and Programming (ICALP 2023)},
  pages =	{23:1--23:16},
  series =	{LIPIcs},
  year =	{2023},
  volume =	{261},
  doi =		{10.4230/LIPIcs.ICALP.2023.23},
}

@article{Brandenburg15,
  author       = {Franz J. Brandenburg},
  title        = {On 4-Map Graphs and 1-Planar Graphs and their Recognition Problem},
  journal      = {CoRR},
  volume       = {abs/1509.03447},
  year         = {2015},
  url          = {http://arxiv.org/abs/1509.03447},
  eprinttype    = {arXiv},
  eprint       = {1509.03447},
}

@article{Thomassen88a,
  author       = {Carsten Thomassen},
  title        = {Rectilinear drawings of graphs},
  journal      = {J. Graph Theory},
  volume       = {12},
  number       = {3},
  pages        = {335--341},
  year         = {1988},
  doi          = {10.1002/JGT.3190120306},
}

@article{GrigorievB07,
  author       = {Alexander Grigoriev and Hans L. Bodlaender},
  title        = {Algorithms for Graphs Embeddable with Few Crossings per Edge},
  journal      = {Algorithmica},
  volume       = {49},
  number       = {1},
  pages        = {1--11},
  year         = {2007},
  doi          = {10.1007/S00453-007-0010-X},
}

@book{Diestel-GT-2025,
  author       = {Reinhard Diestel},
  title        = {Graph Theory, 6th Edition},
  series       = {Graduate texts in mathematics},
  volume       = {173},
  publisher    = {Springer},
  year         = {2025},
}

@article{Kuratowski30,
  title = {Sur le probleme des courbes gauches en topologie},
  author = {Kuratowski, Kazimierz},
  journal = {Fundamenta mathematicae},
  volume = {15},
  number = {1},
  pages = {271--283},
  year = {1930},
  _bib2doi_finished = {true},
}

@inproceedings{BattistaT89,
  author       = {Giuseppe Di Battista and Roberto Tamassia},
  title        = {Incremental Planarity Testing (Extended Abstract)},
  booktitle    = {30th Annual Symposium on Foundations of Computer Science},
  pages        = {436--441},
  publisher    = {{IEEE} Computer Society},
  year         = {1989},
  doi          = {10.1109/SFCS.1989.63515},
}

@inproceedings{GutwengerM00,
  author       = {Carsten Gutwenger and Petra Mutzel},
  title        = {A Linear Time Implementation of {SPQR}-Trees},
  booktitle    = {Graph Drawing, 8th International Symposium, {GD} 2000},
  series       = {Lecture Notes in Computer Science},
  volume       = {1984},
  pages        = {77--90},
  year         = {2000},
  doi          = {10.1007/3-540-44541-2\_8},
}

@article{Grohe04,
  author       = {Martin Grohe},
  title        = {Computing crossing numbers in quadratic time},
  journal      = {J. Comput. Syst. Sci.},
  volume       = {68},
  number       = {2},
  pages        = {285--302},
  year         = {2004},
  doi          = {10.1016/J.JCSS.2003.07.008},
}

@book{CyganFKLMPPS15,
  author       = {Marek Cygan and Fedor V. Fomin and Lukasz Kowalik and
                  Daniel Lokshtanov and D{\'{a}}niel Marx and Marcin Pilipczuk and
                  Michal Pilipczuk and Saket Saurabh},
  title        = {Parameterized Algorithms},
  publisher    = {Springer},
  year         = {2015},
  doi          = {10.1007/978-3-319-21275-3},
}

@book{DowneyF13,
  author       = {Rodney G. Downey and Michael R. Fellows},
  title        = {Fundamentals of Parameterized Complexity},
  series       = {Texts in Computer Science},
  publisher    = {Springer},
  year         = {2013},
  doi          = {10.1007/978-1-4471-5559-1},
}

@article {MR1042649,
    AUTHOR = {Courcelle, Bruno},
     TITLE = {The monadic second-order logic of graphs. {I}. {R}ecognizable
              sets of finite graphs},
   JOURNAL = {Inform. and Comput.},
  FJOURNAL = {Information and Computation},
    VOLUME = {85},
      YEAR = {1990},
    NUMBER = {1},
     PAGES = {12--75},
      ISSN = {0890-5401,1090-2651},
   MRCLASS = {05C75 (03C13 05C05 68Q50 68Q70)},
  MRNUMBER = {1042649},
MRREVIEWER = {Damian\ Niwi\'nski},
       DOI = {10.1016/0890-5401(90)90043-H},
       URL = {https://doi.org/10.1016/0890-5401(90)90043-H},
}

@article {MR1417901,
    AUTHOR = {Bodlaender, Hans L.},
     TITLE = {A linear-time algorithm for finding tree-decompositions of
              small treewidth},
   JOURNAL = {SIAM J. Comput.},
  FJOURNAL = {SIAM Journal on Computing},
    VOLUME = {25},
      YEAR = {1996},
    NUMBER = {6},
     PAGES = {1305--1317},
  MRNUMBER = {1417901},
       DOI = {10.1137/S0097539793251219},
       URL = {https://doi.org/10.1137/S0097539793251219},
}

@article{BodlaenderDDFLP16,
  author       = {Hans L. Bodlaender and P{\aa}l Gr{\o}n{\aa}s Drange and
                  Markus S. Dregi and Fedor V. Fomin and
                  Daniel Lokshtanov and Micha\l{} Pilipczuk},
  title        = {A \(c^k n\) 5-Approximation Algorithm for Treewidth},
  journal      = {{SIAM} J. Comput.},
  volume       = {45},
  number       = {2},
  pages        = {317--378},
  year         = {2016},
  url          = {https://doi.org/10.1137/130947374},
  doi          = {10.1137/130947374},
}

@article{PachT97,
  author       = {J{\'{a}}nos Pach and G{\'{e}}za T{\'{o}}th},
  title        = {Graphs Drawn with Few Crossings per Edge},
  journal      = {Comb.},
  volume       = {17},
  number       = {3},
  pages        = {427--439},
  year         = {1997},
  url          = {https://doi.org/10.1007/BF01215922},
  doi          = {10.1007/BF01215922},
}

@article{Schumacher86,
author 		= {Schumacher, Von H.},
title 		= {Zur {S}truktur 1-planarer {G}raphen},
journal 	= {Mathematische Nachrichten},
volume 		= {125},
number 		= {1},
pages 		= {291--300},
doi 		= {10.1002/mana.19861250122},
url 		= {https://doi.org/10.1002/mana.19861250122},
year 		= {1986}
}
